\documentclass[prodmode,acmjacm]{acmupsmall}

\usepackage[ruled]{algorithm2e}
\renewcommand{\algorithmcfname}{ALGORITHM}
\SetAlFnt{\small}
\SetAlCapFnt{\small}
\SetAlCapNameFnt{\small}
\SetAlCapHSkip{0pt}
\IncMargin{-\parindent}

\acmVolume{0}
\acmNumber{0}
\acmArticle{00}
\acmYear{2011}
\acmMonth{0}

\usepackage{latexsym}
\usepackage{amsmath}
\usepackage{amssymb}
\usepackage{vmargin}

\usepackage{tikz}

\newlength{\shor}
\def\tu#1{\langle #1\rangle}

\newcommand{\up}{\uparrow}
\newcommand{\down}{\downarrow}
\newcommand{\upts}{\uparrow}
\newcommand{\downts}{\downarrow}

\def\setto{\mathrel{:=}}
\def\First#1{\ensuremath{\mathop{\mathrm{First}}(#1)}}
\def\Rest#1{\ensuremath{\mathop{\mathrm{Rest}}(#1)}}
\def\Put#1#2{\ensuremath{\mathop{\mathrm{Put}}(#1,#2)}}
\def\IsNotEmpty#1{\ensuremath{\mathop{\mathrm{IsNotEmpty}}(#1)}}
\def\NewList{\ensuremath{\mathop{\mathrm{NewList}}()}}

\def\crd{{\mathrel{\curlyvee}}}
\def\cru{{\mathrel{\curlywedge}}}
\def\Fu#1{\lceil #1\rceil}
\def\Cr#1{\lfloor #1\rfloor}

\def\ts{{\ast}}

\def\tsy{{\ts_Y}}
\def\tss{\bullet}

\newcommand{\Pcl}[2][T]{\mathop{\textit{cl}_{#1^{\ast}}\kern-.4pt(#2)}}

\newcommand{\Tcl}[2][T]{\mathop{\textit{cl}_{#1}\kern-.4pt(#2)}}

\def\deg#1#2{{}^{#1\!}/#2}
\def\impl#1#2{#1 \!\Rightarrow\! #2}

\begin{document}

\markboth{R. Belohlavek, V. Vychodil}{Attribute Dependencies for Data with Grades}

\title{Attribute Dependencies for Data with Grades}
\author{RADIM BELOHLAVEK
\affil{Palacky University, Olomouc}
VILEM VYCHODIL
\affil{Palacky University, Olomouc}
}

\begin{abstract}
This paper examines attribute dependencies in data that involve grades, such as 
a grade to which an object is red or a grade to which two objects are similar.
We thus extend the classical agenda by allowing  graded, or ``fuzzy'', attributes
instead of Boolean, yes-or-no attributes in case of attribute implications, and allowing
approximate match based on degrees of similarity instead of exact match based on equality
in case of functional dependencies.
In a sense, we move from bivalence, inherently present in the now-available theories of dependencies, 
to a more flexible setting that involves grades. Such a shift has far-reaching consequences.
We argue that a reasonable theory of dependencies may be developed by making use of
mathematical fuzzy logic, a recently developed many-valued logic.
Namely, 
 the theory of dependencies is then based on a solid logic calculus  the same way the classical
dependencies are based on classical logic.
For instance, rather than handling degrees of similarity in an ad hoc manner, we consistently treat
them as truth values, the same way as \emph{true} (match)  and \emph{false} (mismatch)  are treated in classical theories.
In addition, several notions intuitively embraced in the presence of grades, such as a degree
of validity of a particular dependence or a degree of entailment,  naturally emerge and receive a conceptually clean treatment
in the presented approach.
In the paper, we discuss motivations, provide basic notions of syntax and semantics, and develop 
basic results which include entailment of dependencies, associated closure structures, 
a logic of dependencies with two versions of completeness theorem, results and algorithms
regarding complete non-redundant sets of dependencies, relationship to and a possible
reductionist interface to classical dependencies,  and relationship to functional dependencies
over domains with similarity. We also outline future research topics.
\end{abstract}

\category{F.4.1}{Mathematical Logic}{Model theory}
\category{H.2.8}{Database Applications}{Data mining}
\category{I.2.3}{Deduction and Theorem Proving}{Uncertainty, ``fuzzy,'' and probabilistic reasoning}
\category{I.2.4}{Knowledge Representation Formalisms and Methods}{Relation systems}

\terms{Theory}

\keywords{attribute dependence, grade, similarity, logic, 
redundancy, functional dependence}

\acmformat{
Belohlavek, R., Vychodil, V. 2012. Attribute Dependencies for Data with Grades.}

\begin{bottomstuff}
This work is supported by the Czech Science Foundation,
under grant No. P202/10/0262.
%
%
Authors' addresses: 
R. Belohlavek, V. Vychodil, Department of Computer Science,
Palacky University, Olomouc.
\end{bottomstuff}

\maketitle

\section{Introduction}
\label{sec:i}

Attribute dependencies are fundamental for understanding and processing data. In the past, dependencies describing
various types of attribute relationships have been studied, particularly in relational databases and data analysis/mining.
Arguably, the most important dependencies are those of the form
\begin{eqnarray}\label{eqn:AIform}
 A\Rightarrow B
\end{eqnarray}
 where $A$ and $B$ are sets of attributes.
They are interpreted in two basic ways---in binary datasets (tables with yes-or-no attributes)
describing which objects have which attributes and in relations (tables with general attributes) 
describing the values of objects for the attributes. 
In binary datasets, $A\Rightarrow B$ is considered valid if 
\begin{eqnarray}\label{eqn:AIval}
 &&\mbox{every object (table row) having all attributes from $A$ has all attributes from $B$,}
\end{eqnarray}
or, more generally, if a certain percentage (called confidence) of the objects having $A$ also have $B$ and another 
percentage (called support) of objects have all the attributes from $A\cup B$.
Such dependencies are used in data analysis and are known as attribute implications \cite{CaRo:CDA,GaWi:FCA,GuDu:Fmiirtdb},
see also \cite{DeCa:Ddbtbsf,Fag:Fdrdpl},
or association rules when the support and confidence are considered 
\cite{AgImSw:Marbsild,HaHoRa:Gmmdm,HaHa:MHF,HiGuNa:Aarl,TaStKu:IDM}. 
In relations, $A\Rightarrow B$ is considered valid if 
\begin{eqnarray}\label{eqn:FDval}
  \nonumber
 &&\mbox{every two tuples (table rows) with the same values on attributes from $A$}\\
 &&\mbox{have the same values on attributes from $B$.}
\end{eqnarray}
Such dependencies are called functional dependencies and are fundamental to relational databases 
\cite{Arm:Dsdbr,Cod:Rmdlsdb,Mai:TRD}.

A common feature of the two interpretations is a bivalent character of the conditions involved in (\ref{eqn:AIval}) and (\ref{eqn:FDval}). The bivalence results from the nature of the data.
Namely, a given object either does or does not have a given attribute; two given tuples either do or do 
not have the same value for a given attribute. It turns out that it is becoming increasingly important 
for data models to  account for fuzziness \cite{Fag:Cfims,Fag:Cfi},
which is inherently present in human cognition and plays a fundamental role in how people
communicate knowledge \cite{Zad:Fl,Zad:Flnnsc,Zad:Itnfl}.
Two points in case are fuzzy (or graded) attributes, such as ``green'' or ``performing well'', and similarity relations.
In these and other cases, fuzziness is conveniently represented by grades (degrees, scores) which are usually
numbers ranging between $0$ and $1$. 
Thus, an object $x$ may be assigned a grade  to which $x$ is green---the higher the grade, the more
green $x$ is.
Likewise, two objects $x$ and $y$ may be assigned a degree to which $x$ and $y$ are similar.
A scale of grades bounded by $0$ and $1$ thus naturally replaces the two-element set of truth values
of classical logic with $0$ representing falsity (``attribute does not apply'', ``values do not match'')
and $1$ representing truth (``attribute applies'', ``values match'').
For data with grades, the ordinary dependencies have limited applicability. 
Namely, rather than knowing that (full) presence of some attributes  implies 
(full) presence of some other attributes, one is naturally interested in rules that take the grades into account.
Such rules are the main subject of the present paper.

In particular, we consider rules saying that
%
presence of attributes $y_i$ with grades at least $a_i$ implies (or implies partially)
presence of attributes $z_i$ with grades at least $b_i$.
Therefore, from rules of the form
\begin{eqnarray}\label{eqn:YZ}
   \{ y_1,\dots,y_p\} \!\Rightarrow\!
  \{z_1,\dots,z_q\}
\end{eqnarray}
we come to rules of the form
\begin{eqnarray}\label{eqn:YZgrades}
   \{{}^{a_1\!\!}/y_1,\dots,{}^{a_p\!\!}/y_p\} \!\Rightarrow\!
  \{{}^{b_1\!\!}/z_1,\dots,{}^{b_q\!\!}/z_q\},
\end{eqnarray}
such as 
\begin{eqnarray}\label{eqn:uflahc}
   \{{}^{0.5\!\!}/\mathrm{unhealthy\ food},{}^{0.9\!\!}/\mathrm{little\ activity}\} \!\Rightarrow\!
  \{{}^{0.7\!\!}/\mathrm{high\ cholesterol}\}.
\end{eqnarray}
From a functional dependence point of view, such rules may be interpreted in tables whose
domains are equipped with similarity relations assiging similarity grades to pairs of elements
of the domains. In such tables, the rules specify that two tuples with similar values on attributes 
$y_1,\dots,y_p$ have similar values on $z_1,\dots,z_q$. In particular, rule (\ref{eqn:YZgrades}) says
that similarity to degrees $a_i$ or higher on attributes $y_i$
implies similarity to degrees $b_i$ or higher on attributes $z_i$, generalizing thus ordinary functional
dependencies which say that a match of two tuples on attributes $y_1,\dots,y_p$
implies a match on $z_1,\dots,z_q$.

Using grades to represent fuzziness is the fundamental idea of fuzzy logic \cite{Zad:Fs}. We use fuzzy
logic as a formal framework for our approach. Fuzzy logic enables us to manipulate the grades
by means of the truth functions of logic connectives. In the past, various models of processing data
with grades using fuzzy logic connectives, notably ``fuzzy conjunction'',  have been studied 
in a more or less ad hoc way. 
In this perspective,
one aspect of our work is that we consistently use the so-called mathematical fuzzy 
logic \cite{Ciea:HMFL,Got:TMVL,Got:Mfl,Haj:MFL,Haj:Wimfl} as a formal framework. Mathematical fuzzy logic
is a recently developed branch of logic that provides us with general principles and notions such as
theory, model, or entailment, and enables us to process data with grades in a clean way. 

Our reliance on mathematical fuzzy logic is similar to the reliance of the ordinary dependencies on
classical logic. In case of grades, however, the logic framework is more explicit. Namely, while
in the ordinary case the assertions like (\ref{eqn:AIval}) or the notion of entailment have a clear meaning
and one rarely needs to resort to the formal agenda of classical logic, in case of grades, 
the meaning needs to be supplied by an explicit resorting to fuzzy logic principles. 
Due to a consistent use of fuzzy logic, the verbal description of validity conditions and
manipulation regarding the dependencies remains the same as in the ordinary case, 
retaining thus a clear meaning.
For instance, the validity of  rule (\ref{eqn:YZgrades}) in data with grades may still
be verbally described by (\ref{eqn:AIval}), the grades being ``hidden in the interpretation''.
A natural consequence of working with grades is that key logic notions such as validity
 or entailment become graded. That is, we speak of a degree to which a given rule is valid or
a degree to which a rule follows from other rules leaving validity or entailment to degree $1$ 
(full validity or full entailment) important particular cases.

While the reliance on mathematical fuzzy logic provides us with reasonable guiding principles, 
the resulting notions and problems tend, naturally, to be more involved both conceptually
as well as technically due to the presence of intermediary grades and, in addition,
due to the fact that we develop the theory for a general scale $\mathbf{L}$ of truth degrees
with $\mathbf{L}$ acting as a parameter. 
The conceptual aspect regarding the extension from the ordinary, bivalent
framework to a framework involving grades may, using a loose analogy, 
be compared to an extension from a deterministic to a probabilistic framework.
As regards the technical  aspect, a point in case for illustration
is the fact that, as a rule, the ordinary proofs by cases, corresponding
to \emph{false} and \emph{true}, no longer work and need to be replaced by different
schemata which are based on algebraic maniputation of the grades.
In this perspective, the paper illustrates both aspects, the conceptual and the technical,
by numerous cases. 

The paper is organized as follows.
In Section \ref{sec:p}, we present preliminaries on scales of truth degrees
and operations on them and the basic principles of fuzzy logic.
Section \ref{sec:gai} presents the basic notions regarding graded
attribute implications, their validity,  theories, models, entailment, 
and related
closure and other structures.
In Section \ref{sec:fal}, we present a system for reasoning with graded attribute
implications that is based on Armstrong-like rules
and prove two versions of syntactico-semantical completeness,
the ordinary-style one claiming that entailment coinsides with provability and 
the graded-style one claiming that degrees of entailment equal 
degrees of provability. 
Section \ref{sec:bai} elaborates on the notion of a base, that is a non-redundant
set of graded attribute implications that contains, via entailment, complete information
about validity of all implications in a given data.
In particular, we focus on bases constructed by means of so-called pseudo-intents.
The algorithms for the problem of computing bases and some other problems 
regarding attribute implications are presented in Section \ref{sec:a}.
In Section \ref{sec:rgai}, we explore the problem of whether and to what extent 
it is possible to reduce the notions and problems regarding graded attribute
implications, notably the problem of computing a base,
to the corresponding problems regarding ordinary attribute implications.
Section \ref{sec:rfd} presents the above mentioned alternative semantics for graded
attribute implications in which  implications are interpreted as functional
dependencies over a certain extension of Codd's relational model, in which
domains of attribute values are equipped with binary fuzzy relations.
The binary relations may, in particular, be preference relations or similarity relations,
in which case the extension becomes a relational model enabling similarity
queries and other data processing involving 
 similarity relations. We examine such extension in detail in a another paper.
In this paper, we show that the two semantics are equivalent in that their notions of 
(degree of) entailment coincide.

\section{Scales of grades and basic principles of fuzzy logic}
\label{sec:p}

The dependencies studied in this paper are of the form (\ref{eqn:YZgrades}) and we assume that they are 
interpreted in tables with graded attributes. 
We assume that the grades involved (i.e. $a_i$s, $b_i$s, and the table entries) belong to a fixed set
$L$. Furthermore, we assume that $L$ is bounded by $0$ and $1$, partially ordered (usually a chain), 
and equipped with operations which are (truth functions of) logic connectives.
In accordance with fuzzy logic, we interpret the grades in $L$ as truth values, or truth degrees, with 
$0$ and $1$ representing falsity and truth. 
The intermediate degrees $a$, i.e. those with $0<a<1$, represent partial truth.
As in classical logic, grades are assigned to propositions to represent their validity.
The grade assigned to proposition $\varphi$ in structure $\mathbf{M}$ is denoted by
\[
    ||\varphi||_\mathbf{M} \mbox{ or just } ||\varphi||.
\]
Higher grades indicate truer propositions, hence
\[
     ||\text{$x$ is red}|| = 0.7 \qquad\text{and}\qquad  ||\text{$y$ is red}|| = 0.9 
\]
implies that $y$ is considered more red than $x$.
We consider  (truth functions of)  conjunction and implication and denote them by $\otimes$ and $\to$.
As usual in fuzzy logic, we assume truth functionality of connectives. That is, the truth degree
of $\varphi\&\psi$ and $\varphi\Rightarrow\psi$ (conjunction and implication of $\varphi$ and $\psi$)
is defined as
\[
      ||\varphi\&\psi|| = ||\varphi||\otimes||\psi||   \qquad\text{and}\qquad
      ||\varphi\Rightarrow\psi|| = ||\varphi||\to||\psi||.
\]
This way, the operations may be looked at as aggregation operations \cite{Fag:Cfi}. For instance, if 
\[
   ||\text{$x$ is brown}|| = 0.8  \qquad\text{and}\qquad  ||\text{$x$ is heavy}|| = 0.5, 
\]
and if $\otimes$ is the Goguen conjunction (see below in this section), then  the degree to which $x$ is brown and heavy
is $0.8\otimes 0.5=0.8\cdot 0.5=0.4$.
To be able to evaluate truth degrees of quantified formulas, we assume that as a partially ordered set,
$L$ forms a complete lattice, i.e. infima and suprema of arbitrary sets of grades exist. Namely,
if $\varphi$ is a formula with a free variable $x$ ranging over a set $D$, one naturally defines 
\[
      ||(\forall x)\varphi|| =\textstyle \bigwedge_{e} ||\varphi||_e   \qquad\text{and}\qquad
      ||(\exists x)\varphi|| =\textstyle \bigvee_{\!e} ||\varphi||_e,
\]
where $e$ ranges over all valuations of $x$ in $D$.
It has been recognized in his seminal work by Goguen \cite{Gog:Lfs,Gog:Lic} that a class of general
scales of grades equipped with operations suitable for fuzzy logic is the class of all complete
residuated lattices \cite{WaDi:Rl}. Residuated lattices and their variants are currently the main structures
used in mathematical fuzzy logic \cite{Gaea:RL,Got:TMVL,Got:Mfl,Haj:MFL} and are used as the basic structures 
of grades in this paper.

A complete residuated lattice 
\cite{Haj:MFL,Haj:Ovt} is an algebra
$\mathbf{L}=\langle L,\wedge,\vee,\otimes,\rightarrow,0,1\rangle$
such that $\langle L,\wedge,\vee,0,1 \rangle$ is
a complete lattice with $0$ and $1$ being the least and greatest elements, 
respectively; $\langle L,\otimes,1 \rangle$ is a commutative monoid
(i.e. $\otimes$  is commutative, associative,
and $a \otimes 1 = 1 \otimes a = a$ for each $a \in L$);
and $\otimes$ with $\rightarrow$
satisfy the so-called adjointness property:
\begin{eqnarray}
  a \otimes b \leq c \quad \mbox{if{}f} \quad a \leq b \rightarrow c
  \label{adj}
\end{eqnarray}
for each $a,b,c \in L$.
Commonly used residuated lattices are those with $L=[0,1]$ (unit
interval), $\wedge$ and $\vee$ being minimum and maximum, $\otimes$ being a
left-continuous t-norm \cite{Got:TMVL,Haj:MFL}  and $\rightarrow$ its residuum.
Three most important pairs of adjoint operations on the unit interval are:
\begin{eqnarray}
  \mbox{\L ukasiewicz:} &&
  \left.
    \begin{array}{r@{~=~}l}
      a \otimes b & \max(a + b - 1, 0), \\[2pt]
      a \rightarrow b & \min(1 - a + b, 1),
    \end{array}
  \right.\label{Def:Luk} \\[4pt]
  \mbox{G\"odel:} &&
  \left.
    \begin{array}{r@{~=~}l}
      a \otimes b & \min(a,b), \\[2pt]
      a \rightarrow b &
      \left\{
        \begin{array}{@{\,}ll}
          1 & \mbox{if}\ a \leq b, \\
          b & \mbox{otherwise,}
        \end{array}
      \right.
    \end{array}
  \right. \\[4pt]
  \mbox{Goguen (product):} &&
  \left.
    \begin{array}{r@{~=~}l}
      a \otimes b & a \cdot b, \\[2pt]
      a \rightarrow b &
      \left\{
        \begin{array}{@{\,}ll}
          1 & \mbox{if}\ a \leq b, \\
          \frac{b}{a} & \mbox{otherwise.}
        \end{array}
      \right.
    \end{array}
  \right.
\end{eqnarray}
Another important class of examples consists of residuated lattices that are
finite equidistant subchains in $[0,1]$, i.e.  $L = \{0, \frac{1}{n}, \dots,\frac{n-1}{n}, 1\}$.
Such chains may be endowed with the restrictions of {\L}ukasiewicz, G\"odel operations,
or other discrete t-norm-based operations \cite{MaTo:Tnds}.
Importantly, a particular example for $n=1$ yields $L=\{0,1\}$ in which case 
$\otimes$ and $\to$ are the classica conjunction and implication. In this case, $\mathbf{L}$ is 
the two-element Boolean algebra of classical logic and is denoted by $\mathbf{2}$ in this paper.

The following are the basic properties of complete residuated lattices that are needed
in our paper, see e.g. \cite{Bel:FRS,Got:TMVL,Haj:MFL}:

  \def\I{\rightarrow}%
  \def\E{\leftrightarrow}%
  \def\A{\wedge}%
  \def\O{\vee}%
  \def\bigA{\bigwedge}%
  \def\bigO{\bigvee}%
  \def\bigT{\bigotimes}%
  \def\smallA{\mathop{{\textstyle\bigwedge}}\nolimits}%
  \def\smallO{\mathop{{\textstyle\bigvee}}\nolimits}%
  \def\smallT{\mathop{{\textstyle\bigotimes}}\nolimits}%
  \def\smallcup{\mathop{{\textstyle\bigcup}}\nolimits}%
  \def\smallcap{\mathop{{\textstyle\bigcap}}\nolimits}%
  \def\C{\A}%
  \def\D{\O}%
  \def\M{\A}%
  \def\J{\O}%
  \def\R{\I}%
  \def\B{\E}%
  \def\N{\neg}%
  \def\T{\otimes}%

\begin{theorem} \label{thm:prl1}
Every complete residuated lattice satisfies
{ \begin{eqnarray}
  \label{eqn:prl1B}\label{Prop:leqiff}
 && a\leq b \qquad \mbox{\em if{}f } \qquad a\I b=1, \\
  \label{eqn:prl1C}
 && a\I a=1, \quad a\I 1=1, \quad 0\I a=1, \\
  \label{eqn:prl1D}
 && 1\I a=a, \\
  \label{eqn:prl1E}
 && a\T 0=0, \\
 \label{eqn:prl1F}
 && a\T b\leq a, \quad a\leq b\I a, \\
  \label{eqn:prl1A}
 && a\T(a\I b) \leq b, \quad
 b\leq a\I(a\T b), \quad a\leq(a\I b)\I b, \\
 \label{eqn:prl1H}\label{Prop:chprem}
 && (a\T b)\I c = a\I (b\I c) = b\I(a\I c), \\
 \label{eqn:prl1I}\label{Prop:Ttrans}
 && (a\I b)\T (b\I c)\leq a\I c, \\
 \label{eqn:isot}\label{Prop:isoT}
 && a_1\leq a_2 \ \mbox{ \em and } \ b_1 \leq b_2 \ \mbox{ \em implies } \
 a_1\T b_1 \leq a_2\T b_2, \\
 \label{eqn:isor}\label{Prop:isoR}
 && a_1\geq a_2 \ \mbox{ \em and } \ b_1 \leq b_2 \ \mbox{ \em implies } \
 a_1\I b_1 \leq a_2\I b_2, \\
%
 \label{Prop:ititit} 
 && (a \I b) \T (c \I d) \leq (a \T c) \I (b \T d), \\
%
 \label{eqn:prlsiA}\label{Prop:OaTb_aTOb}
  &&  a\T \smallO_{\!i\in I} b_i = \smallO_{\!i\in I} (a\T
  b_i),
\\
 \label{eqn:prlsiB}\label{Prop:AaIb_aIAb}
  &&   a\I \smallA_{i\in I} b_i = \smallA_{i\in I}(a\I
b_i),\\
 \label{eqn:prlsiC}\label{Prop:OaIb_aIOb}
  &&   \smallO_{\!i\in I} a_i \I b = \smallA_{i\in I}(a_i\I
  b),
   \\
 \label{eqn:prlsiD}\label{Prop:AaTb_aTAb}
  &&  a\T \smallA_{i\in I} b_i \leq \smallA_{i\in I}
(a\T b_i), \\
 \label{eqn:prlsiE}
  && \smallO_{\!i\in I} (a\I b_i) \leq a\I \smallO_{\!i\in I}
b_i,  \\
 \label{eqn:prlsiF}
  && \smallO_{\!i\in I} (a_i\I b)\leq \smallA_{i\in I} a_i\I
b, \\
  \label{Prop:Aab_AaAb}
  && \smallA_{i \in I}{(a_i \I b_i)} \leq
  \smallA_{i \in I}{a_i} \I \smallA_{i \in I}{b_i}.
\end{eqnarray}
}
\end{theorem}

Residuated lattices may be equipped with further operations. We utilize truth-stressing hedges
(shortly, hedges)
which are functions ${}^\ast:L\to L$ that represent intensifying linguistic modifiers such as ``very'' or ``highly''.
Such modifiers are used in propositions like ``this book is very good'' or, put differently, 
``it is very true that this book is good'', and may be thought of as unary logic connectives \cite{Haj:MFL,Haj:Ovt}. 
If ${}^\ast$ is the hedge representing the modifier ``very'', then the truth degree of
the proposition ``it is very true that $\varphi$'', shortly ``very $\varphi$'', is $||\varphi||^\ast$. 
That is,  one applies ${}^\ast$ to the truth degree of $\varphi$.
We assume that a truth-stressing hedge satisfies the following conditions, which are
inspired by the conditions from \cite{Haj:Ovt}:
\begin{eqnarray}
  1^{\ast} &=& 1, \label{TS:One} \\
  a^{\ast} &\leq& a, \label{TS:Sub} \\
  (a \rightarrow b)^{\ast} &\leq&
  a^{\ast} \rightarrow b^{\ast}, \label{TS:MP} \\
  a^{\ast\ast} &=& a^{\ast}, \label{TS:ID}
\end{eqnarray}
for each $a,b \in L$ ($i \in I$).
Properties (\ref{TS:One})--(\ref{TS:ID}) have a natural interpretation.
For instance, (\ref{TS:One}) says that if a proposition $\varphi$ is true (to degree $1$),
it is also very true (to degree $1$).
(\ref{TS:Sub}) says that if $\varphi$ is very true, then $\varphi$ is true;
(\ref{TS:MP}), which is equivalent to $a^{\ast}\otimes(a\rightarrow b)^{\ast}\leq b^{\ast}$, 
says that if $\varphi$ is very true and $\varphi \Rightarrow \psi$ is very,
then $\psi$ is very true; and (\ref{TS:ID}) says that ``very very $\varphi$''
has the same truth degree as ``very $\varphi$''.

Two boundary cases of hedges are
(i) identity, i.e. $a^{\ast} = a$ ($a \in L$);
(ii) globalization~\cite{TaTi:Gist}:
\begin{eqnarray}
  \label{eqn:glob}
  a^{\ast} = \left\{
    \begin{array}{@{\,}ll}
      1 & \quad \mbox{if}\ a = 1, \\
      0 & \quad \mbox{otherwise.}
    \end{array}
  \right.
\end{eqnarray}
Note that identity is the only hedge on the two-element Boolean algebra $\mathbf{2}$.

Given a complete residuate lattice $\mathbf{L}$, one defines the usual notions
regarding fuzzy sets:
an $\mathbf{L}$-set (fuzzy set) $A$ in universe $U$ is a mapping
$A\!:U \to L$, $A(u)$ being interpreted as 
``the degree to which $u$ belongs to $A$''.
If $U = \{u_1,\dots,u_n\}$ then $A$ can be denoted by
$A = \{\deg{a_1}{u_1},\dots,\deg{a_n}{u_n}\}$ meaning
that $A(u_i)$ equals $a_i$ for each $i=1,\dots,n$.
For brevity, we introduce the following convention:
we write $\{\dots,u,\dots\}$ instead of $\{\dots,\deg{1}{u},\dots\}$,
and we also omit elements of $U$ whose membership degree is zero.
For example, we write $\{u,\deg{0.5}{v}\}$
instead of $\{\deg{1}{u},\deg{0.5}{v},\deg{0}{w}\}$, etc.
Let $L^U$ or $\mathbf{L}^U$ (if the operations on $L$ are to be emphasized) 
denote the collection of all $\mathbf{L}$-sets in $U$.
The basic operations with $\mathbf{L}$-sets are defined componentwise. For instance,
the intersection of $\mathbf{L}$-sets $A,B \in \mathbf{L}^U$ is
an $\mathbf{L}$-set $A \cap B$ in $U$ such that
$(A \cap B)(u) = A(u) \wedge B(u)$ for each $u \in U$, etc.
Binary $\mathbf{L}$-relations (binary fuzzy relations) between $X$ and $Y$
can be thought of as $\mathbf{L}$-sets in the universe $X \times Y$.
That is, a binary $\mathbf{L}$-relation $I\in \mathbf{L}^{X\times Y}$ between
a set $X$ and a set $Y$ is a mapping assigning to each $x\in X$ and each $y\in
Y$ a truth degree $I(x,y)\in L$ (a degree to which $x$ and $y$ are related by
$I$).
An $\mathbf{L}$-set $A\in \mathbf{L}^X$ is called crisp if $A(x)\in\{0,1\}$
for each $x\in X$. Crisp $\mathbf{L}$-sets may obviously be identified with ordinary
sets. For a crisp $A$, we also write $x\in A$ if $A(x)=1$ and 
$x\not\in A$ if $A(x)=0$.
An $\mathbf{L}$-set $A\in\mathbf{L}^X$ is called empty (denoted by $\emptyset$)
if $A(x)=0$ for each $x\in X$. 

For $a\in L$ and $A\in\mathbf{L}^X$, the $\mathbf{L}$-sets $a\otimes A\in\mathbf{L}^X$
and $a\to A\in\mathbf{L}^X$
 are defined by
$(a\otimes A)(x)=a\otimes A(x)$ and $(a\to A)(x)=a\to A(x)$.


\section{Graded Attribute Implications and Their Semantics}
\label{sec:gai}

\subsection{Definition and Validity in Tables with Grades}

Throughout the paper, we assume that $Y$ is a finite and nonempty set of attributes.
The dependencies we consider, such as (\ref{eqn:uflahc}), are defined as follows.

\begin{definition}
A \emph{\emph{(}graded\emph{)} attribute implication} \emph{over $Y$}
is an expression $A \Rightarrow B$,
where $A,B \in \mathbf{L}^Y$ ($A$ and $B$ are $\mathbf{L}$-sets of attributes in $Y$).
\end{definition}

Note that since both $A$ and $B$ may be crisp in $A\Rightarrow B$,  
 i.e. $A(y), B(y)\in\{0,1\}$ for each $y\in Y$,  ordinary attribute implications 
(association rules, functional dependencies) are a particular case of graded attribute implications.
In addition, if $\mathbf{L}$ is the two-element Boolean algebra, graded implications
become just the the ordinary attribute implications.

Graded attribute implications are to be interpreted in tables whose entries
contain grades to which objects (represented by rows) have attributes (represented
by columns). Such tables are represented as triplets $\tu{X,Y,I}$ consisting of
non-empty sets $X$ of objects and $Y$ of attributes and an $\mathbf{L}$-relation
$I:X\times Y\to L$ for which the degree $I(x,y)$ is interpreted as the grade to which
the attribute $y\in Y$ applies to the object $x\in X$.

Consider first the implication 
\begin{eqnarray}\label{eqn:AIexample}
   \{{}^{1\!\!}/y_1,{}^{0.5\!\!}/y_3\} \!\Rightarrow\!
  \{{}^{0.8\!\!}/y_2,{}^{1\!\!}/y_4\}
\end{eqnarray}
and the table
%
%
\begin{equation}\label{eqn:simple_table}
 \mbox{
    \begin{tabular}{@{\,}c@{~}|@{~}c@{~~}c@{~~}c@{~~}c@{\,}}
      $I$ & $y_1$ & $y_2$ & $y_3$ & $y_4$  \\
      \hline
      $x_1$ & 1.0 & 0.9 & 0.8 & 1.0  \\
      $x_2$ & 1.0 & 0.7 & 0.8 & 1.0 \\
      $x_3$ & 0.9 & 0.5 & 0.8 & 1.0 
    \end{tabular}
  }
\end{equation}
On intuitive grounds,  (\ref{eqn:AIexample}) is satisfied by the object $x_1$ because
$x_1$ has all the attributes from the antecedent $A=\{{}^{1\!\!}/y_1,{}^{0.5\!\!}/y_3\}$
to the specified grades, i.e. $A(y_1)\leq I(x_1,y_1)$ and $A(y_3)\leq I(x_1,y_3)$,
and has also the attributes from the consequent $B=\{{}^{0.8\!\!}/y_2,{}^{1\!\!}/y_4\}$
to the specified grades, since $B(y_2)\leq I(x_1,y_2)$ and $B(y_4)\leq I(x_1,y_4)$.
While $x_2$ has the objects from $A$ to the specified grades as well, $y_2$
applies to $x_2$ to grade $0.7$ which is smaller than the grade $0.8$ prescribed
by $B$. Since $0.7$ is only slightly smaller than $0.8$, one naturally considers
(\ref{eqn:AIexample}) as an implication which is \emph{almost satisfied} by the object $x_2$, 
that is, satisfied to a high degree. 
The object $x_3$ does not have the attributes from $A$ to the specified grades,
because it posseses the attribute $y_1$ to grade $0.9$ while the grade prescribed
by $A$ is $1$. In testing the validity of (\ref{eqn:AIexample}) in the table, one may
therefore want to disregard $x_3$. However, if one wishes to work consistently with 
partial satisfiability, the same way one works with classic satisfiability,
one should involve $x_3$ and modify the test to take into account that $x_3$ satisfies
the antecendent $A$ partially.
Clearly, both  approaches, one in which only the objects fully satisfying $A$ participate in testing
the validity of $A\Rightarrow B$ and the other in which also objects partially satisfying $A$
paticipate in the test, coincide in the classical case with $0$ and $1$ as the only
grades. In the general case with intermediate grades involved, both approaches are plausible
and lead to two, different kinds of sematnics. As we show next, it turns out that both of the approaches
can conveniently be regarded as two particular cases of a general way to assess validity
of $A\Rightarrow B$ that is parameterized by how one evaluates the satifaction of $A$.

We now provide a definition of validity of a graded attribute implication $A\Rightarrow B$
in a table $\tu{X,Y,I}$ with grades. The basic structures in which $A\Rightarrow B$ is
evaluated are $\mathbf{L}$-sets of attributes. 
The rationale is that every row of $\tu{X,Y,I}$ corresponding to
the object $x\in X$ may be seen as the $\mathbf{L}$-set $I_x\in \mathbf{L}^Y$ given by 
$I_x(y)=I(x,y)$ for every $y\in Y$. Consider thus $M\in\mathbf{L}^Y$ representing object $x$, 
i.e. $M(y)$ is interpreted as the grade to which $x$ has $y$. 
According to (\ref{eqn:AIval}),  the truth degree $||A\Rightarrow B||_M$ to which $A\Rightarrow B$
is valid in $M$ is intended to be the truth degree of the proposition
``if $x$ has all attributes from $A$ then $x$ has all attributes from $B$'',
or equivalently, ``if $A$ is contained in $M$ then $B$ is contained in $M$''.

Containment of an $\mathbf{L}$-set $C$ in an $\mathbf{L}$-set $D$ is conveniently represented by the 
degree $S(C,D)$ of inclusion of $C$ in $D$ \cite{Got:TMVL}, defined by
\begin{eqnarray}
  S(C,D) = \textstyle{\bigwedge}_{y \in Y}\bigl(C(y) \rightarrow D(y)\bigr).
  \label{Eq:S}
\end{eqnarray}
$S(C,D)$ is the truth degree of proposition ``for each $y\in Y$: if $y$ belongs to $C$ then $y$ belongs to $D$''.
Clearly, $S$ is a graded relation which generalizes the inclusion relation of classical sets
in that if $L=\{0,1\}$ then $S$ is just the characteristic function of classical inclusion.
In particular, we write $C \subseteq D$ if $S(C,D) = 1$.
As a consequence of the fact that $a\to b=1$ if{}f $a\leq b$ we get that 
$C \subseteq D$ if and only if $C(y) \leq D(y)$ for each $y \in Y$,
i.e. $C\subseteq D$ means that $C$ is ``fully contained'' in $D$.
In what follows we use the well-known properties of graded inclusion \cite{Got:TMVL}.

With $S(A,M)$ and $S(B,M)$ being the degrees to which $A$ and $B$ are contained in $M$, 
respectively, one can define the degree  to which $A\Rightarrow B$
is valid in $M$ by $||A\Rightarrow B||_M=S(A,M)\to S(B,M)$. We provide a slightly more
general definition to account for both approaches described above, utilizing the notion
of hedge introduced in Section \ref{sec:p}.

\begin{definition}
Let $\mathbf{L}$ be a complete residuated lattice $\mathbf{L}$ with a truth-stressing 
hedge ${}^{\ast}$. 
The \emph{degree $||A\Rightarrow B||_M \in L$ to which $A \Rightarrow B$
is valid in an $\mathbf{L}$-set $M$} of attributes is defined by
\begin{eqnarray}\label{eqn:valDeg}
  ||A \Rightarrow B||_M = S(A,M)^{\ast} \rightarrow S(B,M).
\end{eqnarray}
\end{definition}

\begin{remark}
  (a) 
  If ${}^{\ast}$ is the globalization, i.e. $a^\ast=1$ for $a=1$ and $a^\ast=0$
  for $a<1$, we get $||A\Rightarrow B||_M=S(B,M)$ if $A\subseteq M$ and
  $||A\Rightarrow B||_M=1$ if $A\not\subseteq M$. Namely, 
   if $A\subseteq M$ then $S(A,M)^\ast \to S(B,M)=1^\ast\to S(B,M)=1\to S(B,M)=1$,
   and 
  if $A\not\subseteq M$, i.e. $S(A,M)<1$,
  then $S(A,M)^\ast \to S(B,M)=0\to S(B,M)=1$.
  This corresponds to the first approach mentioned above, in which only objects
  fully satisfying $A$ participate in testing validity.
  In addition, $A\Rightarrow B$ is fully true, i.e. $||A\Rightarrow B||_M=1$, if and only if
  \begin{equation}\label{eqn:AIglob}
      A\subseteq M \text{ implies } B\subseteq M.
  \end{equation}
  In this case, the degrees $A(y)$ and $B(y)$ may be seen as thresholds. Namely, according to
  (\ref{eqn:AIglob}), $A\Rightarrow B$ is satisfied by the object $x$ represented by $M$
  if and only if each attribute $y\in Y$  applies to $x$ in grade at least $A(y)$, 
  then each attribute $y\in Y$  applies to $x$ in grade at least $B(y)$.
  
  (b) 
  If ${}^{\ast}$ is the identity, then $||A\Rightarrow B||_M=S(A,M)\to S(B,M)$.
  This corresponds to the second approach mentioned above, in which also objects
  partially satisfying $A$ participate in the test of validity.
  In addition, since $a\to b=1$ if{}f $a\leq b$ for any $a,b\in L$,
  $A\Rightarrow B$ is fully true if and only if
  \begin{equation}\label{eqn:AIid}
      S(A,M)\leq S(B,M).
  \end{equation}
%

  (c)
  Globalization and identity represent the two natural ways to interpret graded
  attribute implications. In what follows, we develop the results for general hedges
  ${}^\ast$, covering thus both globalization and identity as particular cases.
\end{remark}

For a system ${\cal M}$ of $\mathbf{L}$-sets in $Y$, the degree
$||A \Rightarrow B||_{\cal M}$ to which $A \Rightarrow B$ is valid in (each
$M$ from) $\cal M$ is defined by
\begin{eqnarray}
  ||A \Rightarrow B||_{\cal M} =\textstyle \bigwedge_{M\in{\cal M}}||A \Rightarrow B||_M.
\end{eqnarray}
The degree $||A \Rightarrow B||_{\langle X,Y,I\rangle}$ to which $A \Rightarrow B$ is valid in
 a table $\langle X,Y,I\rangle$ with grades is defined by
\begin{eqnarray}
   ||A \Rightarrow B||_{\langle X,Y,I\rangle} =  ||A \Rightarrow B||_{\{I_x \,|\, x\in X\}} .
\end{eqnarray}
Recall that $I_x$ represents the $x$th row in $\tu{X,Y,I}$, i.e. for each $y\in Y$,
\begin{eqnarray}\label{eqn:Ix}
   I_x(y)=I(x,y).
\end{eqnarray}
Hence
 $||A \Rightarrow B||_{\langle X,Y,I\rangle}$ is naturally interpreted
as the degree to which $A\Rightarrow B$ is valid in every row of table $\tu{X,Y,I}$.

\begin{example}
   Consider again the implication $A\Rightarrow B$ in (\ref{eqn:AIexample}), 
   the table in (\ref{eqn:simple_table}), and the informal requirements
   discussed above in this example.
   Let $\mathbf{L}$ be the complete residuated lattice given by the  {\L}ukasiewicz
   operations on $[0,1]$. 
   Since $S(A,I_{x_1})=1$ and $S(B,I_{x_1})=1$, we get
   \[
     ||A\Rightarrow B||_{I_{x_1}}=S(A,I_{x_1})^\ast\to S(B,I_{x_1})=1\to 1=1.
   \]
   $A\Rightarrow B$ is thus fully satisfied by $x_1$, independently of the choice of ${}^\ast$
   because $1^\ast=1$ is always the case.
   For $x_2$, we have again $S(A,I_{x_2})=1$ but in this case, 
   $S(B,I_{x_2})=\bigwedge_{y\in Y}((B(y)\to I(x_2,y)))=B(y_2)\to I(x_2,y_2)=0.8\to 0.7=
   0.9$, whence
   \[
     ||A\Rightarrow B||_{I_{x_2}}=S(A,I_{x_2})^\ast\to S(B,I_{x_2})=1\to 0.9=0.9,
   \]
   again independently of the choice of ${}^\ast$. This corresponds to the intuitive requirement
   that $A\Rightarrow B$ be almost satisfied by $x_2$ because the grades specified by $B$
   are almost attained by the object $x_2$. 
   For $x_3$, we have $S(A,I_{x_3})=A(y_1)\to I(x_3,y_1)=1\to 0.9=0.9$, i.e. $A$ is only partially
   satisfied by $x_3$. According to the first approach to the semantics of implications,
   $x_3$ should not participate in the test of validity. Indeed, for ${}^\ast$ being globalization
   which corresponds to the first approach, we obtain
   \[
     ||A\Rightarrow B||_{I_{x_3}}=S(A,I_{x_3})^\ast\to S(B,I_{x_3})=0\to S(B,I_{x_3})=1,
   \]   
   because $0\to a=1$ for any degree $a$.
   For ${}^{\ast}$ being the identity, we get
   \[
     ||A\Rightarrow B||_{I_{x_3}}=S(A,I_{x_3})\to S(B,I_{x_3})=0.9\to 0.7=0.8,
   \]   
   which corresponds to the second approach to the semantics. We see that
   $x_3$ enters the test of validity in that the degree $S(B,I_{x_3})=0.7$ to which $x_3$
   satisfies $B$ is modified by the degree $S(A,I_{x_3})=0.9$ to which $x_3$ satisfies
   $A$. In particular, the modification is accomplished by shifting up the degree
   $S(B,I_{x_3})$; the smaller $S(A,I_{x_3})$ the more significant the shift. This is because
   we always have $S(A,I_{x_3})\to S(B,I_{x_3})\geq S(B,I_{x_3})$ and because $\to$ is antitone
   in the first argument.

   This example also makes it clear that testing to what degree an object $x$ satisfies 
   a consequent $B$ (or antecedent $A$) actually amounts to comparing the degrees $B(y)$ and $I(x,y)$
   for every attribute $y$. If $B(y)\leq I(x,y)$, the test is passed with degree $1$ for attribute $y$.
   If $B(y)\not\leq I(x,y)$, the test is passed with degree $B(y)\to I(x,y)<1$ for attribute $y$.
   In the end, the thus obtained degrees are aggregated by means of infimum which yields
   the degree to which $x$ satisfies $B$. Taking $B(y)\to I(x,y)$ if $B(y)\not\leq I(x,y)$ is appropriate
   because $\to$ is antitone in the first and isotone in the second argument.
   For example, for the {\L}ukasiewicz operations, $B(y)\to I(x,y)$ is $1-B(y)+I(x,y)$; 
   for the Goguen operations, $B(y)\to I(x,y)=I(x,y)/B(y)$.
\end{example}

\subsection{Theories, Models, and Semantic Entailment}
\label{sec:tmse}

When reasoning with degrees, theories are naturally conceived as $\mathbf{L}$-sets of formulas.
A (\emph{graded}) \emph{theory}  is therefore an $\mathbf{L}$-set of graded implications over $Y$. The degree 
$T(A\Rightarrow B)$ is considered as the degree to which we assume the validity of $A\Rightarrow B$.
This approach corresponds to the notion of a theory as an $\mathbf{L}$-set (fuzzy set) of axioms in fuzzy
logic~\cite{Pav:Ofl}, see also \cite{Ger:FL,Haj:MFL,NoPeMo:MPFL}.
If $T$ is a \emph{crisp theory}, i.e. $T(A\Rightarrow B)$ is $0$ or $1$ for every $A\Rightarrow B$,  
we write $A\Rightarrow B\in T$ if $T(A\Rightarrow B)=1$ and $A\Rightarrow B\not\in T$
if $T(A\Rightarrow B)=0$.

For a theory $T$, the set $\mathrm{Mod}(T)$
of all \emph{models} of $T$ is defined by
\[
   \mathrm{Mod}(T)=\{M\in\mathbf{L}^Y \,|\, \mbox{for each }
   A,B\in\mathbf{L}^Y: T(A\Rightarrow B)\leq ||A \Rightarrow B||_M\}.
\]
That is, $M\in\mathrm{Mod}(T)$ means that for each attribute implication 
$A\Rightarrow B$, the degree to which $A\Rightarrow B$ holds in $M$ is higher
than or at least equal to the degree $T(A\Rightarrow B)$ prescribed by $T$.
In particular, for a crisp $T$ we have
$\mathrm{Mod}(T)=\{M\in\mathbf{L}^Y \,|\, \mbox{for each }
   A\Rightarrow B\in T: ||A \Rightarrow B||_M=1\}$.

The degree $||A \Rightarrow B||_T \in L$ to which $A \Rightarrow B$
\emph{semantically follows} from a fuzzy set $T$ of attribute implications is defined by
\begin{eqnarray}
  \label{eqn:entAI}
  ||A \Rightarrow B||_T =
  \textstyle\bigwedge_{M \in \mathrm{Mod}(T)}||A \Rightarrow B||_M.
\end{eqnarray}
That is, $||A \Rightarrow B||_T$ may be seen as the degree to which $A\Rightarrow B$ is valid
in every model of $T$.

We need the following lemma.

\begin{lemma}\label{Th:Sem_graded}
  \upshape{(i):}
  $c \rightarrow S(B,M) = S(c \otimes B,M)=S(B,c\to M)$;

  \upshape{(ii):}
  $c \to ||A \Rightarrow B||_M = 
  ||A \Rightarrow c \otimes B||_M$.

  \upshape{(iii):}
  $c \leq ||A \Rightarrow B||_M$ if{}f\/
  $||A \Rightarrow c \otimes B||_M = 1$.
%
\end{lemma}
\begin{proof}
  (i): On account of (\ref{eqn:prlsiB}) and (\ref{eqn:prl1H}) and we have
  \begin{eqnarray*}
    &&\kern-14mm
    c \rightarrow S(B,M) =
    c \rightarrow \textstyle\bigwedge_{y \in Y}(B(y) \rightarrow M(y)) = 
    \textstyle\bigwedge_{y \in Y}(c \rightarrow (B(y) \rightarrow M(y))) = \\
         &=&
    \textstyle\bigwedge_{y \in Y}((c \otimes B(y)) \rightarrow M(y)) = 
    \textstyle\bigwedge_{y \in Y}((c \otimes B)(y) \rightarrow M(y)) =
    S(c \otimes B, M).
  \end{eqnarray*}
 $S(c \otimes B,M)=S(B,c\to M)$ is an easy consequence of (\ref{eqn:prl1H}).

  (ii): Using (\ref{eqn:prl1H}) and (i),
   \begin{eqnarray*}
    && \kern-14mm
      c \to ||A \Rightarrow B||_M = c\to(S(A,M)^\ts\to S(B,M)) =
   S(A,M)^\ts\to(c\to S(B,M)) =\\
    &=& S(A,M)^\ts\to S(c\otimes B,M) =  ||A \Rightarrow c \otimes B||_M.
   \end{eqnarray*}

  (iii): Direct consequence of (ii) and (\ref{eqn:prl1B}).
%
\end{proof}

Lemma~\ref{Th:Sem_graded} implies 
every graded theory may be transformed to a crisp theory with the same 
models and thus (degrees of) consequences:

\begin{theorem}\label{thm:cT}
  Let $T$ be a theory, $A\Rightarrow B$ be a graded attribute implication. 
For the crisp theory $\mathrm{cr}(T)$ 
  defined by
  \begin{equation}
    \label{eqn:cT}
     \mathrm{cr}(T) = \{ A\Rightarrow T(A\Rightarrow B)\otimes B \,|\,
     A,B\in\mathbf{L}^Y \mbox{ and } T(A\Rightarrow B)\otimes  
     B\not=\emptyset\}
  \end{equation}
  we have
  \begin{eqnarray}
     \label{eqn:modTT}
     \mathrm{Mod}(T) &=& \mathrm{Mod}(\mathrm{cr}(T)),\\
     \label{eqn:semTT}
     ||A\Rightarrow B||_T &=& ||A\Rightarrow B||_{\mathrm{cr}(T)}. 
  \end{eqnarray}
\end{theorem}
\begin{proof}
%
  (\ref{eqn:modTT}) directly using (iii) of  Lemma~\ref{Th:Sem_graded}.
  (\ref{eqn:semTT}) is a consequence of (\ref{eqn:modTT}).
\end{proof}

Furthermore,
Lemma~\ref{Th:Sem_graded} enables us to reduce the concept of a degree of 
entailment of  to the concept of entailment in degree $1$ (full entailment): 

\begin{theorem}\label{Th:Sem_graded3}
  For a graded theory $T$ and an implication $A\Rightarrow B$ we have
\[
  ||A \Rightarrow B||_T = \textstyle\bigvee\{c \in L \,|~
  ||A \Rightarrow c \otimes B||_T = 1\}.
\]
\end{theorem}
\begin{proof}
    Using (iii) of Lemma~\ref{Th:Sem_graded}, we have
    \begin{eqnarray*}
      &&\kern-6mm
      ||A \Rightarrow B||_T =
      \textstyle\bigwedge_{M \in \mathrm{Mod}(T)}||A \Rightarrow B||_M = \\
      &&\kern-4mm
      = \textstyle\bigvee\{c \in L \,|\, c \leq ||A \Rightarrow B||_M
      \mbox{~for~each~} M \in \mathrm{Mod}(T)\} = 
       \textstyle\bigvee\{c \in L \,|~||A \Rightarrow c \otimes B||_T = 1\}.
    \end{eqnarray*}
\end{proof}  

Therefore, the concept of a degree of entailment by graded theories
may be reduced to that of entailment in degree 1 (full entailment) by crisp theories:

\begin{corollary}\label{thm:reducT}
  \(
    ||A \Rightarrow B||_T = \textstyle\bigvee\{c \in L \,|~
    ||A \Rightarrow c \otimes B||_{\mathrm{cr}(T)} = 1\},
  \)
  with $\mathrm{cr}(T)$ defined by \textup{(\ref{eqn:cT})}.
\end{corollary}


\subsection{Closure Properties of Models of Graded Implications}
\label{sec:cp}

In the classic setting, models of theories of implications (equivalently, functional dependencies)
are closed under intersections. This enables one to test whether $A\Rightarrow B$ follows from
a theory $T$ by checking whether $A\Rightarrow B$ is valid in a single model of $T$, namely
the least model of $T$ containing $A$ \cite{GaWi:FCA,Mai:TRD}. In this section we establish
the corresponding results for a setting with grades.

Recall from \cite{BeFuVy:Fcots} that a system $\mathcal{S}\subseteq \mathbf{L}^Y$ of $\mathbf{L}$-sets
in $Y$ is called an $\mathbf{L}^\ast$-closure system if it is closed under intersections and $a^\ast$-shifts,
i.e. satisfies the following conditions:
\begin{eqnarray}
  &&\label{eqn:ci}
  \text{if $A_j\in\mathcal{S}$ for $j\in J$ then $\textstyle\bigcap_{j\in J} A_j\in\mathcal{S}$,}\\
  &&\label{eqn:cs}
 \text{if $a\in L$ and $A\in\mathcal{S}$ then $a^\ast\to A\in\mathcal{S}$}. 
\end{eqnarray}
Note that here, $(\bigcap_{j\in J} A_i)(y)=\bigwedge_{j\in J} A_j(y)$ and
  $(a^\ast\to A)(y) = a^\ast\to A(y)$.
Recall furthermore that an $\mathbf{L}^*$-closure operator \cite{BeFuVy:Fcots}
on a set $Y$ is a mapping $C\!: \mathbf{L}^Y \to \mathbf{L}^Y$
satisfying, for each $A,A_1,A_2 \in \mathbf{L}^Y$,
\begin{eqnarray}
  A  &\subseteq& C(A), \label{eqn:fco1} \\
  S(A_1,A_2)^* &\leq& S(C(A_1),C(A_2)), \label{eqn:fco2} \\
  C(A) &=& C(C(A)),\label{eqn:fco3}
\end{eqnarray}
where $S$ is the degree of inclusion defined by (\ref{Eq:S}).
If $L=\{0,1\}$,  $\mathbf{L}^\ast$-closure systems 
and $\mathbf{L}^\ast$-closure operators
may be identified with ordinary closure systems and closure operators \cite{DaPr:ILO}, since
(\ref{eqn:cs}) is satisfied for free and (\ref{eqn:fco2}) asserts monotony of $C$
with respect to set inclusion in this case.
According to \cite{BeFuVy:Fcots}, letting for an $\mathbf{L}^\ast$-closure system
$\mathcal{S}$ and an $\mathbf{L}^\ast$-closure operator $C$,
  \begin{eqnarray}
    C_{\cal S}(B) = \textstyle\bigcap_{i\in I}(S(B,A_i)^* \rightarrow A_i)
    \label{eqn:leastcl}
  \end{eqnarray}
and
\[
  {\cal S}_C = \{A\in \mathbf{L}^U \,|\, A = C(A)\},
\]
$C_{\cal S}$ is an $\mathbf{L}^\ast$-closure operator, ${\cal S}_C$ is an 
$\mathbf{L}^\ast$-closure system, and the mappings ${\cal S}\mapsto C_{\cal S}$
and $C\mapsto{\cal S}_C$ are mutually inverse bijections.

%
%

\begin{theorem}\label{thm:Modcls}
  $\mathrm{Mod}(T)$ is an $\mathbf{L}^*$-closure system in $Y$ for any graded theory $T$ of implications
  over $Y$.
\end{theorem}
\begin{proof}
  We need to check
  (\ref{eqn:ci}) and (\ref{eqn:cs}). Due to Theorem \ref{thm:cT}, we may safely assume
  that $T$ is crisp.

  (\ref{eqn:ci}):
  Consider a $J$-indexed system $\{M_j \in \mathrm{Mod}(T) \,|\, j \in J\}$
  of models of $T$. We show that $\bigcap_{j \in J}\!M_j$ is a model of $T$.
  Thus, we check that, for each $A \Rightarrow B \in T$,
  $||A \Rightarrow B||_{\bigcap_{j \in J}\!M_j} = 1$.
  Since each $M_j$ is a model of $T$, 
  we have $||A \Rightarrow B||_{M_j} = 1$, i.e.
  $S(A,M_j)^* \leq S(B,M_j)$, for any $A \Rightarrow B \in T$. 
  Now, since
  $(\bigwedge_{j \in J}a_j)^* \leq \bigwedge_{j \in J}a^*_j$, we get
  \begin{align*}
    \textstyle S(A,\bigcap_{j \in J}\!M_j)^* &=
    \textstyle \bigl(\bigwedge_{j \in J}S(A,M_j)\bigr)^* \leq
    \textstyle \bigwedge_{j \in J}S(A,M_j)^* \leq 
   \textstyle\bigwedge_{j \in J}S(B,M_j)=S(B,\bigcap_{j \in J}\!M_j),
  \end{align*}
  proving $||A \Rightarrow B||_{\bigcap_{j \in J}\!M_j} = 1$, and hence
  $\bigcap_{j \in J}\!M_j \in \mathrm{Mod}(T)$.

  (\ref{eqn:cs}):
  Let $M \in \mathrm{Mod}(T)$ and  $a \in L$.
  We need to check that $a^* \rightarrow M$ belongs to $\mathrm{Mod}(T)$.
  Since $M$ is a model of $T$, for each $A \Rightarrow B \in T$
  we have $||A \Rightarrow B||_M = 1$, i.e., $S(A,M)^* \leq S(B,M)$.
  Using Lemma \ref{Th:Sem_graded} (i), \eqref{TS:MP}, \eqref{TS:ID}, and monotony of $\rightarrow$
  in the second argument, we get
  \begin{align*}
    S(A,a^* \rightarrow M)^* &= (a^* \rightarrow S(A,M))^* \leq \\
    &\leq a^{**} \rightarrow S(A,M)^* = a^* \rightarrow S(A,M)^* \leq 
    a^* \rightarrow S(B,M) = S(B,a^* \rightarrow M),
  \end{align*}
  establishing $||A\Rightarrow B||_{a^* \rightarrow M }=1$ for an 
  arbitrary $A\Rightarrow B\in T$, whence
  $a^* \rightarrow M \in \mathrm{Mod}(T)$.%
\end{proof}

The following assertion shows the converse claim to Theorem \ref{thm:Modcls}.

\begin{theorem}\label{thm:Clsmod}
  Let ${\cal S}$ be an $\mathbf{L}^*$-closure system in $Y$.
  Then there exists a theory $T$ of graded attribute implications over $Y$ such that
  ${\cal S} = \mathrm{Mod}(T)$.
\end{theorem}
\begin{proof}
  Put $T = \{A \Rightarrow C_{\cal S}(A) \,|\, A \in \mathbf{L}^Y\}$.
  Let $M \in {\cal S}$. 
  Then $M = C_{\cal S}(M)$ and due to \eqref{eqn:fco2},
  $S(A,M)^* \leq S(C_{\cal S}(A),C_{\cal S}(M)) = S(C_{\cal S}(A),M)$,
  which gives
  $||A \Rightarrow C_{\cal S}(A)||_M = 1$, i.e. $M$ is a model of $T$.
  This proves ${\cal S} \subseteq \mathrm{Mod}(T)$.
  Conversely let $M \not\in {\cal S}$, i.e. $M \ne C_{\cal S}(M)$.
  Then  $M \subset C_{\cal S}(M)$ by~\eqref{eqn:fco1},
  whence $S(C_{\cal S}(M),M) \ne 1$.
  As a result, 
  \[
  ||M \Rightarrow C_{\cal S}(M)||_M = S(M,M)^* \rightarrow
  S(C_{\cal S}(M),M) = 1^* \rightarrow S(C_{\cal S}(M),M) =
  S(C_{\cal S}(M),M) \ne 1,
  \]
   i.e. $M \not\in \mathrm{Mod}(T)$, proving
  $\mathrm{Mod}(T) \subseteq {\cal S}$.%
\end{proof}

  Theorem~\ref{thm:Modcls} and Theorem~\ref{thm:Clsmod} imply that
  systems of models of graded attibute implications over $Y$ coincide with 
  $\mathbf{L}^*$-closure systems over $Y$.
In particular, given a theory $T$ and an arbitrary $A\in\mathbf{L}^Y$,
one may consider the least model of $T$ that contains $A$.
As is well-known from the ordinary case \cite{GaWi:FCA}, an ordinary implication
$A\Rightarrow B$ follows from $T$ if an only if the least model of $T$
that contains $A$  includes $B$. As we show next, this property
generalizes to the setting involving grades in that degree of entailment
equals degree of inclusion. 
In our setting,
the least model is $C_{\mathrm{Mod}(T)}(A)$ where $C_{\mathrm{Mod}(T)}$
is the $\mathbf{L}^*$-closure operator corresponding to $T$ according to 
Theorem \ref{thm:Clsmod}.
As shown by the following theorem, $C_{\mathrm{Mod}(T)}(A)$ may be used 
to determine the degree to which $A\Rightarrow B$ semantically follows from $T$.
Namely, the degree equals the degree to which $A\Rightarrow B$ is valid in 
$C_{\mathrm{Mod}(T)}(A)$ as well as the degree of inclusion
of $B$ in $C_{\mathrm{Mod}(T)}(A)$.

\begin{theorem}\label{thm:ent}
  \itshape
  For every theory $T$ and a  graded attribute implication $A \Rightarrow B$,
  \begin{displaymath}
    ||A \Rightarrow B||_T =
    ||A \Rightarrow B||_{C_{\mathrm{Mod}(T)}(A)} =
    S(B,C_{\mathrm{Mod}(T)}(A)).
  \end{displaymath}
\end{theorem}
\begin{proof}
  Clearly,  $||A\Rightarrow B||_T \leq
  ||A\Rightarrow B||_{C_{\mathrm{Mod}(T)}(A)}$
  because $C_{\mathrm{Mod}(T)}(A) \in \mathrm{Mod}(T)$.
  Moreover, since $C_{\mathrm{Mod}(T)}$ satisfies \eqref{eqn:fco1},
  \begin{align*}
    ||A&\Rightarrow B||_{C_{\mathrm{Mod}(T)}(A)} 
    = S(A,C_{\mathrm{Mod}(T)}(A))^* \rightarrow
    S(B,C_{\mathrm{Mod}(T)}(A))\\
    & = 
    1 \rightarrow S(B,C_{\mathrm{Mod}(T)}(A)) =
    S(B,C_{\mathrm{Mod}(T)}(A)).
  \end{align*}
  Take any $M \in \mathrm{Mod}(T)$. Due to \eqref{eqn:fco2}
  and $M=C_{\mathrm{Mod}(T)}(M)$, 
  \begin{align*}
     & S(B,C_{\mathrm{Mod}(T)}(A)) \otimes S(A,M)^* \leq 
    S(B,C_{\mathrm{Mod}(T)}(A)) \otimes
    S(C_{\mathrm{Mod}(T)}(A),C_{\mathrm{Mod}(T)}(M)) \leq \\
    &\leq S(B,C_{\mathrm{Mod}(T)}(M)) = S(B,M).
  \end{align*}
  Applying adjointness, we get
  \begin{align*}
    S(B,C_{\mathrm{Mod}(T)}(A)) \leq
    S(A,M)^* \!\rightarrow S(B,M) =
    ||A \Rightarrow B||_M,
  \end{align*}
  for each $M \in \mathrm{Mod}(T)$.
  Hence, 
  $S(B,C_{\mathrm{Mod}(T)}(A)) \leq ||A \Rightarrow B||_T$.
\end{proof}

\subsection{Related structures and alternative formulas for validity}
\label{sec:rs}

Every table $\tu{X,Y,I}$ with grades induces an important pair of operators.
These operators, along with the sets of their fixpoints, were studied in
\cite{Bel:Fgc,Bel:Clofl,BeVy:Fcalh}. In this section, we present the basic
connections of these structures to graded attribute implications. 
In addition, we provide alternative formulas for validity of implications. 

Given a table $\tu{X,Y,I}$, consider the operators ${}^\upts:\mathbf{L}^X\to\mathbf{L}^Y$
and ${}^\upts:\mathbf{L}^Y\to\mathbf{L}^X$ given by 
\begin{equation} \label{eqn:updownts}
  \textstyle
 A^\up(y)=\bigwedge_{x\in X} ( A(x)^{\ast}\rightarrow I(x,y))
 \qquad\text{and}\qquad
 B^\down(x)=\bigwedge_{y\in Y} ( B(y)\rightarrow I(x,y)).
\end{equation}
The pair $\tu{{}^\upts,{}^\downts}$ forms an $\mathbf{L}^\ast$-Galois connection
\cite{Bel:Fgc,Bel:Clofl,BeVy:Fcalh}.
Note that the formulas in (\ref{eqn:updownts}) are not symmetric because
we consider only a particular form of these operators, which are directly linked
to graded attribute implications. The general formulas involve two hedges, one for
$X$ and one for $Y$.
Using basic rules of predicate fuzzy logic,  $A^\up(y)$ is the truth degree of 
``for each $x\in X$: if it is very true that $x$ belongs to $A$ then
$y$ applies to $x$''. 
Likewise,  $B^\down(x)$ is the truth degree of ``for each $y\in Y$: if $y$ belongs 
to $B$ then $y$ applies to $x$''.
The set
\[
\mathcal{B}(X^{\ast},Y,I)=\{\langle A,B\rangle\in \mathbf{L}^X\times \mathbf{L}^Y \,|\, A^\up=B, \ B^\down=A\}
\]
of all fixpoints of $\langle{}^\up,{}^\down\rangle$ is called the \emph{concept lattice}
of $\tu{X,Y,I}$.
Its elements, called \emph{formal concepts} of $\tu{X,Y,I}$, are naturally interpreted as
concepts in the sense of traditional logic \cite{GaWi:FCA}.
Namely, for every formal concept $\tu{A,B}\in\mathcal{B}(X^{\ast},Y,I)$, $A$ and $B$ may be seen
as its \emph{extent} and its \emph{intent}, i.e.  the collections of objects and attributes, respectively,
which are covered by the concept. 
Both $A$ and $B$ are $\mathbf{L}$-sets, i.e. represent graded collections and apply to objects and attributes to
 degrees, not necessarily 0 and 1 only. 
The set of all intents, which plays an important role for graded attribute implications, is denoted
by $\mathrm{Int}(X^{\ast},Y,I)$, i.e.
\[
   \mathrm{Int}(X^{\ast},Y,I) =\{  B\in\mathbf{L}^Y \mid \tu{A,B}\in\mathcal{B}(X^{\ast},Y,I) \mbox{ for some } A \}.
\]
Note that 
\begin{equation}\label{eqn:Mint}
 \text{for each } B\in \mathbf{L}^Y: \  B \in  \mathrm{Int}(X^{\ast},Y,I) \text{ if and only if } B = B^{\downts\upts}
\end{equation}

As the following theorem shows, validity of $A\Rightarrow B$ in a data table $\tu{X,Y,I}$ may equivalently
be expressed as
the  validity of $A\Rightarrow B$ in the set of all intents of $\tu{X,Y,I}$ as well as 
the degree of inclusion of $B$ in the ${}^{\upts\down}$-closure of $A$.

\begin{theorem}\label{thm:alter}
  For every $\langle X,Y,I\rangle$, 
  \begin{eqnarray}
    ||A \Rightarrow B||_{\langle X,Y,I\rangle} = 
    ||A \Rightarrow B||_{\mathrm{Int}(X^{\ast},Y,I)} =
    S(B,A^{\downarrow\uparrow}).
    \label{Th:B_eq_tab}
  \end{eqnarray}
\end{theorem}
\begin{proof}
  First, we check $||A \Rightarrow B||_{\langle X,Y,I\rangle} =
  ||A \Rightarrow B||_{\mathrm{Int}(X^{\ast},Y,I)}$.
    Observe that $||A \Rightarrow B||_{\langle X,Y,I\rangle} \leq
  ||A \Rightarrow B||_{\mathrm{Int}(X^{\ast},Y,I)}$ if{}f
  for each $M \in \mathrm{Int}(X^{\ast},Y,I)$ we have
  $||A \Rightarrow B||_{\langle X,Y,I\rangle} \leq ||A \Rightarrow B||_M$, i.e.
   \\[4pt]
  \centerline{$\bigwedge_{x \in X}\bigl(S(A,I_x)^{\ast} \rightarrow
    S(B,I_x)\bigr) \leq S(A,M)^{\ast} \rightarrow S(B,M) $.} \\[4pt]
   As $I_x(y)=I(x,y)$, we have $S(A,I_x)=\bigwedge_{y\in Y}(A(y)\to I_x(y))=A^\down(x)$.
  Therefore, the last inequality is equivalent to
   \\[4pt]
  \centerline{$\bigwedge_{x \in X}
    \bigl(A^{\downts}(x)^{\ast} \rightarrow B^{\downts}(x)\bigr)
    \leq S(A,M)^{\ast} \rightarrow S(B,M) $,} \\[4pt]
  i.e. to \\[4pt]
  \centerline{$S(A^{\downts\ast}, B^{\downts}) = \bigwedge_{x \in X}
    \bigl((A^{\downts\ast}(x) \rightarrow B^{\downts}(x)\bigr)
    \leq S(A,M)^{\ast} \rightarrow S(B,M) $,} \\[4pt]
  which is equivalent to 
  \begin{eqnarray}\label{Eq:Tab_eq_Int}
    S(A,M)^{\ast} \otimes S(A^{\downts\ast}, B^{\downts}) \leq S(B,M)
  \end{eqnarray}
  due to adjointness of $\otimes$ and $\to$.
  Thus, it suffices
  to prove (\ref{Eq:Tab_eq_Int})   for each $M \in  \mathrm{Int}(X^{\ast},Y,I)$.
  For this purpose, consider the operator ${}^\Uparrow$, the ``unhedged'' version
  of ${}^\upts$ defined by 
  \[
     A^\Uparrow(y)=\textstyle\bigwedge_{x\in X} ( A(x)\rightarrow I(x,y)).
  \]
  The pair $\tu{{}^\Uparrow,{}^\downts}$ forms an $\mathbf{L}$-Galois connection and 
  hence satisfies $S(C_1,C_2)\leq S(C_2^\Uparrow,C_1^\Uparrow)$,
  $S(D_1,D_2)\leq S(D_2^\downts,D_1^\downts)$,  and $D\subseteq D^{\downts\Uparrow}$,
  see \cite{Bel:Fgc}.
  Due to (\ref{eqn:Mint}) and since
  $S(C,D)\otimes S(D,E)\leq S(C,E)$, we obtain 
  \begin{eqnarray*}
    &&\kern-1cm
    S(A,M)^{\ast} \otimes S(A^{\downts\ast}, B^{\downts}) \leq
    S(M^{\downts},A^{\downts})^{\ast} \otimes
    S(A^{\downts\ast}, B^{\downts}) \leq \\
    &\leq&
    S(M^{\downts\ast},A^{\downts\ast}) \otimes
    S(A^{\downts\ast}, B^{\downts}) \leq
    S(M^{\downts\ast}, B^{\downts}) \leq \\
    &\leq&
    S(B^{\downts\Uparrow},M^{\downts\ast\Uparrow}) =
    S(B^{\downts\Uparrow},M^{\downts\upts}) =  S(B^{\downts\Uparrow},M)
    \leq   S(B,M),
  \end{eqnarray*}
  verifying~(\ref{Eq:Tab_eq_Int}) and thus 
  $||A \Rightarrow B||_{\langle X,Y,I\rangle} \leq ||A \Rightarrow B||_{\mathrm{Int}(X^{\ast},Y,I)}$
  To check  $||A \Rightarrow B||_{\langle X,Y,I\rangle} \geq ||A \Rightarrow B||_{\mathrm{Int}(X^{\ast},Y,I)}$
  it is sufficient to observe that for each $x\in X$,  $I_x\in \mathrm{Int}(X^{\ast},Y,I)$.
  This fact follows from (\ref{eqn:Mint}) since, as one can easily see, \\[4pt]
  \centerline{%
    $I_x = \{{}^{1\!\!}/x\}^{\uparrow} = \{{}^{1\!\!}/x\}^{\ast\Uparrow} =
    \{{}^{1\!\!}/x\}^{\ast\Uparrow\downts\ast\Uparrow} =
    {I_x}^{\downts\ast\Uparrow} = {I_x}^{\downarrow\uparrow}$.}

  \medskip
 Second, we check
  $||A \Rightarrow B||_{\langle X,Y,I\rangle} =
  S(B,A^{\downarrow\uparrow})$.
  We have
  \begin{eqnarray*}
    &&
    \kern-1.2cm
    ||A \Rightarrow B||_{\langle X,Y,I\rangle} = \\
    &=&
    \textstyle\bigwedge_{x \in X}
    \bigl(S(A,I_x)^{\ast} \rightarrow S(B,I_x)\bigr) = \\
    &=&
    \textstyle\bigwedge_{x \in X}
    \bigl(A^{\downts}(x)^{\ast} \rightarrow B^{\downts}(x)\bigr) = \\
    &=&
    \textstyle\bigwedge_{x \in X}
    \bigl(A^{\downts\ast}(x) \rightarrow
    \textstyle\bigwedge_{y \in Y}
    \bigl(B(y) \rightarrow I(x,y)\bigr)\bigr) = \\
    &=&
    \textstyle\bigwedge_{y \in Y}
    \textstyle\bigwedge_{x \in X}
    \bigl(A^{\downts\ast}(x) \rightarrow
    \bigl(B(y) \rightarrow I(x,y)\bigr)\bigr) = \\
    &=&
    \textstyle\bigwedge_{y \in Y}
    \textstyle\bigwedge_{x \in X}
    \bigl(B(y) \rightarrow
    \bigl(A^{\downts\ast}(x) \rightarrow I(x,y)\bigr)\bigr) = \\
    &=&
    \textstyle\bigwedge_{y \in Y}
    \bigl(B(y) \rightarrow
    \textstyle\bigwedge_{x \in X}
    \bigl(A^{\downts\ast}(x) \rightarrow I(x,y)\bigr)\bigr) = \\
    &=&
    \textstyle\bigwedge_{y \in Y}
    \bigl(B(y) \rightarrow A^{\downts\upts}(y)\bigr) = 
     S(B,A^{\downts\uparrow}),
  \end{eqnarray*}
  proving the claim.
\end{proof}

We now present several other formulas expressing the degree  
$||A\Rightarrow B||_{\langle X,Y,I\rangle}$.
They show that, in a sense, globalization may be regarded as the basic
hedge in the definition (\ref{eqn:valDeg}). 
First, for a hedge $\ts$ on $\mathbf{L}$ put $\mathrm{fix}(\ts)=\{a\in L\,|\, a^\ts=a\}$ (set of all fixpoints of $\ts$).
Furthermore, for $\tss,\ts:L\rightarrow L$ put $\tss\leq\ts$ iff $a^\tss\leq a^\ts$ for each $a\in L$
($\tss$ is as strong or stronger than $\ts$). 
One can easily see that for hedges $\ts$ and $\tss$ on a complete residuated lattice $\mathbf{L}$,
\begin{eqnarray}\label{eqn:pohedge}
 \mbox{$\tss\leq\ts$ \quad if{}f \quad $\mathrm{fix}(\tss)\subseteq\mathrm{fix}(\ts)$.} 
\end{eqnarray}
Denote by $||A\Rightarrow B||^\tss_{\dots}$ the degree of validity
of $A\Rightarrow B$ in $\dots$ that involves the hedge $^\tss$. Thus,
$||A\Rightarrow B||^\tss_{M}=S(A,M)^\tss\to S(B,M)$ and the like.
Omitting the superscript, i.e. $||A\Rightarrow B||_{M}$ always means
$||A\Rightarrow B||^\ts_{M}$.
We need the following lemma.

\begin{lemma} \label{thm:lemma2}
 For $A,B,M\in\mathbf{L}^Y$, and hedges $\tss$ and $\ts$ for which 
 $\tss\leq \ts$ we have 
 \[
   ||A\Rightarrow B||_{M}=\textstyle\bigwedge_{a\in L}
 \bigl(S(a^{\ts}\otimes A,M)^{\tss} \rightarrow S(a^{\ts}\otimes B,M)\bigr)=
 S(S(A,M)^{\ts}\otimes B,M).
\]
\end{lemma}
\begin{proof}
  $||A\Rightarrow B||_{M}= S(A,M)^{\ts}\rightarrow S(B,M)=
  S(S(A,M)^{\ts}\otimes B,M)$ follows directly from 
   Lemma~\ref{Th:Sem_graded}. 

   Next, we check both inequalities of
   $||A\Rightarrow B||_{M} = \bigwedge_{a\in L}
   \bigl(S(a^{\ts}\otimes A,M)^{\tss} \rightarrow S(a^{\ts}\otimes B,M)\bigr)$.
``$\leq$'' is true iff for each $a\in L$ we have $S(a^{\ts}\otimes A,M)^{\tss}\otimes ||A\Rightarrow B||_M\leq S(a^{\ts}\otimes B,M)$ and since $S(a^{\ts}\otimes B,M)=a^\ts\rightarrow S(B,M)$, the latter inequality is equivalent to
\[
  a^{\ts}\otimes S(a^{\ts}\otimes A,M)^{\tss}\otimes ||A\Rightarrow B||_M\leq S(B,M)
\]
which is true. Indeed, 
\begin{eqnarray*}
    &&
    \kern-1.2cm
a^{\ts}\otimes S(a^{\ts}\otimes A,M)^{\tss}\otimes ||A\Rightarrow B||_M\leq 
a^{\ts}\otimes S(a^{\ts}\otimes A,M)^{\ts}\otimes ||A\Rightarrow B||_M= \\
  &=&
a^{\ts}\otimes (a^{\ts}\rightarrow S(A,M))^{\ts}\otimes ||A\Rightarrow B||_M\leq 
a^{\ts}\otimes (a^{\ts}\rightarrow S(A,M)^{\ts})\otimes ||A\Rightarrow B||_M\leq \\
  &\leq&
 S(A,M)^{\ts}\otimes (S(A,M)^{\ts}\rightarrow S(B,M))\leq S(B,M).
\end{eqnarray*}
To check ``$\geq$'', observe that 
\begin{eqnarray*}
    &&
    \kern-1.2cm
   \textstyle
   \bigwedge_{a\in L} \bigl(S(a^{\ts}\otimes A,M)^{\tss} \rightarrow S(a^{\ts}\otimes B,M)\bigr) \leq \text{ (put $a=S(A,M)$)}\\
    &\leq& S(S(A,M)^{\ts}\otimes A,M)^{\tss} \rightarrow S(S(A,M)^{\ts}\otimes B,M)=\\
   &=&1^{\tss} \rightarrow S(S(A,M)^{\ts}\otimes B,M)= S(A,M)^{\ts}\to S(B,M)=||A\Rightarrow B||_M.
\end{eqnarray*}
\end{proof}


\begin{theorem}
\label{thm:dvai}
For a data table $\langle X,Y,I\rangle$ with grades,
hedges $\tss$ and $\ts$ with $\tss\leq\ts$, and a graded attribute attribute implication $A\Rightarrow B$,
\begin{eqnarray}
&&  \nonumber   \kern-1.2cm
 ||A \Rightarrow B||_{\langle X,Y,I\rangle} = \\
&& \label{eqn:alter1}
 \textstyle\bigwedge_{a\in L}
  ||a^\ts\otimes A\Rightarrow a^\ts\otimes B||^\tss_{\langle X,Y,I\rangle} = \\
&& \label{eqn:alter3}
   \textstyle \bigwedge_{a\in L} ||a^\ts\otimes A\Rightarrow a^\ts\otimes B||_{\langle X,Y,I\rangle} =\\
&& \label{eqn:alter2}
  \textstyle \bigwedge_{a\in L} ||A\Rightarrow B ||^\tss_{\tu{X,Y,a^{\ast}\to I}} = \\
&& \label{eqn:alter4}
   \textstyle ||A\Rightarrow B ||^\tss_{\mathrm{Int}(X^{\ast},Y,I)}.
\end{eqnarray}  
%
\end{theorem}
\begin{proof}
(\ref{eqn:alter1}):
Since $\{\deg{1}{x}\}^\up=I_x$ and
\[\textstyle 
    ||a^\ts\otimes A\Rightarrow a^\ts\otimes B||^\tss_{\langle X,Y,I\rangle} =  
    \bigwedge_{x\in X}
    \bigl(S(a^{\ts}\otimes A,\{\deg{1}{x}\}^\up)^{\tss} \rightarrow   S(a^{\ts}\otimes B,\{\deg{1}{x}\}^\up)\bigr),
\]
the fact that 
$||A \Rightarrow B||_{\langle X,Y,I\rangle}$ equals (\ref{eqn:alter1})
 follows directly from Lemma~\ref{thm:lemma2} and the definition of 
$||A \Rightarrow B||_{\langle X,Y,I\rangle}$.

(\ref{eqn:alter3}): The expression is a particular case of (\ref{eqn:alter1}) for
$\bullet=\ast$.

(\ref{eqn:alter2}):
Since 
\begin{eqnarray*}
  &&\textstyle 
  S(a^{\ts}\otimes C,I_x) =
  \bigwedge_{y\in Y} ((a^\ts\otimes C(y))\to I(x,y) ) =
     \bigwedge_{y\in Y} (C(y)\to (a^\ts \to I(x,y) ) ) =\\
     &=& \textstyle  S(C,a^{\ast}\to I_x) = S(C,(a^{\ast}\to I)_x),
\end{eqnarray*}
where $a^{\ast}\to I$ is the $a^{\ast}$-shift of $I$, i.e. $(a^{\ast}\to I)(x,y)=a^{\ast}\to I(x,y)$,
we get
\begin{eqnarray*}
  &&\textstyle
  \bigwedge_{a\in L}
  ||a^\ts\otimes A\Rightarrow a^\ts\otimes B||^\tss_{\langle X,Y,I\rangle} =
   \bigwedge_{a\in L,x\in X} \bigl(S(a^\ast\otimes A,I_x)^\tss \to S(a^\ast\otimes B,I_x)\bigr) =\\
  &=&\textstyle
   \bigwedge_{a\in L,x\in X} \bigl(S(A,(a^\ast\to I)_x)^\tss \to S(B,(a^\ast\to I)_x)\bigr) =
   \bigwedge_{a\in L} ||A\Rightarrow B||^\tss_{\tu{X,Y,a^\ast\to I}}, 
\end{eqnarray*}
proving that (\ref{eqn:alter2}) equals (\ref{eqn:alter1}).

%

(\ref{eqn:alter4}):
In view of Theorem \ref{thm:alter},
to prove $||A\Rightarrow B||_{\tu{X,Y,I}}\leq ||A\Rightarrow B||^\tss_{\mathrm{Int}(X^\ast,Y,I)}$ 
 it suffices to check
$||A\Rightarrow B ||_{\mathrm{Int}(X^{\ast},Y,I)} \leq   
||A\Rightarrow B ||^\tss_{\mathrm{Int}(X^{\ast},Y,I)}$, 
which  follows from $S(A,M)^{\ts}\rightarrow S(B,M)\leq S(A,M)^{\tss}\rightarrow S(B,M)$.
Conversely, since 
\begin{align*}
 &\textstyle ||A\Rightarrow B||_{\tu{X,Y,I}} = 
   \bigwedge_{a\in L} ||A\Rightarrow B||^\tss_{\tu{X,Y,a^\ast\to I}} =
  \textstyle\bigwedge_{x\in X,a\in L} \bigl(S(A,(a\to I)_x)^{\tss} \rightarrow S(B,(a\to I)_x)\bigr) = \\
  &=
\textstyle \bigwedge_{x\in X,a\in L} \bigl(S(A,\{\deg{a}{x}\}^\up)^{\tss} \rightarrow S(B,\{\deg{a}{x}\}^\up)\bigr),
\end{align*}
the inequality 
$||A\Rightarrow B||^\tss_{\mathrm{Int}(X^\ast,Y,I)}\leq ||A\Rightarrow B||_{\tu{X,Y,I}}$
is equivalent to 
%
\[\textstyle
  \bigwedge_{M\in\mathrm{Int}(X^{\ast},Y,I)}\bigl(S(A,M)^{\tss}\rightarrow S(B,M)\bigr)  \leq \bigwedge_{x\in X,a\in L} \bigl(S(A,\{\deg{a}{x}\}^\up)^{\tss} \rightarrow S(B,\{\deg{a}{x}\}^\up)\bigr),
\]
which follows from the fact that $\{\deg{a}{x}\}^\up\in\mathrm{Int}(X^{\ast},Y,I)$  for each $a\in L$ and $x\in X$.
\end{proof}

\begin{remark} 
(1) We encounter (\ref{eqn:alter1}) in Section \ref{sec:fal} where we prove completeness of certain Armstrong-like rules
for graded attribute implications.

(2) The hedge $\tss$ of Theorem~\ref{thm:dvai} can range in the sense of (\ref{eqn:pohedge}) arbitrarily from 
globalization, which is the least hedge, up to $\ts$ (boundary condition Theorem~\ref{thm:dvai}). 
In particular, with $\tss$ being the globalization, Theorem~\ref{thm:dvai} says that globalization
is in a sense, the basic hedge since the degree of validity of $A\Rightarrow B$ based on a general hedge
$\ts$ may be expressed as a degree of validity of $A\Rightarrow B$ that is based on globalization.
%
%
\end{remark}

%
%
%

\section{Logic of Graded Attribute Implications}\label{sec:fal}

In this section, we introduce a system for reasoning with graded attribute implications
and prove two versions of completeness for this system. In Section \ref{sec:arc}, 
we prove the ordinary-style completeness, i.e. we prove  that a graded attribute implication 
$A\Rightarrow B$ is provable from a set $T$ of 
implications if{}f the degree $||A\Rightarrow B||_T$ to which  $A\Rightarrow B$ semantically follows
 from $T$ equals~$1$.
In Section \ref{sec:gsc}, we present a graded-style completeness theorem.
Namely, we introduce  the concept of a degree $|A\Rightarrow B|_T$ of provability
of an implication $A\Rightarrow B$ from an $\mathbf{L}$-set $T$ of
implications and show that $|A\Rightarrow B|_T=||A\Rightarrow B||_T$,
i.e. the degree of provability coincides with the degree of semantic entailment. 

\subsection{Armstrong-like rules and ordinary-style completeness} \label{sec:arc}


\def\axA{Ax}
\def\axB{Cut}
\def\axC{Mul} 

Our axiomatic system consists of the following Armstrong-like \emph{deduction rules} \cite{Arm:Dsdbr}.\\[4pt]
\makebox[2.7em][r]{\textup{(\axA)}}~~%
infer $A \cup B \Rightarrow A$, \\[2pt]
\makebox[2.7em][r]{\textup{(\axB)}}~~%
from $A \Rightarrow B$ and $B\cup C\Rightarrow D$
infer $A\cup C \Rightarrow D$, \\[2pt]
\makebox[2.7em][r]{\textup{(\axC)}}~~%
from $A \Rightarrow B$  
infer $c^{\ast}\otimes A \Rightarrow c^{\ast}\otimes B $, \\[4pt]
for each $A,B,C,D \in \mathbf{L}^Y$, and $c \in L$.

\begin{remark}\label{rem:abc}
  (1)
  Rules (\axA) and (\axB) are inspired by the well-known ordinary rules
  of axiom and cut from which they differ in that
  $A,B,C,D$ represent $\mathbf{L}$-sets.
  
  (2)
  Rule (\axC), the rule of multiplication, is a new rule. Note that 
  $c^{\ast}\otimes A$ is defined by $(c^{\ast}\otimes A)(y)=c^{\ast}\otimes A(y)$.
  If ${}^\ast$ is globalization, (\axC) can be omitted. Indeed, for
  $c=1$, we have $c^\ast=1$ and (\axC) becomes ``from $A \Rightarrow B$ infer
  $A \Rightarrow B$'', a trivial rule. For $c<1$, we have $c^\ast=0$
  and (\axC) becomes ``from $A \Rightarrow B$ infer
  $\emptyset\Rightarrow\emptyset$'' which can be omitted since
  $\emptyset\Rightarrow\emptyset$ can be inferred by (\axA).
\end{remark}

Provability is defined as usual:
A graded attribute implication $A \Rightarrow B$ is called \emph{provable}
from a set $T$ of implications using a set $\cal R$ of
deduction rules, written $T \vdash_{\cal R} A \Rightarrow B$,
if there is a sequence (a proof) $\varphi_1,\dots,\varphi_n$ of 
implications such that $\varphi_n$ is $A \Rightarrow B$ and for each
$\varphi_i$ we either have $\varphi_i \in T$ or $\varphi_i$ is inferred
 from some of the preceding formulas 
using some deduction rule from $\cal R$.
If $\cal R$ consists of (\axA)--(\axC), 
we usually omit $\cal R$, and use, for instance, $T \vdash A \Rightarrow B$ instead of $T \vdash_{\cal R} A \Rightarrow B$.

A deduction rule ``from $\varphi_1,\dots,\varphi_n$ infer $\varphi$'' ,
with graded attribute implications $\varphi_1,\dots,\varphi_n,\varphi$,
is derivable from a set $\cal R$ of deduction rules if 
$\{\varphi_1,\dots,\varphi_n\}\vdash_{\cal R} \varphi$.
The following lemma presents some derived rules (one easily checks that
the arguments from the ordinary case apply~\cite{Mai:TRD}).

\begin{lemma}\label{thm:derArm}
The following deduction rules are derivable
from \textup{(\axA)} and \textup{(\axB):}\\[4pt]
\makebox[2.7em][r]{\textup{(Ref)}}~~%
 infer $A \Rightarrow A$, \\[2pt]
\makebox[2.7em][r]{\textup{(Wea)}}~~%
 from $A \Rightarrow B$
 infer $A \cup C \Rightarrow B$, \\[2pt]
\makebox[2.7em][r]{\textup{(Add)}}~~%
 from $A \Rightarrow B$ and $A \Rightarrow C$
 infer $A \Rightarrow B\cup C$, \\[2pt]
\makebox[2.7em][r]{\textup{(Pro)}}~~%
 from $A \Rightarrow B\cup C$
 infer $A \Rightarrow B$, \\[2pt]
\makebox[2.7em][r]{\textup{(Tra)}}~~%
 from $A \Rightarrow B$ and $B \Rightarrow C$
 infer $A \Rightarrow C$, \\[4pt]
for each $A,B,C,D \in \mathbf{L}^Y$.
\end{lemma}
%



A deduction rule ``from $\varphi_1,\dots,\varphi_n$ infer $\varphi$''
is  \emph{sound} if 
$\mathrm{Mod}(\{\varphi_1,\dots,\varphi_n\})\subseteq \mathrm{Mod}(\{\varphi\})$,
i.e. every model of all of $\varphi_1,\dots,\varphi_n$ is a model of $\varphi$.

\begin{lemma}\label{thm:sound}
  Each of the deduction rules \textup{(\axA)--(\axC)} is sound.
\end{lemma}
\begin{proof}
  The soundness of (\axA) is trivial as $||A\cup B\Rightarrow A||_M=1$ holds for any $M\in\mathbf{L}^Y$.
  If $M\in\mathrm{Mod}(\{A \Rightarrow B,B\cup C\Rightarrow D\})$, i.e.
  $S(A,M)^\ast\leq S(B,M)$ and $S(B\cup C,M)^\ast\leq S(D,M)$, then using $S(P\cup Q,R)=S(P,R)\wedge S(Q,R)$
  and $(a\wedge b)^\ast\leq a^\ast\wedge b^\ast$ we get
  \begin{eqnarray*}
  S(A\cup C,M)^\ast &=& (S(A,M)\wedge S(C,M))^\ast=(S(A,M)\wedge S(C,M))^{\ast\ast}\leq\\
  &\leq&  (S(A,M)^\ast\wedge S(C,M)^\ast)^{\ast}\leq (S(B,M)\wedge S(C,M))^{\ast}=
   S(B\cup C,M)^{\ast}\leq S(D,M), 
  \end{eqnarray*}
proving the soundness of (\axB).
  Let $||A\Rightarrow B||_M=1$. Putting $X=\{x\}$ and $I(x,y)=M(y)$
  for each $y\in Y$, (\ref{eqn:alter1}) yields 
  $||A\Rightarrow B||_M\leq ||a^\ast\otimes A\Rightarrow a^\ast\otimes  B||_M$,
  and hence $||a^\ast\otimes A\Rightarrow a^\ast\otimes  B||_M=1$, for every $a\in L$,
  proving the soundness of (\axC).
\end{proof}

A set $T$ of graded attribute implications is 
\begin{itemize}
  \item[--]
   \emph{syntactically closed} if  for every $A\Rightarrow B$, $T \vdash A \Rightarrow B$ implies 
$A \Rightarrow B \in T$, 
  \item[--]
   \emph{semantically closed} if  for every $A\Rightarrow B$, $|| A \Rightarrow B||_T=1$ implies 
$A \Rightarrow B \in T$. 
\end{itemize}

Clearly, $T$ is syntantically closed if{}f
$T=\{A\Rightarrow B \,|\, T \vdash A \Rightarrow B\}$;
analogously, $T$ is semantically closed if{}f
$T=\{A\Rightarrow B \,|\, ||A \Rightarrow B||_T=1\}$. 


\begin{lemma}\label{thm:semsyn}
  Let $T$ be a set of graded attribute implications.
  If $T$ is semantically closed then $T$ is syntactically closed.
\end{lemma}
\begin{proof}
 Follows from Lemma \ref{thm:sound} by standard arguments.
\end{proof}

\begin{lemma}\label{thm:synsem}
  Let $T$ be a set of graded attribute implications and let both $Y$ and $L$
  be finite. If $T$ is syntactically closed then $T$ is semantically closed.
\end{lemma}
\begin{proof}
  Let $T$ be syntactically closed. We need to show that if $||A\Rightarrow B||_T=1$
  then $A\Rightarrow B\in T$. We prove this by verifying that if 
  $A\Rightarrow B\not\in T$ then $||A\Rightarrow B||_T\not=1$.  
  Let thus $A\Rightarrow B\not\in T$.
  Note that since $T$ is syntactically closed,
  $T$ is closed w.r.t. the rules (Ref)--(Tra) of Lemma~\ref{thm:derArm}.

  To see that $||A\Rightarrow B||_T\not=1$, we show that there exists
  a model of $T$ that is not a model of $A\Rightarrow B$.
  For this purpose, consider $M=A^+$ where $A^+$ is the largest $\mathbf{L}$-set 
  such that $A\Rightarrow A^+\in T$. Note that $A^+$ exists. Namely,
  $S=\{C\,|\, A\Rightarrow C\in T\}$ is non-empty since $A\Rightarrow A\in T$
  by (Ref),
  $S$ is finite by finiteness of $Y$ and $L$,
  and for $A\Rightarrow C_1,\dots,A\Rightarrow C_n\in T$, we have
  $A\Rightarrow\bigcup_{i=1}^n C_i\in T$ by a repeated use of (Add).

  We now check that (a) $A^+$ is a model of $T$ and that
  (b) $A^+$ is not a model of $A\Rightarrow B$.

  (a):
  Let $C\Rightarrow D\in T$. We need to show $||C\Rightarrow D||_{A^+}=1$,
  i.e. $S(C,A^+)^\ast \rightarrow S(D,A^+)=1$ which is equivalent to
  $S(C,A^+)^\ast\otimes D\subseteq A^+$ due to adjointness of $\otimes$ and $\to$. 
  Since $A^+$ is the
  largest one for which $A\Rightarrow A^+\in T$, in order 
  to verify $S(C,A^+)^\ast\otimes D\subseteq A^+$, it is sufficient to 
  show that $A\Rightarrow S(C,A^+)^\ast\otimes D\in T$.
  We claim  (a1) $A\Rightarrow A^+\in T$,
  (a2) $A^+\Rightarrow S(C,A^+)^\ast\otimes C\in T$, and (a3)
  $S(C,A^+)^\ast\otimes C\Rightarrow S(C,A^+)^\ast\otimes D\in T$.
  Indeed, $A\Rightarrow A^+\in T$ by definition of $A^+$.
  $A^+\Rightarrow S(C,A^+)^\ast\otimes C\in T$ is an instance of (\axA)
  because $S(C,A^+)^\ast\otimes C \subseteq A^+$,
  which follows from
  \begin{eqnarray*}
   (S(C,A^+)^\ast\otimes C)  (y)&\leq& \textstyle C(y)\otimes S(C,A^+) =
   C(y) \otimes \bigwedge_{z\in Y}(C(z)\to A^+(z)) \leq\\
   &\leq& C(y) \otimes C(y)\to A^+(y) \leq A^+(y).
  \end{eqnarray*}
   Finally,   $S(C,A^+)^\ast\otimes C\Rightarrow S(C,A^+)^\ast\otimes D\in T$ by applying (\axC)
  to $C\Rightarrow D\in T$.
  Now,  (Tra) applied to (a1), (a2), and (a3) yields 
  $A\Rightarrow S(C,A^+)^\ast\otimes D\in T$, proving (a).

  (b):
  We need to show  $||A\Rightarrow B||_{A^+}\not=1$. 
  Note that 
  \[
     ||A\Rightarrow B||_{A^+}=S(A,A^+)^\ast \rightarrow S(B,A^+)=1\rightarrow S(B,A^+)=S(B,A^+).
  \]
  Therefore, if $||A\Rightarrow B||_{A^+}=1$, one has 
  $1=S(B,A^+)$, whence $B\subseteq A^+$.
  Since $A\Rightarrow A^+\in T$, (Pro) would give $A\Rightarrow B\in T$, a
  contradiction to the assumption.
  \end{proof}


The following theorem is the ordinary-style completeness theorem of (\axA)--(\axC).

\begin{theorem}
  \label{Th:Comp}
  Let\/ $L$ and $Y$ be finite. For a set $T$ be
  of graded attribute implications and a graded attribute implication $A\Rightarrow B$,\\[4pt]
  \centerline{$T \vdash A \Rightarrow B$
    \quad if{}f \quad
    $||A \Rightarrow B||_T = 1$.}
\end{theorem}
\begin{proof}
  Let 
  $\mathit{syn}(T)=\{A\Rightarrow B\,|\, T\vdash A\Rightarrow B\}$ and
  $\mathit{sem}(T)=\{A\Rightarrow B\,|\, ||A\Rightarrow B||_T=1\}$.
    We need to show $\mathit{syn}(T)=\mathit{sem}(T)$.
  One easily checks that 
  $\mathit{syn}(T)$ and $\mathit{sem}(T)$ are the least syntactically and semantically 
  closed sets of graded attribute implications that contain $T$, respectively. 
  As $\mathit{syn}(T)$ is syntactically closed, it is also semantically closed
  by Lemma~\ref{thm:synsem} which means
  $\mathit{sem}(\mathit{syn}(T))= \mathit{syn}(T)$. Therefore, since
  $T\subseteq\mathit{syn}(T)$, we get
  \[
    \mathit{sem}(T)\subseteq\mathit{sem}(\mathit{syn}(T))=
    \mathit{syn}(T).
  \]
  In a similar manner, using Lemma \ref{thm:semsyn}, we get
  $\mathit{syn}(T)\subseteq\mathit{sem}(T)$,
  showing $\mathit{syn}(T)=\mathit{sem}(T)$.
\end{proof}

Note that as is well-known,  (\axA) and (\axB) form a syntactico-semantically
complete system in the ordinary case (i.e. with fuzzy sets replaced by 
ordinary sets). 
The system consisting of  (\axA), (\axB), and (\axC) results by adding
a new rule, (\axC), to a (\axA) and (\axB). 
In this perspective,  (\axC) is the rule that handles intermediate degrees.
Alternatively, one could modify (\axB) and use

\medskip\noindent
\makebox[2.7em][r]{\textup{(Cut')}}~~%
from $A \Rightarrow e \otimes B$ and
$B \cup C \Rightarrow D$ infer
$A \cup C \Rightarrow e^{\ast} \otimes D$ \\[4pt]
instead of adding (\axC). Namely:

\begin{lemma}
  \textup{(\axA), (\axB)}, and \textup{(\axC)} are equivalent to \textup{(\axA)} and \textup{(Cut')}.
\end{lemma}
\begin{proof}
  First, we show that (Cut') is derivable from (\axA), (\axB), and (\axC).
  Let $\vdash A \Rightarrow e \otimes B$ and
  $\vdash B \cup C \Rightarrow D$. Then 
  $\vdash e^{\ast}\otimes(B \cup C)\Rightarrow e^{\ast} \otimes D$ by (\axC),
  $\vdash (e \otimes B) \cup C \Rightarrow e^{\ast} \otimes D$ by (Wea), and
  $\vdash A \cup C \Rightarrow e^{\ast} \otimes D$ by (Cut).

  Conversely, since (\axB) is an instance of (Cut') for $e=1$, it suffices to show that
   (\axC) is derivable from (\axA) and (Cut').
%
  Since $c^{\ast} \otimes A \Rightarrow c^{\ast} \otimes A$ is an instance
  of (Ax'), we get $c^{\ast} \otimes A \Rightarrow c^{\ast\ast} \otimes B$
  by (Cut') applied on 
  $c^{\ast} \otimes A \Rightarrow c^{\ast} \otimes A$ and $A \Rightarrow B$;
  (\ref{TS:ID}) gives that
  $c^{\ast} \otimes A \Rightarrow c^{\ast\ast} \otimes B$ equals
  $c^{\ast} \otimes A \Rightarrow c^{\ast} \otimes B$ which is the
  desired formula.
\end{proof}

In the setting which involves grades, the degree $S(A,B)$ of inclusion of 
the $\mathbf{L}$-set $A$ in the $\mathbf{L}$-set $B$, as defined by (\ref{Eq:S}) 
is an important concept generalizing the classical set inclusion.
Another one, generalizing set equality, is the degree $A\approx B$
of equality of $A$ and $B$, defined by
\[
 A\approx B =\textstyle\bigwedge_{y\in Y} \bigl(A(y)\leftrightarrow B(y)\bigr),
\]
where $a\leftrightarrow b=(a\rightarrow b)\wedge(b\rightarrow a)$ is the biresiduum
of $a$ and $b$. 
Note that $A\approx B$ is a truth degree of the proposition ``for each $y\in Y$:
$y$ belongs to $A$ if{}f $y$ belongs to $B$'' and that $\approx$ is an 
$\mathbf{L}$-equivalence relation~\cite{Bel:FRS,Got:TMVL,Haj:MFL}.
Therefore, $A\approx B$ may be conceived as a degree of similarity of $A$ and $B$.
Both $S(A,B)$ and $A\approx B$ naturally enter derived rules, as illustrated
by the following lemma.

\begin{lemma}\label{thm:derRules}
  The following deduction rules are derivable
  from \textup{(\axA)}--\textup{(\axC):} \\[4pt]
  \makebox[2.7em][r]{\textup{(S)}}~~%
   from $A\Rightarrow B$ infer $C \Rightarrow S(A,C)^\ast\otimes B$, \\[2pt]
  \makebox[2.7em][r]{\textup{(Sub)}}~~%
   from $A\Rightarrow B$ infer $C \Rightarrow D\otimes 
   S(A,C)^\ast\otimes S(D,B)$, \\[2pt]
  \makebox[2.7em][r]{\textup{(Sim)}}~~%
   from $A\Rightarrow B$ infer $C \Rightarrow D\otimes 
   (A\approx C)^\ast\otimes (D\approx B)$, \\[4pt]
  for each $A,B,C,D \in \mathbf{L}^Y$.
\end{lemma}
\begin{proof}
   One may easily check that (S), (Sub), and (Sim) are all sound rules.
  The assertion then follows from completeness of (\axA)--(\axC).
\end{proof}

\subsection{Graded-style completeness}
\label{sec:gsc}

The ordinary-style completeness does not capture semantic entailment
to its full extent in that it only provides a syntactic characterization
of entailment to degree $1$. In this section, we provide a graded-style
completeness theorem which extends to general degrees of entailment.
Note that the graded-style completeness goes back to Pavelka's seminal 
work~\cite{Pav:Ofl} and is further worked out, e.g., in \cite{Ger:FL,Haj:MFL}.
It is based on Goguen's idea \cite{Gog:Lic} of a  proof as a sequence of
weighted formulas, i.e. pairs $\tu{\varphi,a}$ where $\varphi$ is a formula
and $a$ a truth degree to which $\varphi$ has been inferred
using deduction rules that operate on both formulas and truth degrees.
Graded-style completeness then says that the possibly intermediate degree 
of entailment of $\varphi$ equals the degree of provability of $\varphi$, 
i.e. the supremum of $a$s that appear in $\tu{\varphi,a}$ at the end of
proofs. 

In our treatment, this agenda is not employed. Instead, we utilize 
Corollary~\ref{thm:reducT} and ordinary-style completeness and define
the notion of of provability degree accordingly.
Namely, for a fuzzy set $T$ of graded attribute implications and for $A \Rightarrow B$
we define the \emph{degree $|A \Rightarrow B|_T \in L$ to which
  $A \Rightarrow B$ is provable from $T$} by
\begin{eqnarray}
  \label{eqn:provdeg}
  |A \Rightarrow B|_T = \textstyle\bigvee\{c \in L \,|\,
  \mathrm{cr}(T) \vdash A \Rightarrow c \otimes B\},
\end{eqnarray}
where $\mathrm{cr}(T)$ is defined by (\ref{eqn:cT}).
Alternatively, $|A \Rightarrow B|_T$ may be defined as the largest $c$
for which $\mathrm{cr}(T) \vdash A \Rightarrow c \otimes B$:

\begin{lemma}
The set 
$\{c \in L \,|\,  \mathrm{cr}(T) \vdash A \Rightarrow c \otimes B\}$ 
 in (\ref{eqn:provdeg}) has a largest element.
\end{lemma}
\begin{proof}
   Due to Theorem \ref{Th:Comp}, we prove the claim by checking  that 
   if $||A\Rightarrow c_k\otimes B||_{\mathrm{cr}(T)}=1$ for $k\in K$,
   then $||A\Rightarrow (\bigvee_{\!k\in K}c_k)\otimes B||_{\mathrm{cr}(T)}=1$.
   Observe that $||A\Rightarrow (\bigvee_{\!k\in K}c_k)\otimes B||_{\mathrm{cr}(T)}=1$
   means that for each model $M\in\mathrm{Mod}(\mathrm{cr}(T))$, we have
   $||A\Rightarrow (\bigvee_{\!k\in K}c_k)\otimes B||_M=1$, i.e.
   $S(A,M)^\ast\leq S((\bigvee_{\!k\in K}c_k)\otimes B,M)=
   \bigwedge_{y\in Y} ((B(y) \otimes \bigvee_{\!k\in K}c_k) \to M(y))$,
   which holds iff for every $y\in Y$, 
  $S(A,M)^\ast\leq (B(y) \otimes \bigvee_{\!k\in K}c_k) \to M(y)$.
  Due to adjointness and the distributiviy of $\otimes$ over $\bigvee$, 
  the last inequality is equivalent to 
  $\bigvee_{\!k\in K}(c_k\otimes B(y) \otimes S(A,M)^\ast)\leq  M(y)$
  which holds if for each $k\in K$,
  $c_k\otimes B(y) \otimes S(A,M)^\ast\leq  M(y)$. This is equivalent to
  $S(A,M)^\ast\leq (c_k\otimes B(y))\to M(y)$.
  Now, the last inequality holds for every $y\in Y$ if{}f
  $S(A,M)^\ast\leq S(c_k\otimes B, M)$ which is true due to the assumptions
  $||A\Rightarrow c_k\otimes B||_{\mathrm{cr}(T)}=1$ and $M\in\mathrm{Mod}(\mathrm{cr}(T))$.
\end{proof}

We now have:

\begin{theorem}
  Let\/ $\mathbf{L}$ and $Y$ be finite. Then for every fuzzy set $T$ of
  fuzzy attribute implications and $A \Rightarrow B$ we have
  $$|A \Rightarrow B|_T = ||A \Rightarrow B||_T.$$
\end{theorem}
\begin{proof}
  Consequence of Corollary~\ref{thm:reducT} and Theorem~\ref{Th:Comp}.
\end{proof}

\section{Bases of Graded Attribute Implications}
\label{sec:bai}


\subsection{Completeness, non-redundancy, bases}
\label{sec:cnrb}

When exploring graded attribute implications of a table $\tu{X,Y,I}$ with grades, one is
interested in small informative sets of implications.
A reasonable approach is to require, on one hand, that such set contains
information 
about validity in $\tu{X,Y,I}$ of all implications and, on the other hand, is
non-redundant. Such sets are investigated in this section.

\begin{definition}\label{def:base}
  A set $T$ of graded attribute implications is called \emph{complete in} $\tu{X,Y,I}$ if
  \begin{eqnarray}\label{eqn:base}
      ||A \Rightarrow B||_T = ||A \Rightarrow B||_{\langle X,Y,I\rangle}
  \end{eqnarray}
  for every implication $A\Rightarrow B$. 
\end{definition}

\begin{remark}
  (1)
   That is,
  $T$ is complete if the degree of entailment by $T$ coincides with the degree
  of validity in $\tu{X,Y,I}$. In this sense, a complete set contains all information about validity in $\tu{X,Y,I}$
   via semantic entailment.

  (2)
  Every $A\Rightarrow B$ from a complete $T$ is valid in  $\tu{X,Y,I}$ to degree $1$.
  This is a direct consequence of (\ref{eqn:base}) and the fact that $||A \Rightarrow B||_T=1$
  for $A\Rightarrow B\in T$.
%
\end{remark}

Completeness of $T$ may be characterized in terms of models  of $T$ the following way:

\begin{theorem}\label{thm:basechar}
$T$ is complete in $\tu{X,Y,I}$ if{}f
$\mathrm{Mod}(T) = \mathrm{Int}(X^{\ast},Y,I)$.
\end{theorem}
\begin{proof}
  Let $T$ be complete in $\tu{X,Y,I}$.
  Let first $M \in \mathrm{Mod}(T)$.
  Due to completeness of $T$ and Theorem \ref{thm:alter}, 
     $$||M \Rightarrow M^{\down\up}||_T =
    ||M \Rightarrow M^{\down\up}||_{\tu{X,Y,I}} = S(M^{\down\up},M^{\down\up}) = 1.$$
    As $M \in \mathrm{Mod}(T)$, $||M \Rightarrow M^{\down\up}||_T =1$ 
   yields $||M \Rightarrow M^{\down\up}||_M = 1$ from which it follows
   $1 = S(M,M)^{\ast} \leq S(M^{\down\up},M)$,
  i.e. $M^{\down\up} \subseteq M$. 
  Since $M\subseteq M^{\down\up}$ is always the case, we have
  $M= M^{\down\up}$, hence 
   $M \in \mathrm{Int}(X^{\ast},Y,I)$ by virtue of (\ref{eqn:Mint}).
  We proved $\mathrm{Mod}(T) \subseteq \mathrm{Int}(X^{\ast},Y,I)$.
  Conversely, let $M \in \mathrm{Int}(X^{\ast},Y,I)$.
  Clearly, since $T$ is complete, Theorem \ref{thm:alter} implies
  \[
       ||A \Rightarrow B||_M \geq ||A \Rightarrow B||_{\mathrm{Int}(X^{\ast},Y,I)} = ||A\Rightarrow B||_T
   \]
  for every implication $A\Rightarrow B$.
  In particular, if $A\Rightarrow B\in T$ then $||A\Rightarrow B||_T=1$ and the last inequality
  yields 
    $||A \Rightarrow B||_M=1$. This shows $M \in \mathrm{Mod}(T)$ and thus
  $\mathrm{Int}(X^{\ast},Y,I) \subseteq \mathrm{Mod}(T)$.

  Conversely, if $\mathrm{Mod}(T) = \mathrm{Int}(X^{\ast},Y,I)$ then $T$ is complete since due to
  Theorem \ref{thm:alter},
  $$||A \Rightarrow B||_T = ||A \Rightarrow B||_{\mathrm{Mod}(T)}=
  ||A \Rightarrow B||_{\mathrm{Int}(X^{\ast},Y,I)} =
  ||A \Rightarrow B||_{\langle X,Y,I\rangle}.$$
\end{proof}

\begin{remark}\label{rem:baseMod}
  (1)
   Each of the two  inclusions of $\mathrm{Mod}(T) = \mathrm{Int}(X^{\ast},Y,I)$ has
   a natural meaning. Namely,
  as an inspection of the proof of Theorem \ref{thm:basechar} shows, (a)
  $\mathrm{Mod}(T) \subseteq \mathrm{Int}(X^{\ast},Y,I)$ if and only if
   $ ||A \Rightarrow B||_T \geq ||A \Rightarrow B||_{\langle X,Y,I\rangle}$ for every $A\Rightarrow B$,
  and (b)
  $\mathrm{Mod}(T) \supseteq \mathrm{Int}(X^{\ast},Y,I)$ if and only if
   $ ||A \Rightarrow B||_T \leq ||A \Rightarrow B||_{\langle X,Y,I\rangle}$ for every $A\Rightarrow B$.
%
%
%
%
%

  (2)
  Let $\mathrm{Mod}(T)\supseteq \{I_x \mid x\in X\}$
  with $I_x$s given by (\ref{eqn:Ix}). Then 
  $ ||A \Rightarrow B||_T= ||A \Rightarrow B||_{\mathrm{Mod}(T)} \leq ||A \Rightarrow B||_{\langle X,Y,I\rangle}$.
  If, on the other hand, $ ||A \Rightarrow B||_T\leq ||A \Rightarrow B||_{\langle X,Y,I\rangle}$, then according
  to (1)(b), $\mathrm{Mod}(T) \supseteq \mathrm{Int}(X^{\ast},Y,I)$, whence also 
  $\mathrm{Mod}(T)\supseteq \{I_x \mid x\in X\}$ because  $\mathrm{Int}(X^{\ast},Y,I)\supseteq \{I_x \mid x\in X\}$.
  This shows that in (1)(b), the condition $\mathrm{Mod}(T) \supseteq \mathrm{Int}(X^{\ast},Y,I)$
  may be replaced by $\mathrm{Mod}(T)\supseteq \{I_x \mid x\in X\}$.
  Now, since $\mathrm{Mod}(T)\supseteq \{I_x \mid x\in X\}$ says that 
  every $A\Rightarrow B\in T$ is valid in $\tu{X,Y,I}$ to degree $1$, 
  we conclude that in order to check that $T$ is complete in $\tu{X,Y,I}$, it suffices to check that 
  every $A\Rightarrow B\in T$ be valid in $\tu{X,Y,I}$ to degree $1$ and that $\mathrm{Mod}(T) \subseteq \mathrm{Int}(X^{\ast},Y,I)$.
%
\end{remark}

\begin{definition}\label{def:nrbase}
  A set $T$ of graded implications is called a \emph{base of} $\tu{X,Y,I}$ 
  if $T$ is complete in $\tu{X,Y,I}$ and no proper subset of $T$
  is complete in $\tu{X,Y,I}$.
\end{definition}

Alternatively, one can define the notion of a  base the following way.
Call a set $T$ of implications \emph{redundant} if there exists $A \Rightarrow B \in T$
such that $||A \Rightarrow B||_{T - \{A \Rightarrow B\}} = 1$.
Otherwise, call $T$ non-redundant. 

\begin{lemma}\label{thm:nonr}
   The following conditions are equivalent:
   \bgroup%
   \addtolength{\leftmargini}{1.2em}
   \begin{itemize}
     \item[\rm (i)]
    $T$ is a non-redundant set of implications.
     \item[\rm (ii)]
     For every $A \Rightarrow B \in T$: $\mathrm{Mod}(T) \subset \mathrm{Mod}(T-\{A \Rightarrow B\})$.
     \item[\rm (iii)]
     For every $A \Rightarrow B \in T$ there exists $C\Rightarrow D$ such that 
     $||C\Rightarrow D||_{T-\{A \Rightarrow B\}} <||C\Rightarrow D||_{T}$.
   \end{itemize}
   \egroup%
\end{lemma}
\begin{proof}
   (i) $\Rightarrow$ (ii):
  As $\mathrm{Mod}(T) \subseteq \mathrm{Mod}(T-\{A \Rightarrow B\})$
   is always the case, it is sufficient to show that  
   $\mathrm{Mod}(T) \not= \mathrm{Mod}(T-\{A \Rightarrow B\})$.
   Indeed, if $\mathrm{Mod}(T) = \mathrm{Mod}(T-\{A \Rightarrow B\})$,
   then 
   $$||A\Rightarrow B||_{T-\{A \Rightarrow B\}}
  =||A\Rightarrow B||_{\mathrm{Mod}(T-\{A \Rightarrow B\})}
     =||A\Rightarrow B||_{\mathrm{Mod}(T)} 
=||A\Rightarrow B||_{T}=1,$$
   a contradiction to non-redundancy of $T$.

   (ii) $\Rightarrow$ (iii): 
   Due to (ii), there exists a model $M$ of $T-\{A \Rightarrow B\}$ which is not
   a model of $A\Rightarrow B$. Hence, putting $C\Rightarrow D=A\Rightarrow B$,
   we get $||A\Rightarrow B||_{T-\{A \Rightarrow B\}} <1=||A\Rightarrow B||_{T}$.

   (iii) $\Rightarrow$ (i): 
   If $T$ were redundant then for some $A \Rightarrow B \in T$ we have
   $||A\Rightarrow B||_{T-\{A \Rightarrow B\}} =1$ from which it follows
   $\mathrm{Mod}(T) \supseteq \mathrm{Mod}(T-\{A \Rightarrow B\})$. 
   Since the converse inclusion is obvious, we get
   $\mathrm{Mod}(T) = \mathrm{Mod}(T-\{A \Rightarrow B\})$. As a result,
    $||C\Rightarrow D||_{T-\{A \Rightarrow B\}} =||C\Rightarrow D||_{T}$ for each
   $C\Rightarrow D$, a contradiction to (iii).
\end{proof}


As the following theorem shows,  bases are just complete sets that are
non-redundant as sets of implications.

\begin{theorem}
  $T$ is a base of $\langle X,Y,I\rangle$ if and only if
  \bgroup%
  \addtolength{\leftmargini}{1.2ex}%
  \begin{itemize}
  \item
    $T$ is complete in $\langle X,Y,I\rangle$, and
  \item
    $T$ is non-redundant as a set of implications.
  \end{itemize}
  \egroup%
\end{theorem}
\begin{proof}
   The assertion follows directly from the fact that non-redundancy of $T$ is equivalent to 
   condition (iii) of Lemma \ref{thm:nonr} and the fact that
   if $T$ is complete in $\tu{X,Y,I}$, one has $||C\Rightarrow D||_{T}=||C\Rightarrow D||_{\tu{X,Y,I}}$.
%
\end{proof}

Since one is naturally interested in implications that are fully true in data, the following
concept is of interest. A set $T$ of graded attribute implications is called a 
\emph{$1$-complete in $\langle X,Y,I\rangle$} if 
\[
||A\Rightarrow B||_T=1  \text{ if{}f } ||A\Rightarrow B||_{\langle X,Y,I\rangle}=1
\]
for every implication $A\Rightarrow B$. That is,  full consequences of a $1$-complete set 
need to be just the implications fully true in $\tu{X,Y,I}$.
Clearly, completeness in $\tu{X,Y,I}$ implies $1$-completeness in $\tu{X,Y,I}$. 
Interestingly, we have:

\begin{theorem}\label{thm:1base}
   $T$ is complete in $\langle X,Y,I\rangle$ if{}f $T$ is $1$-complete in $\langle X,Y,I\rangle$.
\end{theorem}
\begin{proof}
  We need to show that if $T$ is a $1$-complete, it is complete. Let thus $T$ be 1-complete in
  $\tu{X,Y,I}$.
  Due to Theorem \ref{thm:alter}, $T$ is complete if{}f
   $||A\Rightarrow B||_{T}=||A\Rightarrow B||_{\mathrm{Int}(X^\ast,Y,I)}$, i.e. if{}f
  $||A\Rightarrow B||_{\mathrm{Mod}(T)}=||A\Rightarrow B||_{\mathrm{Int}(X^\ast,Y,I)}$ for every $A\Rightarrow B$,
   which we now verify. 
   First, observe that 
   \begin{equation}\label{eqn:1baseaux2}
       ||A\Rightarrow S(B,A^{\down\up})\otimes B||_{\mathrm{Int}(X^\ast,Y,I)}=1.
   \end{equation}
   Indeed, the equality holds if{}f for each $M\in \mathrm{Int}(X^\ast,Y,I)$, 
   \[
        S(A,M)^\ts \leq S(S(B,A^{\down\up})\otimes B, M).
   \]
   Since $S(S(B,A^{\down\up})\otimes B, M)=S(B,A^{\down\up}) \to S(B,M)$,
   the last inequality is equivalent to
   \begin{equation}\label{eqn:1baseaux1}
       S(B,A^{\down\up})\otimes  S(A,M)^\ts \leq S(B, M).
   \end{equation}
   Now, due to the properties of $\tu{{}^\up,{}^\down}$ established in \cite{Bel:Fgc},
   and due to $(\bigwedge_{k\in K} a_k)^\ts\leq \bigwedge_{k\in K}a_k^\ts$ and (\ref{TS:MP}), we get
   \begin{eqnarray*}
     && \kern-14mm \textstyle
     S(A,M)^\ts \leq S(A^\downts,M^\downts)^\ts =
     (\bigwedge_{x\in X}(A^\downts(x)\to M^\downts(x)))^\ts \leq
     \bigwedge_{x\in X}(A^\downts(x)^\ts\to M^\downts(x)^\ts) =\\
      &=& S({A^\downts}^\ts,{M^\downts}^\ts) \leq S(A^{\downts\ts\Uparrow},M^{\downts\ts\Uparrow})
      =S(A^{\downts\upts}, M^{\downts\upts}).
   \end{eqnarray*}
    Therefore, using $M\in \mathrm{Int}(X^\ast,Y,I)$ and thus $M=M^{\downts\upts}$, we get
   \begin{eqnarray*}
     S(B,A^{\down\up})\otimes  S(A,M)^\ts \leq 
     S(B,A^{\downts\upts})\otimes  S(A^{\downts\upts}, M^{\downts\upts}) \leq
     S(B,M^{\downts\upts}) = S(B,M),
   \end{eqnarray*}
    verifying (\ref{eqn:1baseaux1}) and thus also (\ref{eqn:1baseaux2}).
    Next,
    (\ref{eqn:1baseaux2}), Theorem \ref{thm:alter}, and the assumption that $T$ is
    1-complete yields
   \[
      1=||A\Rightarrow S(B,A^{\down\up})\otimes B||_{\mathrm{Int}(X^\ast,Y,I)} =||A\Rightarrow S(B,A^{\down\up})\otimes B||_{\mathrm{Mod}(T)}.
   \]
   Due to (ii) of Lemma~\ref{Th:Sem_graded} and (\ref{eqn:prlsiD}),
   \begin{eqnarray*}
     && \kern-14mm \textstyle
     ||A\Rightarrow S(B,A^{\down\up})\otimes B||_{\mathrm{Mod}(T)}=
      \bigwedge_{M\in \mathrm{Mod}(T)}  ||A\Rightarrow S(B,A^{\down\up})\otimes B||_M =\\
     &=& \textstyle
      \bigwedge_{M\in \mathrm{Mod}(T)} \bigl(  S(B,A^{\down\up})\rightarrow||A\Rightarrow B||_M\bigr) =
       S(B,A^{\down\up})\rightarrow  \bigwedge_{M\in \mathrm{Mod}(T)}  ||A\Rightarrow B||_M =\\
      &=& S(B,A^{\down\up})\rightarrow  ||A\Rightarrow B||_{\mathrm{Mod}(T)}.
   \end{eqnarray*}
    As a result,
  \[
      S(B,A^{\down\up})\rightarrow  ||A\Rightarrow B||_{\mathrm{Mod}(T)} =1,
  \]
  i.e. due to (\ref{eqn:prl1B}), $S(B,A^{\down\up})\leq  ||A\Rightarrow B||_{\mathrm{Mod}(T)}$.
  Since $S(B,A^{\down\up})=||A\Rightarrow B||_{\mathrm{Int}(X^\ast,Y,I)}$ due to Theorem \ref{thm:alter},
  we established $||A\Rightarrow B||_{\mathrm{Int}(X^\ast,Y,I)}\leq ||A\Rightarrow B||_{\mathrm{Mod}(T)}$. 
  The converse inequality, $||A\Rightarrow B||_{\mathrm{Int}(X^\ast,Y,I)}\geq ||A\Rightarrow B||_{\mathrm{Mod}(T)}$,
   follows directly from ${\mathrm{Int}(X^\ast,Y,I)}\subseteq\mathrm{Mod}(T)$,
  which is a consequence of $1$-completeness of $T$
  (the same argument as in the proof of Theorem \ref{thm:basechar} applies).
\end{proof}

\subsection{Bases and pseudo-intents}\label{sec:nrbpi}

A particular type of bases may be obtained from the following collections of $\mathbf{L}$-sets
of attributes.

\begin{definition}\label{def:pseudint}
A set ${\cal P} \subseteq \mathbf{L}^Y$ is called a \emph{system of pseudo-intents}
of $\langle X,Y,I\rangle$ if for each $P \in \mathbf{L}^Y$:
\begin{equation}\label{eqn:sp}
     \mbox{  $P \in {\cal P}$ \quad if{}f \quad $P \ne P^{\down\up}$ and
    $||Q \Rightarrow Q^{\down\up}||_P = 1$ 
    for each $Q \in {\cal P}$ with $Q \ne P$.}
\end{equation}
\end{definition}

\begin{remark} \label{rem:classicP}
(a)  Recall that $||Q \Rightarrow Q^{\down\up}||_P = 1$ means that 
  $S(Q,P)^\ast \leq S(Q^{\down\up},P)$. Hence, if ${}^\ast$ is the globalization,
  then since $S(Q,P)^\ast=1$ if $Q\subseteq P$ and $S(Q,P)^\ast=0$ if $Q\not\subseteq P$,
  condition (\ref{eqn:sp}) simplifies to
\begin{equation} \label{eqn:spG}
     \mbox{  $P \in {\cal P}$ \quad if{}f \quad $P \ne P^{\down\up}$ and
    $Q^{\down\up}\subset P$ 
    for each $Q \in {\cal P}$ with $Q \subset P$.}
\end{equation}
If $L$ is, moreover, finite then it is easily seen that there exists a unique system of pseudointents
of $\tu{X,Y,I}$.
In general, a system of pseudointents is not unique and may even not exist, as we demonstrate below.

(b)
For $L=\{0,1\}$, globalization is the only hedge ${}^\ast$ on $L$. 
One easily observes that in this case, (\ref{rem:classicP}) essentially coincides with the
definition of  a (unique) system of ordinary pseudo-intents \cite{G:BuI,GaWi:FCA,GuDu:Fmiirtdb}.
%
\end{remark}

The importance of the notion of a system of pseudointents derives from the following theorem.

\begin{theorem}\label{thm:nrbase}
  If ${\cal P}$ is a system of pseudo-intents of $\tu{X,Y,I}$ then
  \begin{eqnarray}
    T = \{P \Rightarrow P^{\down\up} \,|\, P \in {\cal P}\}
    \label{Def:T}
  \end{eqnarray}
  is a non-redundant base of $\tu{X,Y,I}$.
\end{theorem}
\begin{proof}
  First, we show that $T$ is complete in $\tu{X,Y,I}$. Due to 
  Theorem \ref{thm:basechar}, it is sufficient to show 
  $\mathrm{Mod}(T)={\mathrm{Int}(X^{\ast},Y,I)}$.
  To show $\mathrm{Mod}(T)\subseteq\mathrm{Int}(X^{\ast},Y,I)$, assume by contradiction
  that $M \in \mathrm{Mod}(T)-\mathrm{Int}(X^{\ast},Y,I)$.
  As $M\not\in \mathrm{Int}(X^{\ast},Y,I)$, we have $M \ne M^{\downarrow\uparrow}$ by (\ref{eqn:Mint}),
  thus in particular $S(M^{\downarrow\uparrow},M) \ne 1$.
   The assumption $M \in \mathrm{Mod}(T)$ yields that
    $||Q \Rightarrow Q^{\downarrow\uparrow}||_M = 1$ for every $Q \in {\cal P}$,
    which implies $M \in {\cal P}$ by Definition \ref{def:pseudint}. Hence,
    $M \Rightarrow M^{\downarrow\uparrow}$ belongs to $T$ and we have
    \begin{eqnarray*}
      ||M \Rightarrow M^{\downarrow\uparrow}||_M
      =
      S(M,M)^{\ast} \rightarrow S(M^{\downarrow\uparrow},M) 
     =
      1^{\ast} \rightarrow S(M^{\downarrow\uparrow},M) =
      S(M^{\downarrow\uparrow},M) \ne 1,
    \end{eqnarray*}
    which contradicts $M \in \mathrm{Mod}(T)$.
    To verify $\mathrm{Mod}(T)\supseteq\mathrm{Int}(X^{\ast},Y,I)$, 
    observe that if $M\in \mathrm{Int}(X^{\ast},Y,I)$ then since $M = M^{\downarrow\uparrow}$ by (\ref{eqn:Mint}),
    we have for every $P\in \mathbf{L}^Y$, thus in particular for every $P\in\mathcal{P}$,
    \begin{eqnarray*}
      S(P,M)^{\ast}=
      S(M^{\Downarrow},P^{\Downarrow})^{\ast} \leq
      S(M^{\Downarrow\ast},P^{\Downarrow\ast}) \leq 
      S(P^{\Downarrow\ast\Uparrow},M^{\Downarrow\ast\Uparrow}) =
      S(P^{\downarrow\uparrow},M^{\downarrow\uparrow}) =
      S(P^{\downarrow\uparrow},M),
    \end{eqnarray*}
     whence
     $$||P \Rightarrow P^{\downarrow\uparrow}||_M \leq 
     S(P,M)^{\ast} \rightarrow S(P^{\downarrow\uparrow},M) = 1,$$
     establishing that $M$ is a model of $T$.

     Second, we check that $T$ is  non-redundant.
    If $T' \subset T$, there exists 
    $P \in {\cal P}$ such that $P \Rightarrow P^{\downarrow\uparrow}\not\in T'$.
    Since $P \in {\cal P}$, Definition \ref{def:pseudint} yields that
    $||Q \Rightarrow Q^{\downarrow\uparrow}||_P = 1$ for every $Q \in {\cal P}$ with 
    $Q \ne P$, whence  $P \in \mathrm{Mod}(T')$.
    Since    $||P \Rightarrow P^{\downarrow\uparrow}||_P =
    S(P^{\downarrow\uparrow},P) \ne 1$ and since $T$ is cimplete in $\tu{X,Y,I}$, we obtain
  \begin{eqnarray*} 
    && \kern-1.2cm
      ||P \Rightarrow P^{\downarrow\uparrow}||_{\langle X,Y,I\rangle} =
         ||P \Rightarrow P^{\downarrow\uparrow}||_T  = 1
     \ne
       ||P \Rightarrow P^{\downarrow\uparrow}||_P \geq\\
     &\geq&\textstyle
       \bigwedge_{M\in\mathrm{Mod}(T')} ||P \Rightarrow P^{\downarrow\uparrow}||_M =
      ||P \Rightarrow P^{\downarrow\uparrow}||_{T'},
  \end{eqnarray*}
   establishing that $T'$ is not complete in $\tu{X,Y,I}$ and hence $T$ is non-redundant.
\end{proof}

As the following example illustrates, there may exist multiple systems of pseudo-intents for a given
$\tu{X,Y,I}$ which, moreover, vary in size.

\begin{example}\label{ex:Pmultipleexist}
Let $\mathbf{L}$ with $L = \{0,0.5,1\}$ be a  G\"odel chain
with ${}^{\ast}$ being the identity on $L$. Consider $\langle X,Y,I\rangle$, where $X=\{x\}$, $Y=\{y,z\}$,
and $I(x,y) = I(x,z) = 0$. The following systems of $\mathbf{L}$-sets of
attributes are the systems of pseudo-intents of $\langle X,Y,I\rangle$: \\[4pt]
\centerline{%
  \bgroup%
  \renewcommand{\arraystretch}{1.2}%
  \begin{tabular}{l@{~~}l}
    ${\cal P}_1 = \{\{z\},\{{}^{0.5\!\!}/y,{}^{0.5\!\!}/z\},\{y\}\}$,&
    ${\cal P}_3 = \{\{y\},\{{}^{0.5\!\!}/z\}\}$, \\
    ${\cal P}_2 = \{\{z\},\{{}^{0.5\!\!}/y\}\}$, &
    ${\cal P}_4 = \{\{{}^{0.5\!\!}/y\},\{{}^{0.5\!\!}/z\}\}$.
  \end{tabular}
  \egroup} 
\end{example}

In general, there exist finite data tables $\tu{X,Y,I}$ for which
there does not exist any system of pseudo-intents not even if ${}^\ast$ is
the globalization. This is illustrated by the following example.

\begin{example}\label{ex:Pnotexist}
Let $\mathbf{L}$ be any complete residuated lattice with $L=[0,1]$, 
let ${}^\ast$ be the globalization, and let $X = \{x\}$, $Y = \{y\}$,
and $I(x,y) = 0$. It is easily seen that
$\mathrm{Int}(X^{\ast},Y,I) = \{\{\},\{y\}\}$.
Assume that there exists a system $\mathcal{P}$ of pseudo-intents of $\tu{X,Y,I}$.
Due to Theorem~\ref{thm:nrbase},
$T = \{P \Rightarrow P^{\downarrow\uparrow} |\, P \in {\cal P}\}$ is a 
base and, therefore, Theorem~\ref{thm:basechar} implies that 
$\mathrm{Mod}(T) =\mathrm{Int}(X^{\ast},Y,I)= \{\{\},\{y\}\}$. 
Thus, for each $a \in (0,1)$ there must exist
$\{{}^{c\!\!}/y\} \in {\cal P}$  such that 
$||\{{}^{c\!\!}/y\} \Rightarrow
\{{}^{c\!\!}/y\}^{\downarrow\uparrow}||_{\{{}^{a\!\!}/y\}} \ne 1$,
i.e. 
\begin{equation}\label{eqn:ca}
(c\to a)^\ast =S(\{{}^{c\!\!}/y\}, \{{}^{a\!\!}/y\})^\ast\not\leq 
 S(\{{}^{c\!\!}/y\}^{\downarrow\uparrow}, \{{}^{a\!\!}/y\}).
\end{equation}
Since ${}^*$ is the globalization, \eqref{eqn:ca} gives 
$(c\to a)^\ast  = 1$, meaning that $c \leq a$
and thus $c \in [0,a]$.
Since $\{{}^{c\!\!}/y\}$ is a pseudointent, $\{{}^{c\!\!}/y\}\not\in \mathrm{Int}(X^{\ast},Y,I)= \{\{\},\{y\}\}$,
whence $c\not=0$ and thus $c\in(0,a]$.
Now, take any $b \in (0,c)$. Repeating the above idea yields a $d \in (0,b]$
such that $\{{}^{d\!\!}/y\} \in {\cal P}$ and
$||\{{}^{d\!\!}/y\} \Rightarrow
\{{}^{d\!\!}/y\}^{\downarrow\uparrow}||_{\{{}^{b\!\!}/y\}} \ne 1$.
Hence, the system of pseudo-intents ${\cal P}$ contains $\{{}^{c\!\!}/y\}$
and $\{{}^{d\!\!}/y\}$ with $0 < d < c < 1$, i.e.
$\{{}^{d\!\!}/y\} \subset \{{}^{c\!\!}/y\}$. However,
$\{{}^{d\!\!}/y\}^{\downarrow\uparrow} = \{y\} \not\subseteq \{{}^{c\!\!}/y\}$
which contradicts the assumption that ${\cal P}$ is a system of pseudo-intents.
\end{example}

In the remainder of this section we characterize the systems of pseudo-intents of $\tu{X,Y,I}$
as certain maximal independent sets in graphs associated to $\tu{X,Y,I}$.
For $\tu{X,Y,I}$, put
\begin{equation}
  V = \{P \in \mathbf{L}^Y \,|\, P \ne P^{\downarrow\uparrow}\}.
  \label{def:V}
\end{equation}
Clearly, if $V$ is empty, then ${\cal P} = \emptyset$ is the only system of pseudo-intents
of $\langle X,Y,I\rangle$. In this trivial case 
there is no non-trivial implication valid in $\tu{X,Y,I}$.
For non-empty $V$consider the binary relation $E$ on $V$ defined
by
\begin{equation}
  E = \{\langle P,Q\rangle \!\in V \,|\, P \ne Q \mbox{ and }
  ||Q \Rightarrow Q^{\downarrow\uparrow}||_P \ne 1\}
  \label{def:E}
\end{equation}
and the graph $\mathbf{G} = \langle V,E \cup E^{-1}\rangle$. 
The following lemma shows a first link between systems of pseudointents
and the graph $\mathbf{G}$.

\begin{lemma}\label{lem:mis}
  If $\emptyset \ne {\cal P}$ is a system of pseudo-intents then 
  ${\cal P}$ is a maximal independent set in $\mathbf{G}$.
\end{lemma}
\begin{proof}
  Clearly, $\mathcal{P}\subseteq V$, since 
  $P\not=P^{\downarrow\uparrow}$ for every member of a system of pseudo-intents.
  ${\cal P}$ is independent becuse otherwise there exist $P,Q \in {\cal P}$ with  $\langle P,Q\rangle \in E$,
  i.e. $ ||Q \Rightarrow Q^{\downarrow\uparrow}||_P \ne 1$, a contradiction to the definition of a system
  of pseudointents.
  Maximality: If $\mathcal{P}\cup\{P\}$ is independent for some $P\in V-\mathcal{P}$, then
  for each $Q\in \mathcal{P}$ we have $\tu{P,Q}\not\in E$, i.e. 
  $||Q \Rightarrow Q^{\downarrow\uparrow}||_P = 1$. Definition~\ref{def:pseudint} then implies
  $P\in\mathcal{P}$, a contradiction.
\end{proof}

However, as  Example \ref{exm:simple} shows, there may exist maximal independent
sets in $\mathbf{G}$ that are not systems of pseudo-intents.
For this reason, define for any $Q \in V$ and ${\cal P} \subseteq V$ the following
subsets of $V$:
\begin{align*}
  \mathrm{Pred}(Q) &= \{P \in V \,|\, \langle P,Q\rangle \in E\}, \\
  \mathrm{Pred}({\cal P}) &=
  \textstyle\bigcup_{Q \in {\cal P}}\mathrm{Pred}(Q).
\end{align*}
The following characterization of systems of pseudo-intents in terms of $\mathrm{Pred}({\cal P})$
may then be obtained.

\begin{lemma}\label{thm:pi_iff_pred}
  Let ${\cal P} \subseteq V$.
  ${\cal P}$ is a system of pseudo-intents if{}f\/
  $\mathcal{P}=V-\mathrm{Pred}(\mathcal{P})$.
\end{lemma}
\begin{proof}
Since every element $P$ of any system of pseudo-intents satisfies $P\in V$,
Definition~\ref{def:pseudint} implies that being a system of pseudo-intents is equivalent to the following condition:
\[
     \text{for every $P\in V$: $P\in\mathcal{P}$ if{}f for each $Q\in \mathcal{P}-\{P\}$ we have $\tu{P,Q}\not\in E$}.
\]
Since $\tu{P,P}\not\in E$, the last condition is equivalent to
\[
     \text{for every $P\in V$: $P\in \mathcal{P}$ if{}f $P\not\in \mathrm{Pred}(\mathcal{P})$},
\]
which is clearly equivalent to $\mathcal{P}=V-\mathrm{Pred}(\mathcal{P})$.
\end{proof}

Lemma~\ref{lem:mis} and Lemma~\ref{thm:pi_iff_pred} finally yield
the following characterizaiton of systems of pseudo-intents:

\begin{theorem}\label{col:comput}
  ${\cal P} \ne \emptyset$ is a system of pseudo-intents if{}f\/
  ${\cal P}$ is a maximal independent set in $\mathbf{G}$ such that
  $\mathcal{P}=V-\mathrm{Pred}(\mathcal{P})$.
\end{theorem}

\begin{figure}[t]
  \centering
  \def\rot#1{\rotatebox{90}{\small #1}}%
  \def\O{}%
  \def\X{\small $\boldsymbol{\times}$}%
  \def\SE#1{\ensuremath{\bigl\{#1\bigr\}}}
  \renewcommand{\arraystretch}{1.2}%
  \setlength{\tabcolsep}{2pt}%
  \begin{minipage}{.4\linewidth}
    \begin{tabular}{|r|*{6}c|}
      \hline
      &
      \rot{$\SE{}$} &
      \rot{$\SE{{}^{0.5\!}/z}$} &
      \rot{$\SE{z}$} &
      \rot{$\SE{{}^{0.5\!}/y,{}^{0.5\!}/z}$\,} &
      \rot{$\SE{{}^{0.5\!}/y,z}$} &
      \rot{$\SE{y}$} \\
      \hline
      \rule{0pt}{9pt}\small$\SE{}$ & \O & \X & \O & \X & \O & \O \\
      \small$\SE{{}^{0.5\!}/z}$ & \X & \O & \X & \X & \X & \O \\
      \small$\SE{z}$ & \X & \X & \O & \X & \X & \O \\
      \small$\SE{{}^{0.5\!}/y,{}^{0.5\!}/z}$ & \O & \X & \O & \O & \O & \O \\
      \small$\SE{{}^{0.5\!}/y,z}$ & \O & \X & \X & \X & \O & \O \\
      \rule[-5pt]{0pt}{9pt}\small$\SE{y}$ & \O & \O & \O & \O & \O & \O \\
      \hline
    \end{tabular}
  \end{minipage}
  \quad
  \begin{minipage}{.4\linewidth}
    \tikzset{graph vertices/.style={nodes={font=\small, inner sep=2pt}, row sep=4em, column sep=0em}}%
    \tikzset{graph edges/.style={thick}}%
    \begin{tikzpicture}
      \matrix [graph vertices] {
        \node (d) {$\SE{{}^{0.5\!}/y,{}^{0.5\!}/z}$}; \\
        & \node (a) {$\SE{}$};
        & \node [xshift=3em] (c) {$\SE{z}$};
        & \node [xshift=3em] (e) {$\SE{{}^{0.5\!}/y,z}$}; \\
        \node (b) {$\SE{{}^{0.5\!}/z}$};
        &&& \node (f) {$\SE{y}$}; \\
      };
      \path [graph edges]
      (a) edge (b)
      (a) edge (c)
      (a) edge (d)
      (b) edge (c)
      (b) edge (d)
      (b) edge (e)
      (c) edge (d)
      (c) edge (e)
      (d) edge (e);
    \end{tikzpicture}
  \end{minipage}
  \caption{Relation $E$ given by \eqref{def:E} and the induced graph from Example~\ref{exm:simple}}
  \label{Fig:Simple}
\end{figure}
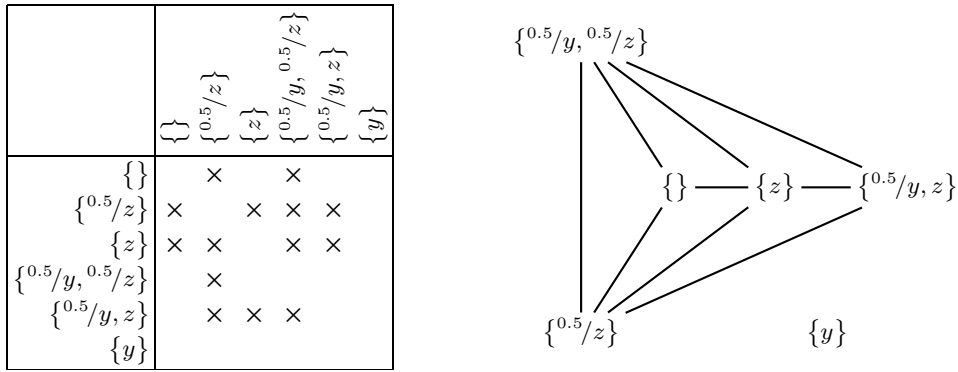

Using Theorem~\ref{col:comput}, one may compute systems of  pseudo-intents
by computing maximal independent sets in $\mathbf{G}$ and checking the additional 
condition $\mathcal{P}=V-\mathrm{Pred}(\mathcal{P})$. Note that this property may be checked
when generating the independent sets. The following example illustrates the procedure.

\begin{example}\label{exm:simple}
  Let $\mathbf{L}$ be a three-element \L ukasiewicz chain with
  $L = \{0,0.5,1\}$, and ${}^{\ast}$ being the identity on $L$. Consider the data table
  $\langle X,Y,I\rangle$ where $X = \{x\}$,
  $Y = \{y,z\}$, $I(x,y) = 0.5$, and $I(x,z) = 0$. The set $V$
  defined by (\ref{def:V}) is the following:
  \begin{displaymath}
    V = \{
    \{\},
    \{{}^{0.5\!}/z\},
    \{z\},
    \{{}^{0.5\!}/y,{}^{0.5\!}/z\},
    \{{}^{0.5\!}/y,z\},
    \{y\}\}.
  \end{displaymath}
  The corresponding binary relation $E$ defined by (\ref{def:E})
  is depicted in Fig.\,\ref{Fig:Simple} (left);
  graph $\mathbf{G} = \langle V,E \cup E^{-1}\rangle$
  is depicted in Fig.\,\ref{Fig:Simple} (right). $\mathbf{G}$ contains
  four maximal independent sets:
  \begin{align*}
    {\cal P}_1 &= \{\{\},\{{}^{0.5\!}/y,z\},\{y\}\}, &
    {\cal P}_3 &= \{\{z\},\{y\}\}, \\
    {\cal P}_2 &= \{\{{}^{0.5\!}/z\},\{y\}\}, &
    {\cal P}_4 &= \{\{{}^{0.5\!}/y,{}^{0.5\!}/z\},\{y\}\}.
  \end{align*}
  Observe that ${\cal P}_1$ and ${\cal P}_3$ do not satisfy
  ${\cal P}_i = V - \mathrm{Pred}({\cal P}_i)$ ($i \in \{1,3\}$) because
  $\{{}^{0.5\!}/y,{}^{0.5\!}/z\} \not\in \mathrm{Pred}({\cal P}_i)$,
   $i=1,3$,  and $\{\} \not\in \mathrm{Pred}({\cal P}_3)$.
  Hence, due to Theorem~\ref{thm:pi_iff_pred}, 
  ${\cal P}_1$ and ${\cal P}_3$ are not systems of pseudo-intents
  On the other hand, ${\cal P}_i =V -  \mathrm{Pred}({\cal P}_i)$ for
  $i=2,4$,, i.e. ${\cal P}_2$ and ${\cal P}_4$ are systems of
  pseudo-intents of $\tu{X,Y,I}$. The corresponding non-redundant bases
  $T_2$ and $T_4$ of $\tu{X,Y,I}$ given by Theorem \ref{thm:nrbase}
   are the following:
  \begin{align*}
    T_2 &= \{\{{}^{0.5\!}/z\} \!\Rightarrow\! \{y,{}^{0.5\!}/z\},
    \{y\} \!\Rightarrow\! \{y,{}^{0.5\!}/z\}\}, \\
    T_4 &= \{\{{}^{0.5\!}/y,{}^{0.5\!}/z\} \!\Rightarrow\! \{y,{}^{0.5\!}/z\},
    \{y\} \!\Rightarrow\! \{y,{}^{0.5\!}/z\}\}.
  \end{align*}
  Further algorithmic aspects of this procedure are discussed in
  Section~\ref{sec:a}.
\end{example}

\subsection{Pseudo-intents and bases corresponding to globalization}

It has been pointed out in Remark \ref{rem:classicP} that if the hedge ${}^\ast$  involved
in the definition of the validity of attribute implications is the globalization and if 
$L$ and $Y$ are finite, there exists a unique system of pseudointents for a given $\tu{X,Y,I}$.
In this section, we show that in this case, the corresponding bases are minimal
in terms of the number of implications contained in the base.
For the subsequent proofs, we need the
following technical observation which applies to general systems of
pseudointents using any hedge.

\begin{lemma}\label{Le:P_cap_Q_Int}
  Let $\mathcal{P}$ be a system of pseudointents of $\tu{X,Y,I}$ and let
  $P,Q \in {\cal P} \cup \mathrm{Int}(X^{\ast},Y,I)$ satisfy
  \begin{eqnarray}
    S(P,Q)^{\ast} &\leq&
    S(P^{\downarrow\uparrow},P \cap Q), \label{Eq:Min1} \\
    S(Q,P)^{\ast} &\leq&
    S(Q^{\downarrow\uparrow},P \cap Q). \label{Eq:Min2}
  \end{eqnarray}
  Then $P \cap Q \in \mathrm{Int}(X^{\ast},Y,I)$.
\end{lemma}
\begin{proof}
  Put $T' = T - \{P \Rightarrow P^{\downarrow\uparrow},
  Q \Rightarrow Q^{\downarrow\uparrow}\}$,
  where $T$ is a set of fuzzy attribute implications
  defined by (\ref{Def:T}).
  Definition~\ref{def:pseudint} and 
  the fact that $||C\Rightarrow C^{\downarrow\uparrow}||_D=1$
  for every $C\in\mathbf{L}^Y$ and $D\in \mathrm{Int}(X^{\ast},Y,I)$
  imply $P,Q \in \mathrm{Mod}(T')$. Hence, for each
  $A \Rightarrow B \in T'$ we have $S(A,P)^{\ast} \leq S(B,P)$
  and $S(A,Q)^{\ast} \leq S(B,Q)$. Consequently,
  \begin{eqnarray*}
    S(A,P \cap Q)^{\ast} =
    (S(A,P) \wedge S(A,Q))^{\ast} \leq 
    S(A,P)^{\ast} \wedge S(A,Q)^{\ast} \leq \\
    \leq S(B,P) \wedge S(B,Q) = 
    S(B,P \cap Q),
  \end{eqnarray*}
  which yields that $P \cap Q$ is a model of $T'$.
  Due to Theorem \ref{thm:basechar},
  it is now sufficient to verify that $P \cap Q$ is a model of
  $\{P \Rightarrow P^{\downarrow\uparrow},
  Q \Rightarrow Q^{\downarrow\uparrow}\}$.
  By virtue of (\ref{Eq:Min1}) and (\ref{Eq:Min2}), we have \\[4pt]
  \centerline{$S(P,P \cap Q)^{\ast} = S(P,Q)^{\ast} \leq
    S(P^{\downarrow\uparrow},P \cap Q)$} \\[2pt]
  and \\[2pt]
  \centerline{$S(Q,P \cap Q)^{\ast} = S(Q,P)^{\ast} \leq
    S(Q^{\downarrow\uparrow},P \cap Q)$,} \\[4pt]
  i.e. $||P \Rightarrow P^{\downarrow\uparrow}||_{P \cap Q} = 1$ and
  $||Q \Rightarrow Q^{\downarrow\uparrow}||_{P \cap Q} = 1$, finishing the proof.
\end{proof}

\begin{remark}
If $P$ and $Q$ are intents or pseudo-intents satisfying $S(P,Q)^{\ast} = S(Q,P)^{\ast} = 0$
then (\ref{Eq:Min1}) and (\ref{Eq:Min2}) are met and due to Lemma~\ref{Le:P_cap_Q_Int},
$P \cap Q$ is an intent.
Hence, if  ${}^{\ast}$ is the globalization and
$P$ and $Q$ are intents or pseudo-intents with $P \not\subseteq Q$ and $Q \not\subseteq P$,
then $P \cap Q$ is an intent.
\end{remark}

\begin{theorem}\label{thm:PGlob}
  Let $\mathbf{L}$ be a finite residuated lattice with ${}^{\ast}$ being
  the globalization, let $Y$ be finite. Let $\mathcal{P}$ be the sytem of
  pseudo-intents
  of $\tu{X,Y,I}$ and $T$ be the corresponding base given by (\ref{Def:T}).
  Then for any base $T'$ of $\langle X,Y,I\rangle$ we have
  $|T| \leq |T'|$.
\end{theorem}
\begin{proof}
  We first show that for each $P \in {\cal P}$,
  $T'$ contains an implication $A \Rightarrow B$ such that
  $A \subseteq P$ and $A^{\downarrow\uparrow} = P^{\downarrow\uparrow}$.
  We then show that two distinct $P,Q \in {\cal P}$ cannot share the same
  implication satisfying this property which proves
  that $|T| = |\mathcal{P}| \leq |T'|$.

  Take any $P \in {\cal P}$. By definition, $P \ne P^{\downarrow\uparrow}$
  and thus $P\not\in \mathrm{Int}(X^{\ast},Y,I)$. Since $T'$ is a base,
  Theorem~\ref{thm:basechar} yields that $T'$ contains $A \Rightarrow B$
  such that $||A \Rightarrow B||_P \ne 1$.  Since ${}^{\ast}$ is the
  globalization, we get $A \subseteq P$ and $B \not\subseteq P$.
  As every implication in $T'$ is valid in $\tu{X,Y,I}$ to degree $1$,
  Theorem \ref{thm:alter} yields  $S(B,A^{\downarrow\uparrow}) = 1$,
  i.e. $B \subseteq A^{\downarrow\uparrow}$.
  Thus, from $B \subseteq A^{\downarrow\uparrow}$ and $B \not\subseteq P$
  it follows that $A^{\downarrow\uparrow} \not\subseteq P$. 
  Now, $A \subseteq P$ and $A^{\downarrow\uparrow} \not\subseteq P$ yield
  $A \subseteq A^{\downarrow\uparrow} \cap P \subset A^{\downarrow\uparrow}$.
  Since $A^{\downarrow\uparrow}$ is the least intent containing $A$, it follows
  that   $A^{\downarrow\uparrow} \cap P$ is not an intent.
  Next, we claim that $P \subseteq A^{\downarrow\uparrow}$. By contradiction,
  if $P \not\subseteq A^{\downarrow\uparrow}$ then since 
  $A^{\downarrow\uparrow} \not\subseteq P$, Lemma~\ref{Le:P_cap_Q_Int}
  would give
  $A^{\downarrow\uparrow} \cap P \in \mathrm{Int}(X^{\ast},Y,I)$,
  a contradiction to the above observation that
  $A^{\downarrow\uparrow} \cap P \not\in \mathrm{Int}(X^{\ast},Y,I)$.
  Now, $A \subseteq P$ yields
  $A^{\downarrow\uparrow} \subseteq P^{\downarrow\uparrow}$
  while $P \subseteq A^{\downarrow\uparrow}$ yields
  $P^{\downarrow\uparrow} \subseteq
  A^{\downarrow\uparrow\downarrow\uparrow} = A^{\downarrow\uparrow}$,
  showing  $A^{\downarrow\uparrow} = P^{\downarrow\uparrow}$.

  Now, consider $P,Q \in \mathcal{P}$ such that $P \ne Q$ and assume that
  $A \Rightarrow B \in T'$ satisfies $A \subseteq P$, $A \subseteq Q$,
  and $P^{\downarrow\uparrow} = A^{\downarrow\uparrow} =
  Q^{\downarrow\uparrow}$. If $P \subset Q$,
  then $P^{\downarrow\uparrow} \subseteq Q$ and thus
  $A^{\downarrow\uparrow} = P^{\downarrow\uparrow} \subseteq
  Q \subset Q^{\downarrow\uparrow} = A^{\downarrow\uparrow}$,
  a contradiction. Dually for $Q \subset P$. Thus, assume that
  $P \nsubseteq Q$ and $Q \nsubseteq P$. Using Lemma~\ref{Le:P_cap_Q_Int},
  we get $P \cap Q \in \mathrm{Int}(X^{\ast},Y,I)$ and using the assumption
  that $A \subseteq P$ and $A \subseteq Q$, it follows that
  $A \subseteq P \cap Q$, i.e.,
  $A^{\downarrow\uparrow} \subseteq (P \cap Q)^{\downarrow\uparrow} =
  P \cap Q \subset P^{\downarrow\uparrow}$, a contradiction.
\end{proof}

\section{Algorithms}\label{sec:a}
In this section, we present algorithms for computing bases. We start by
an algorithm which simplifies the graph-theoretic procedure based on
Theorem~\ref{col:comput} from Section~\ref{sec:nrbpi}. In case of
globalization, we can show that the maximal independent set which determines
the (uniquely given) system of pseudo-intents can be directly computed
without the need to go over all maximal independent sets of the graph.
A simplified algorithm which follows is based on the following observation.

\begin{theorem}\label{th:gbasis}
  Let $\mathbf{L}$ be a finite linearly ordered residuated lattice with
  ${}^*$ being the globalization and let $\sqsubset$ be a strict total
  order on $\mathbf{L}^{\!Y}$ which extends $\subset$. Furthermore,
  assume that $\mathcal{P}$,  $V$, and $E$ are given by
  \eqref{eqn:sp}, \eqref{def:V}, and \eqref{def:E}, respectively.
  Let for $P \in \mathcal{P}$ denote
  $\mathcal{Q} = \{Q \in \mathcal{P} \,|\, Q \sqsubset P\}$.
  Then $P$ is the least element of
  $(V - \mathcal{Q}) - \mathrm{Pred}(\mathcal{Q})$
  with respect to $\sqsubset$.
\end{theorem}
\begin{proof}
  First, we prove that $P \in (V - \mathcal{Q}) - \mathrm{Pred}(\mathcal{Q})$.
  Obviously, $P \in V - \mathcal{Q}$ and thus it suffices to check that
  $P \not\in \mathrm{Pred}(\mathcal{Q})$ which means showing
  $P \not\in \mathrm{Pred}(Q)$ for all $Q \in \mathcal{Q}$. Since ${}^*$
  is globalization, $P \not\in \mathrm{Pred}(Q)$ and $P \ne Q$ mean that
  $Q^{\downarrow\uparrow} \subseteq P$ whenever $Q \subset P$ which is
  indeed true because $P \in \mathcal{P}$, cf.~\eqref{eqn:spG}.
  Second, we prove that $P$ is the least element of  $(V - \mathcal{Q}) - \mathrm{Pred}(\mathcal{Q})$.
 Assume by contradiction that  $P' \sqsubset P$ for some $P' \in (V - \mathcal{Q}) - \mathrm{Pred}(\mathcal{Q})$.
   Since $P' \in V - \mathcal{Q}$,
  we get $P' \not\in \mathcal{Q}$. On the other hand,
  from $P' \not\in \mathrm{Pred}(\mathcal{Q})$ it follows
  that $Q^{\downarrow\uparrow} \subseteq P'$ whenever $Q \in \mathcal{P}$
  and $Q \subset P'$ which by~\eqref{eqn:spG} gives $P' \in \mathcal{P}$
  and thus $P' \sqsubset P$ gives $P' \in \mathcal{Q}$, a contradiction.
\end{proof}

\begin{algorithm}
  \KwData{$\langle X,Y,I\rangle$ (input data),
    $\mathcal{S}$ (list of $\mathbf{L}$-sets
    $\{P \,|\, P \ne P^{\downarrow\uparrow}\}$ sorted by $\sqsubset$)}
  \KwResult{$\mathcal{P}$ (subset of $\mathbf{L}^{\!Y}$)}
  $\mathcal{P} \setto \emptyset$\;
  \While{$\IsNotEmpty{\mathcal{S}}$}{%
    $\mathcal{P} \setto \mathcal{P} \cup \{\First{\mathcal{S}}\}$\;
    $\mathcal{R} \setto \NewList$\;
    $B \setto \First{\mathcal{S}}$\;
    $\mathcal{S} \setto \Rest{\mathcal{S}}$\;
    \While{$\IsNotEmpty{\mathcal{S}}$}{%
      \If{$B \subset \First{\mathcal{S}}$ and
        $B^{\downarrow\uparrow} \nsubseteq \First{\mathcal{S}}$}{%
        $\Put{\mathcal{R}}{\First{\mathcal{S}}}$\;
      }
      $\mathcal{S} \setto \Rest{\mathcal{S}}$\;
    }
    $\mathcal{S} \setto \mathcal{R}$\;
  }
  \Return{$\mathcal{P}$}\;
  \caption{Computing the systems of pseudo-intents (case of globalization)}
  \label{alg:gbasis}
\end{algorithm}

Directly from Theorem~\ref{th:gbasis}, we derive a procedure for computing
pseudo-intents which utilizes the observation that in order to compute
$P \in \mathcal{P}$, it suffices to find all pseudo-intents which are
strictly smaller than $P$ according to a strict total $\sqsubset$ order
extending $\subset$. The procedure is formalized as Algorithm~\ref{alg:gbasis}.
The algorithm involves the following operations with linked lists:
$\First{\mathcal{S}}$ (the first element of list $\mathcal{S}$),
$\Rest{\mathcal{S}}$ (the rest of the list $\mathcal{S}$ except for
the first element), $\Put{\mathcal{S}}{B}$ (destructive modification
of the list $\mathcal{S}$ by putting the element $B$ to its end),
$\IsNotEmpty{\mathcal{S}}$ (condition true if list $\mathcal{S}$ is not empty),
$\NewList$ (a constructor for a new empty list). The algorithm takes
$\langle X,Y,I\rangle$ as the input and a list $\mathcal{S}$ which consists
of all elements of $V$ given by~\eqref{def:V} which are put in the list
in the order according to $\sqsubset$.

\begin{theorem}
  If $\mathbf{L}$ is a finite linear residuated lattice and ${}^*$ is
  globalization, then Algorithm~\ref{alg:gbasis} is correct:
  For $\langle X,Y,I\rangle$
  and $\mathcal{S}$ which is a list of elements
  $P \in \mathbf{L}^{\!Y}$ satisfying $P \ne P^{\downarrow\uparrow}$
  which are sorted according to a total strict order $\sqsubset$ extending
  $\subset$, the algorithm terminates after finitely many steps and
  it returns $\mathcal{P}$ satisfying~\eqref{eqn:sp}.
\end{theorem}
\begin{proof}
  It is easily seen that the algorithm always terminates because
  we consecutively remove elements from the list $\mathcal{S}$ and
  it eventually becomes empty. By induction on the number of loops
  of the outer while-loop, it suffices to check that whenever the algorithm
  reaches the beginning of the loop body, $\mathcal{P}$ contains all
  the elements from~\eqref{eqn:spG} which are smaller than
  $\First{\mathcal{S}}$ according to $\sqsubset$ provided
  that $\mathcal{S}$ is nonempty, and $\mathcal{S}$ consists
  of all elements from $(V - \mathcal{P}) - \mathrm{Pred}(\mathcal{P})$
  and that equality $\mathcal{P} \cap \mathrm{Pred}(\mathcal{P}) = \emptyset$
  is satisfied. The base case is clear. In the induction step, if
  $\mathcal{P}$ and $\mathcal{S}$ have these properties, from
  Theorem~\ref{th:gbasis} it follows that $\mathrm{First}(\mathcal{S})$
  can be added to $\mathcal{P}$ and the inner while-loop of the
  algorithm computes new $\mathcal{S}$ which contains all elements
  of $(V - \mathcal{Q}) - \mathrm{Pred}(\mathcal{Q})$
  for $\mathcal{Q} = \mathcal{P} \cup \{\First{\mathcal{S}}\}$.
  Moreover, $\mathcal{Q} \cap \mathrm{Pred}(\mathcal{Q}) = \emptyset$
  because $P \in \mathcal{Q} \cap \mathrm{Pred}(\mathcal{Q})$
  would violate the fact that all elements from $\mathcal{Q}$ are
  a subset of the elements from~\eqref{eqn:spG}.
  So, for the updated $\mathcal{P}$ and $\mathcal{S}$,
  the condition holds. Therefore, at the end of the computation,
  $\mathcal{S}$ is empty, meaning that 
  $(V - \mathcal{P}) - \mathrm{Pred}(\mathcal{P}) = \emptyset$,
  i.e., $V - \mathcal{P} \subseteq \mathrm{Pred}(\mathcal{P})$.
  Since $\mathcal{P} \cap \mathrm{Pred}(\mathcal{P}) = \emptyset$,
  we get $V - \mathcal{P} = \mathrm{Pred}(\mathcal{P})$.
  Now, apply Theorem~\ref{thm:pi_iff_pred}.
\end{proof}

Algorithm~\ref{alg:gbasis} is limited only to globalization ${}^*$ and
does not produce systems of pseudo-intents for general hedges.
Although it is more efficient than the naive application of 
Theorem~\ref{thm:pi_iff_pred} which involves looking for all maximal
independent sets, it still uses a large search space which is in general
exponential in terms of the size of $Y$ and $\mathbf{L}$.

An alternative approach to computing minimal bases using globalization and
complete sets using general hedges utilizes the idea of computing
fixed points of particular closure operators associated
to $\langle X,Y,I\rangle$. In particular, for any set $T$ of graded attribute
implications and any $\mathbf{L}$-set $C \in \mathbf{L}^Y$, we consider
an non-decreasing sequence of $\mathbf{L}$-sets $C_1,C_2,\ldots$
such that $C_1 = C$ and
\begin{align}
  C_{i+1} &=
  C_i \cup \textstyle\bigcup\{B \,|\, A \Rightarrow B \in T
  \text{ and } A \subset C_i\},
  \label{eqn:Ci}
\end{align}
for any natural number $i$ and put
\begin{align}
  [C]_T &= \textstyle\bigcup_{n=1}^{\infty}C_n.
  \label{eqn:clos}
\end{align}
If $\mathbf{L}$ if finite and linearly ordered and $Y$ is finite,
we get by the Tarski fixed point theorem that $[{\cdots}]_T$ defined by
\eqref{eqn:clos} is a closure operator. In addition, since both $L$
and $Y$ are finite, $[C]_T = C_n$ for some natural $n$.
Furthermore, 
we obtain the following theorem:

\begin{theorem}\label{thm:clT}
  Let $\mathbf{L}$ be finite and linearly ordered, $Y$ be finite,
  $\mathcal{P}$ be a system satisfying~\eqref{eqn:spG}, and
  let $T$ be given by~\eqref{Def:T}.
  Then $T$ is complete in $\langle X,Y,I\rangle$ and 
  \[
       \mathrm{fix}([\cdots]_T) = \mathcal{P} \cup \mathrm{Int}(X^{\ast},Y,I),
  \]
  i.e.
   $C = [C]_T$ if{}f
  $C \in \mathcal{P} \cup \mathrm{Int}(X^{\ast},Y,I)$.
\end{theorem}
\begin{proof}
  The fact that $T$ is complete in $\langle X,Y,I\rangle$ can be
  shown analogously as in the case of Theorem~\ref{thm:nrbase}.
  We therefore omit the proof but notice here that our $\mathcal{P}$,
  uniquely given by~\eqref{eqn:spG} (even if we consider a general hedge),
  need not satisfy~\eqref{eqn:sp}. Now, we prove that the set of all fixed points
  of $[{\cdots}]_T$ coincides with 
  $\mathcal{P} \cup \mathrm{Int}(X^{\ast},Y,I)$.

  Let $P \in \mathcal{P}$ and take
  $Q \Rightarrow Q^{\downarrow\uparrow} \in T$
  such that $Q \subset P$. Directly from~\eqref{eqn:spG},
  we get $Q^{\downarrow\uparrow} \subseteq P$ and thus
  $[P]_T \subseteq P$, i.e., $P$ is a fixed point of $[{\cdots}]_T$.
  Take $B \in \mathrm{Int}(X^\ast,Y,I)$ and 
  $Q \Rightarrow Q^{\downarrow\uparrow} \in T$
  such that $Q \subset B$. By monotony of ${}^{\downarrow\uparrow}$,
  we get $Q^{\downarrow\uparrow} \subseteq B^{\downarrow\uparrow} = B$,
  i.e., $[B]_T \subseteq B$.
  Conversely, let $C = [C]_T$ such that $C \ne C^{\downarrow\uparrow}$.
  It suffices to check that $C \in \mathcal{P}$. Since $C$ is
  a fixed point of $[{\cdots}]_T$, we get $Q^{\downarrow\uparrow} \subseteq C$
  for any $Q \Rightarrow Q^{\downarrow\uparrow} \in T$ such that
  $Q \subset C$. Using~\eqref{eqn:spG} and~\eqref{Def:T}, the latter
  gives $C \in \mathcal{P}$.
\end{proof}


Theorem~\ref{thm:clT} can be used to compute both the sets of intents of $\tu{X,Y,I}$
and the set $\mathcal{P}$ given by~\eqref{eqn:spG} for which the set $T$ given by 
\eqref{Def:T} is complete in $T$.
In case of the globalization, $T$ is a minimal base due to Theorem~\ref{thm:PGlob}. 
A procedure based on this
observation is presented in Algorithm~\ref{alg:general}. In order to
simplify notation, attribute sets used in the algorithm are subsets
of integers. In the algorithm, we use the following notation:
for $a \in L$ such that $a < 1$ we denote by $a^+$
the least element of $(a,1]$.
Such $a^+$ always exists since we assume that $\mathbf{L}$
is a finite and linearly ordered.

\begin{algorithm}
  \KwData{$\langle X,Y,I\rangle$ where $Y=\{1,\ldots,n\}$ (input data)}
  \KwResult{$\mathcal{I}$ and $\mathcal{P}$ (subsets of $\mathbf{L}^{\!Y}$)}
  $\mathcal{P} \setto \emptyset$\;
  $\mathcal{I} \setto \emptyset$\;
  \eIf{$\emptyset = \emptyset^{\downarrow\uparrow}$}{%
    $\mathcal{I} \setto \{\emptyset\}$\;
  }{%
    $\mathcal{P} \setto \{\emptyset\}$\;
  }
  $B \setto \emptyset$\;
  \While{$B \ne Y$}{%
    $T \setto \{B \Rightarrow B^{\downarrow\uparrow} \,|\,
    B \in \mathcal{P}\}$\;
    $C \setto B$\;
    \For{$y \setto 1$ \KwTo $n$}{%
      \If{$C(y) < 1$}{
        $B \setto \bigl[\bigl\{\deg{C(y)^+\!}{y},
          \deg{C(y+1)}{y+1},\ldots,\deg{C(n)}{n}\bigr\}\bigr]_T$\;
        \If{$B(z) = C(z)$ for all $z=1,\ldots,y-1$}{%
          break for loop\;
        }
      }
    }
    \eIf{$B = B^{\downarrow\uparrow}$}{%
      $\mathcal{I} \setto \mathcal{I} \cup \{B\}$\;
    }{%
      $\mathcal{P} \setto \mathcal{P} \cup \{B\}$\;
    }
  }
  \Return{$\mathcal{I}$, $\mathcal{P}$}\;
  \caption{Determining intents and complete sets}
  \label{alg:general}
\end{algorithm}

\begin{theorem}
  If $\mathbf{L}$ is a finite linear residuated lattice,
  then Algorithm~\ref{alg:general} is correct:
  For $\langle X,Y,I\rangle$, the algorithm terminates after finitely
  many steps and returns $\mathcal{I}$ and $\mathcal{P}$
  such that $\mathcal{I} = \mathrm{Int}(X^{\ast},Y,I)$ and
  $\mathcal{P}$ satisfies \eqref{eqn:spG}.
\end{theorem}
\begin{proof}
  Suppose that $Y = \{1,\ldots,n\}$. The algorithm maintains two sets:
  $\mathcal{I}$ contains only $\mathbf{L}$-sets $B$ such that
  $B = B^{\downarrow\uparrow}$ and $\mathcal{P}$ contains only
  $\mathbf{L}$-sets $B$ such that $B \ne B^{\downarrow\uparrow}$.
  The main loop of the algorithm goes through all fixed points
  of $[{\cdots}]_T$ in the lexicographic order $\sqsubset$
  where $B_1 \sqsubset B_2$ if{}f there is $y \in Y$ such that
  $B_1(y) < B_2(y)$ and $B_1(z) = B_2(z)$ for all $z > y$
  (recall that for simplicity we have denoted attributes as integers).
  Indeed, the for-loop finds a lexical successor of $B$ with respect
  to such $\sqsubset$ which is a fixed point of $[{\cdots}]_T$.
  Note that $T$ already contains all necessary implications to compute
  such fixed point because all elements of~\eqref{eqn:spG} which are
  strictly smaller than the current $B$ are already in $\mathcal{P}$.
  Hence, at the end of the computation, $\mathcal{P}$ and $\mathcal{I}$
  consists of all the fixed points of $[{\cdots}]_T$.
  The rest follows from Theorem~\ref{thm:clT}.
\end{proof}

\begin{remark}
  (1):
  If ${}^*$ is a general hedge then Algorithm~\ref{alg:general}
  produces $\mathcal{P}$ such that the corresponding theory $T$ given
  by~\eqref{Def:T} is complete but may be redundant. In order to get
  a non-redundant one, i.e. a base, we may consecutively remove from $T$ graded
  implications which follow from other graded implications from the
  theory, i.e., we may repeatedly apply Lemma~\ref{thm:nonr}\,(iii).
  Namely, $T$ is non-redundant if there is no $A \Rightarrow B$ such
  that $||A \Rightarrow B||_{T-\{A \Rightarrow B\}} = 1$.
  According to Theorem~\ref{thm:ent}, the equality can be checked
  by showing $B \subseteq C_{\mathrm{Mod}(T-\{A \Rightarrow B\})}(A)$,
  i.e., by showing whether $B$ is contained in the least fixed point
  of $C_{\mathrm{Mod}(T-\{A \Rightarrow B\})}$ containing $A$.

  (2):
  If $\mathbf{L}$ and $Y$ are finite, the fixed points of
  $C_{\mathrm{Mod}(T)}$ which play a role in the previous remark
  can be efficiently computed. Namely, for $M \in \mathbf{L}^{\!Y}$
  we may put $M_1 = M$ and
  \begin{align}
    M_{i+1} &=
    M \cup \textstyle\bigcup\{B \otimes S(A,M)^{\ast} \,|\,
    A \Rightarrow B \in T\}
    \label{eqn:Tst}
  \end{align}
  for any natural number $i$. It is east to see  
   that $C_{\mathrm{Mod}(T)}(M) = \bigcup_{n=1}^\infty M_n$.
   Note that since $L$ and $Y$ are finite,  $\bigcup_{n=1}^\infty M_n$
  is equal to $M_n$ for some $n$.

  (3):
  The complexity of computing bases derives from the fact that even for $L=\{0,1\}$, there may
  exist an exponential number of pseudo-intents in terms of $|X|$ and $|Y|$ (number of objects
   and attributes) \cite{KuOb:Sdcpdgbi}. Hence, since the size of a smallest base equals the number of pseudo-intents,
   a smallest base may have an exponential size in terms of $|X|$ and $|Y|$ in the worst case.
   On the other hand, the time delay \cite{JoYaPa:Gmis}, which is for the above reason an appropriate concept in
   our case,  of Algorithm~\ref{alg:general}  is  $\leq O(|L|)$-times 
  the time delay of the  basic algorithm for computing
  ordinary pseudo-intents \cite{Gan:Tbaca}, which is also described in \cite{GaWi:FCA}.
\end{remark}

\section{Reducing Graded Attribute Implications to Ordinary Ones via Thresholding}
\label{sec:rgai}

As mentioned above, ordinary attribute implications are a particular case of graded implications
in which $0$ and $1$ are the only degrees involved.
In this section we look at whether and to what extent the notions regarding graded attribute
implications and their bases may be reduced to those regarding ordinary implications. 
In particular, we show that  every data table $\tu{X,Y,I}$ with graded attributes may be transformed
via a natural thresholding to a table $\tu{X^\times,Y^\times,I^\times}$ with binary attributes in such a way that validity 
of graded implications in $\tu{X,Y,I}$ corresponds to validity of the ordinary implications in $\tu{X^\times,Y^\times,I^\times}$.
A natural question arises of whether bases of $\tu{X,Y,I}$ may be obtained
from the bases of $\tu{X^\times,Y^\times,I^\times}$, since the latter ones may be computed
be existing algorithms \cite{GaWi:FCA}.
As we show, 
the answer to this question is negative.
Namely, while complete sets of ordinary implications in $\tu{X^\times,Y^\times,I^\times}$ 
yield complete sets of graded implications in $\tu{X,Y,I}$ via the transformation, 
it may happen that non-redundant sets of ordinary implications transform to redundant sets
of graded implications.

The transformation via thresholding is based on the following  idea.
Given a graded attribute $y$, one may consider for every truth degree $b\in L$ the corresponding
bivalent attribute $\tu{y,b}$ as follows: $\tu{y,b}$ applies to the object $x$ if and only if
$y$ applies to $x$ at least to degree $b$. This idea is, in fact, a particular case of a more general
one which underlies the following definition.

Given a table $\langle X,Y,I\rangle$ with graded attributes, denote by
$\tu{X^\times,Y^\times,I^\times}$ the table with graded attributes defined by:
\begin{eqnarray*}
   && X^\times = X\times\ast(L), \text{ where } \ast(L)=\{a^\ast \mid a\in L\},\\
  && Y^\times = Y\times L,\\
  && \text{$\langle \langle x,a\rangle,\langle y,b\rangle\rangle\in I^\times$ iff $a\otimes b\leq I(x,y)$.}
%
\end{eqnarray*}

One may easily check using the properties of hedges that $\ast(L)=\{a\in L \mid a^\ast=a\}$,
i.e. $\ast(L)$ is the set of all fixpoints of $\ast$. 
If $\ast$ is globalization, $\ast(L)=\{0,1\}$, 
the new objects of the form $\tu{x,1}$ may be identified with
the original objects $x\in X$ while those of the form $\tu{x,0}$ may be dropped because
they are redundant (every new attribute $\tu{y,b}$ applies to them). In this case, 
the ordinary relation 
$I^\times$ coincides with the one which corresponds to the simple thresholding as described above
because then, $\tu{y,b}$ applies to $x$, i.e. to $\tu{x,1}$, if{}f 
$b = 1\otimes b\leq I(x,y)$, i.e. $y$ applies to $x$ at least to degree $b$.

To transform graded attribute implications to ordinary ones and vice versa,
we utilize the following mappings between $\mathbf{L}$-sets and ordinary sets.
For an $\mathbf{L}$-set $B\in\mathbf{L}^Y$ we define the ordinary subset
$\Cr{B}$ of $Y\times L$ by
\[
 \Cr{B}=\{\langle y,a\rangle\in Y\times L\,|\, a\leq B(y)\}.
\]
For a subset $D\subseteq Y\times L$ we define the $\mathbf{L}$-set
$\Fu{D}$ in $Y$ by
\[
  \Fu{D}(y)=\textstyle\bigvee\{a\,|\, \langle y,a\rangle\in D\}.
\]
With these correspondences, one may look at the relationship between
the validity of graded implications in $\tu{X,Y,I}$ on one hand and the validity of
ordinary implications in $\tu{X^\times,Y^\times,I^\times}$ on the other hand.
Namely, for a given graded implication $A\Rightarrow B$ over $Y$
(i.e. $A,B\in\mathbf{L}^Y$), one may consider the corresponding  ordinary implication 
$\Cr{A}\Rightarrow \Cr{B}$ over $Y\times L$ (i.e. $\Cr{A},\Cr{B}\subseteq Y\times L$),
and conversely, for an ordinary implication $C\Rightarrow D$ over $Y\times L$, one 
may consider the corresponding  graded implication 
$\Fu{C}\Rightarrow \Fu{D}$ over $Y$.
The relationship in question is described by the following theorem which says that
the transformations described above preserve validity of implications 
(for brevity, we write $||A\Rightarrow B||_{I}$ instead of $||A\Rightarrow B||_{\tu{X,Y,I}}$
and the same for $||A\Rightarrow B||_{I^\times}$).

\begin{theorem}
\label{thm:reduction}
For a data table $\langle X,Y,I\rangle$ with graded attributes, 
the corresponding $\tu{X^\times,Y^\times,I^\times}$,
and arbitrary $A\in \mathbf{L}^Y$, $B\in \mathbf{L}^Y$ and  $C,D\subseteq Y\times L$,
 we have
\begin{eqnarray}
  \label{eqn:Ttimes1}
  ||A\Rightarrow B||_{I}=1 
    \quad\text{if and only if}\quad ||\Cr{A}\Rightarrow \Cr{B}||_{I^\times}=1;\\
  \label{eqn:Ttimes2}
||C\Rightarrow D||_{I^\times}=1 \quad \text{if and only if}\quad ||\Fu{C}\Rightarrow \Fu{D}||_{I}=1.
\end{eqnarray}
\end{theorem}

Before we turn to the proof of Theorem \ref{thm:reduction}, we present some auxiliary
results.
Denote by ${}^\cru: {2}^{X\times\ast(L)}\rightarrow {2}^{Y\times L}$ and 
${}^\crd: {2}^{Y\times\tsy(L)}\rightarrow {2}^{X\times\ast(L)}$ the Galois connections
induced by $I^\times$ \cite{Ore:Gc}, i.e. 
\begin{eqnarray*}
     C^\cru &=&\{\tu{y,b}\in Y\times L \mid \text{for each }\tu{x,a}\in C:\, 
            \tu{\tu{x,a},\tu{y,b}}\in I^\times\}, \text{ and}\\
     D^\crd &=&\{\tu{x,a}\in X\times \ast(L) \mid \text{for each }\tu{y,b}\in D:\, 
            \tu{\tu{x,a},\tu{y,b}}\in I^\times\}, 
\end{eqnarray*}
for every $C\subseteq X\times\ast(L)$ and $D\subseteq Y\times L$.
Furthermore, let us extend $\Cr{\ }$ and $\Fu{\ }$ for any $A\in \mathbf{L}^X$
and $C\subseteq X\times\ast(L)$ by putting 
\(
 \Cr{A}=\{\langle x,a\rangle\in X\times \ast(L)\,|\, a\leq A(x)\}
\) and
\(
  \Fu{C}(y)=\bigvee\{a \,|\, \langle x,a\rangle\in C\}
\).
As $(\bigvee_{k} a_k^\ast)^\ast=\bigvee_{k} a_k^\ast$ (due to the isotony and idempotency
of $\ast$), $\ast(L)$ is closed under suprema and, hence,  $\Fu{C}(x)\in\ast(L)$ for every $x\in X$.
The following lemma describes the relationship between 
$\tu{{}^\upts,{}^\downts}$ and $\tu{{}^\cru,{}^\crd}$, and some further properties.

\begin{lemma}\label{thm:cgc}
For every $A\in\mathbf{L}^X$, $B\in\mathbf{L}^Y$, 
$C\subseteq X\times \ast(L)$, and $D\subseteq Y\times L$,
\begin{eqnarray}
    \label{eqn:upcru1}
    && A^\up = \Fu{\Cr{A^\ast}^\cru}, \quad  B^\down = \Fu{\Cr{B}^\crd}, \quad
     C^\cru = \Cr{\Fu{C}^\up}, \quad\text{and}\quad D^\crd = \Cr{\Fu{D}^\down};\\
    \label{eqn:upcru2}   
   &&  C^\cru = \Cr{\Fu{C}}^\cru \quad\text{and}\quad D^\crd = \Cr{\Fu{D}}^\crd;\\
    \label{eqn:upcru3}   
   &&  \Cr{A}^{\cru\crd} = \Cr{A^{\up\down}}, \quad \Cr{B}^{\crd\cru} = \Cr{B^{\down\up}},
   \quad \Fu{C}^{\up\down} = \Fu{C^{\cru\crd}}, \quad\text{and}\quad
   \Fu{D}^{\down\up} = \Fu{D^{\crd\cru}}.
\end{eqnarray}
%
%
%
\end{lemma}
\begin{proof}
  (\ref{eqn:upcru1})
  We have
  \begin{eqnarray*}
    &&\textstyle \Fu{\Cr{A^\ast}^\cru}(y) = \bigvee\{b \mid \tu{y,b}\in \Cr{A^\ast}^\cru \} =
     \bigvee\{b \mid \text{for each }\tu{x,a}\in \Cr{A^\ast}:\, \tu{\tu{x,a},\tu{y,b}}\in I^\times \}=\\
   &=&\textstyle
     \bigvee\{b \mid \text{for each }\tu{x,a}\in \Cr{A^\ast}:\,  a\otimes b\leq I(x,y) \}=\\
   &=&\textstyle
         \bigvee\{b \mid \text{for each } x\in X, a\in\ast(L):\, a\leq A^\ast(x) \text{ implies } 
            a\otimes b\leq I(x,y) \}=\\
   &=&\textstyle
         \bigvee\{b \mid \text{for each } x\in X:\,     A^\ast(x)\otimes b\leq I(x,y) \}=
   \textstyle
         \bigvee\{b \mid  b\leq \bigwedge_{x\in X}  (A^\ast(x)\to I(x,y)) \}=\\
   &=&\textstyle
         \bigwedge_{x\in X}  (A^\ast(x)\to I(x,y)) = A^\up(y),
  \end{eqnarray*}
  proving  $A^\up = \Fu{\Cr{A^\ast}^\cru}$.
  $\quad  B^\down = \Fu{\Cr{B}^\crd}$ is proven
  analogously.

 To verify $C^\cru = \Cr{\Fu{A}^\up}$, we reason as follows:
\begin{eqnarray*}
  && \tu{y,b}\in \Cr{\Fu{C}^\up} \text{ if{}f } 
   b\leq \Fu{C}^\up(y) =\textstyle \bigwedge_{x\in X} (\Fu{C}^\ast(x)\to I(x,y))\\
  &&\textstyle
   \text{ if{}f for each } x\in X:\,  b\leq \Fu{C}^\ast(x)\to I(x,y) = 
        (\bigvee_{\tu{x,a}\in C} a)^\ast \to I(x,y) =  (\bigvee_{\tu{x,a}\in C} a) \to I(x,y)\\
  &&\textstyle
   \text{ if{}f for each } x\in X:\, \bigvee_{\tu{x,a}\in C} (a\otimes b) = 
              (\bigvee_{\tu{x,a}\in C} a)\otimes b \leq I(x,y)\\
   &&\text{ if{}f for each } \tu{x,a}\in C:\, a\otimes b \leq I(x,y)\\
   &&\text{ if{}f for each } \tu{x,a}\in C:\, \tu{\tu{x,a},\tu{y,b}}\in I^\times
     \text{ if{}f } \tu{y,b}\in C^\cru.
\end{eqnarray*}  
  $D^\crd = \Cr{\Fu{D}^\cru}$ is proven analogously.

  (\ref{eqn:upcru2}):
  Due to (\ref{eqn:upcru1}) and since $\Fu{\Cr{A}}=A$ for every $A\in\mathbf{L}^X$,
  $C^\cru = \Cr{\Fu{C}^\up} = \Cr{\Fu{\Cr{\Fu{C}}}^\up}=
    \Cr{\Fu{C}}^\cru$.
   $D^\crd = \Cr{\Fu{D}}^\crd$ is proven analogously.

    (\ref{eqn:upcru3}):
  By virtue of (\ref{eqn:upcru1}) and since 
  $\Fu{\Cr{M}}=M$, we have
   $\Cr{A}^{\cru\crd} = \Cr{\Fu{\Cr{A}}^\up}^\crd =
  \Cr{\Fu{\Cr{\Fu{\Cr{A}}^\up}}^\down} = \Cr{A^{\up\down}}$.
  The second equality is proven analogously.

   Due to (\ref{eqn:upcru1}) and (\ref{eqn:upcru2}),  
   $\quad \Fu{C}^{\up\down} =  \Fu{\Cr{\Fu{\Cr{\Fu{C}}^\cru}}^\crd} 
    =  \Fu{C^{\cru\crd}}$. The last equality is proven dually.
\end{proof}

\begin{proof}[of Theorem \ref{thm:reduction}]
(\ref{eqn:Ttimes1}): 
Due to Theorem \ref{thm:alter}, $||A\Rightarrow B||_{I}=1$
is equivalent to $S(B,A^{\down\up})$, i.e. to $B\subseteq A^{\down\up}$,
and $||\Cr{A}\Rightarrow \Cr{B}||_{I^\times}=1$
is equivalent to $\Cr{B}\subseteq \Cr{A}^{\crd\cru}$.
Now, due to (\ref{eqn:upcru3}), 
$\Cr{A}^{\crd\cru}=\Cr{A^{\down\up}}$.
Hence, we need to check that $B\subseteq A^{\down\up}$
if and only if $\Cr{B}\subseteq \Cr{A^{\down\up}}$ which is clearly the case
since for every $M,N\in \mathbf{L}^Y$, $M\subseteq N$ is equivalent to
$\Cr{M}\subseteq \Cr{N}$.

(\ref{eqn:Ttimes2}):
  We prove the claim by establishing that
  (a) $D\subseteq C^{\crd\cru}$ is equivalent to (b) $\Fu{D}\subseteq\Fu{C}^{\down\up}$.
  Namely, on account of Theorem \ref{thm:alter}, 
  (a) is equivalent to $||C\Rightarrow D||_{I^\times}=1$ and (b), i.e. 
   $S(\Fu{D},\Fu{C}^{\down\up})=1$, is equivalent to
   $||\Fu{C}\Rightarrow \Fu{D}||_{I}=1$.
   Since (a) clearly implies $\Fu{D}\subseteq\Fu{C^{\crd\cru}}$ and since
     $\Fu{C^{\crd\cru}}=\Fu{C}^{\down\up}$ on account of (\ref{eqn:upcru3}), 
  we see that (a) implies (b).
  Assume (b). Then clearly,  $\Cr{\Fu{D}}\subseteq \Cr{\Fu{C}^{\down\up}}$.
  As $\Cr{\Fu{C}^{\down\up}}=\Cr{\Fu{C}}^{\crd\cru}=C^{\crd\cru}$ on account of 
  (\ref{eqn:upcru3}) and (\ref{eqn:upcru2}), we get
  $\Cr{\Fu{D}}\subseteq C^{\crd\cru}$. As $D\subseteq \Cr{\Fu{D}}$, we have
  $D\subseteq C^{\crd\cru}$, proving that (b) implies (a).
   \end{proof}

\begin{remark}
  In addition to (\ref{eqn:Ttimes2}) of Theorem \ref{thm:reduction}, we also have
 \[
  ||C\Rightarrow D||_{I^\times}=1 \quad\text{if and only if}\quad
 ||\Cr{\Fu{C}}\Rightarrow \Cr{\Fu{D}}||_{I^\times}=1.
\]
Namely, the two conditions involved are equivalent to 
 (a) $D\subseteq C^{\crd\cru}$ and (b) $\Cr{\Fu{D}}\subseteq \Cr{\Fu{C}}^{\crd\cru}=
  C^{\crd\cru}$, respectively,
on account of Theorem \ref{thm:alter} and (\ref{eqn:upcru2}).
 Since $D\subseteq \Cr{\Fu{D}}$, (b) and (\ref{eqn:upcru2}) clearly imply (a).
 On the other hand, 
  (a) implies $\Cr{\Fu{D}}\subseteq \Cr{\Fu{C^{\crd\cru}}}$.
  Since $C^{\crd\cru}=\Cr{\Fu{C^\crd}^\up}$, 
 we have $\Cr{\Fu{C^{\crd\cru}}}=
   \Cr{\Fu{ \Cr{\Fu{C^\crd}^\up}  }} = \Cr{\Fu{C^\crd}^\up}  =C^{\crd\cru}$.
  As a result,  (a) implies (b).
 \end{remark}

In view of the above results,  a natural question is whether one can obtain
complete sets and bases of a given table $\langle X,Y,I\rangle$ with graded attributes
from complete sets and bases of the corresponding 
$\langle X^\times,Y^\times,I^\times\rangle$.
This question is the subject of the next theorem and the following remark.

\begin{theorem}
\label{thm:completer}
If $T^\times$ is complete in
$\langle X^\times,Y^\times,I^\times\rangle$ then
\begin{equation}
  \label{eqn:TTtimes}
    \Fu{T^\times} =\{\Fu{C}\Rightarrow\Fu{D}\,|\, C\Rightarrow D\in T^\times\}
\end{equation}
is complete in $\langle X,Y,I\rangle$.
\end{theorem}
\begin{proof}
Let $T^\times$ be complete in
$\langle X\times\ast(L),Y\times L,I^\times\rangle$
Due to Theorem \ref{thm:basechar}, it is sufficient to show that 
$\mathrm{Mod}(\Fu{T^\times})=\mathrm{Int}(X^\ast,Y,I)$.
We prove this fact by showing that the following claims are equivalent
for any $M\in\mathbf{L}^Y$:
\bgroup%
\addtolength{\leftmargini}{.8em}%
\begin{itemize}
  \item[(a)] $M\in \mathrm{Mod}(\Fu{T^\times})$,
  \item[(b)] for each $a\in L$: $\Cr{a^\ast\to M}\in \mathrm{Mod}(T^\times)$,
  \item[(c)] $M\in\mathrm{Int}(X^\ast,Y,I)$.
\end{itemize}
\egroup%

``(a) $\Leftrightarrow$ (b)'': 
Clearly, it suffices to show that for
every $C\Rightarrow D\in T^\times$, 
$M$ is a model of $\Fu{C}\Rightarrow\Fu{D}$ if{}f for each $a\in L$,
$\Cr{a^\ast\to M}$ is a model of $C\Rightarrow D$, i.e. that
$S(\Fu{C},M)^\ast\leq S(\Fu{D},M)$ if{}f
for each $a\in L$, $C\subseteq\Cr{a^\ast\to M}$ implies $D\subseteq\Cr{a^\ast\to M}$.

Observe first that $C\subseteq\Cr{a^\ast\to M}$ is equivalent to 
$a^\ast\leq S(\Fu{C},M)$: 
Namely, $C\subseteq\Cr{a^\ast\to M}$ means that
for every $y\in Y$, if $\tu{y,b}\in C$ then $\tu{y,b}\in\Cr{a^\ast\to M}$,
i.e. $b\leq (a^\ast\to M)(y)=a^\ast\to M(y)$.
Therefore, $C\subseteq\Cr{a^\ast\to M}$ means that 
for every $y\in Y$,
$\bigvee_{\tu{y,b}\in C} b \leq a^\ast\to M(y)$ which holds if{}f
for every $y\in Y$, 
$$\textstyle a^\ast\leq (\bigvee_{\tu{y,b}\in C} b) \to M(y)=\Fu{C}(y)\to M(y)$$
which is equivalent to 
$a^\ast\leq\bigwedge_{y\in Y} (\Fu{C}(y)\to M(y))=S(\Fu{C},M)$.

Since the same holds for $D$, to prove that (a) is equivalent to (b), it is sufficient
to check that 
$S(\Fu{C},M)^\ast\leq S(\Fu{D},M)$ if{}f
for every $a\in L$, $a^\ast\leq S(\Fu{C},M)$ implies $a^\ast\leq S(\Fu{D},M)$,
which is easy to see. 
Indeed, if $S(\Fu{C},M)^\ast\leq S(\Fu{D},M)$ and 
$a^\ast\leq S(\Fu{C},M)$, then 
\[
  a^\ast=a^{\ast\ast}\leq S(\Fu{C},M)^\ast\leq S(\Fu{D},M).
\]
Conversely, putting $a=S(\Fu{C},M)$ the assumption, i.e.
$a^\ast\leq S(\Fu{C},M)$ implies $a^\ast\leq S(\Fu{D},M)$, readily yields
$S(\Fu{C},M)^\ast\leq S(\Fu{D},M)$.

``(b) $\Leftrightarrow$ (c)'':
Theorem \ref{thm:basechar} (actually, its instance for $L=\{0,1\}$) implies that
$\Cr{a^\ast\to M}\in\mathrm{Mod}(T^\times)$ if{}f
$\Cr{a^\ast\to M}\in\mathrm{Int}(\langle X\times\ast(L),Y\times L,I^\times\rangle)$.

Next, observe that $\Cr{N}\in \mathrm{Int}(\langle X\times\ast(L),Y\times L,I^\times\rangle)$
is equivalent to $N\in \mathrm{Int}(X^\ast,Y,I)$.  
(\ref{eqn:Mint}) implies that to check this, it suffices to check that 
$\Cr{N}=\Cr{N}^{\crd\cru}$ is equivalent to $N=N^{\down\up}$.
Using Lemma \ref{thm:cgc}, we reason as follows.
If $\Cr{N}=\Cr{N}^{\crd\cru}$, then 
\begin{eqnarray*}
   && N^{\down\up} = \Fu{ \Cr{\Fu{\Cr{N}^\crd}}^\cru } =
   \Fu{\Cr{N}^{\crd\cru}} =\Fu{\Cr{N}} = N.
\end{eqnarray*}
Here, we used that $\Fu{\Cr{P}} = P$ for any $P\in\mathbf{L}^Y$ (obvious) and
$\Cr{\Fu{\Cr{N}^\crd}}=\Cr{N}^\crd$ which holds because
$\Cr{\Fu{\Cr{N}^\crd}}= \Cr{\Fu{  \Cr{\Fu{\Cr{N}}^\down} }}=
 \Cr{ \Fu{\Cr{N}}^\down  }=\Cr{N}$.
Conversely, if $N=N^{\down\up}$ then
\begin{eqnarray*}
   && \Cr{N}^{\crd\cru} = \Cr{ \Fu{\Cr{\Fu{\Cr{N}}^\down}}^\up } =
   \Cr{N^{\down\up}} =\Cr{N}.
\end{eqnarray*}

Applying this observation to $N=a^\ast\to M$, we see that (b) is equivalent to
that fact that for each $a\in L$: $a^\ast\to M\in\mathrm{Int}(X^\ast,Y,I)$.
The proof is complete by observing that the fact that
for each $a\in L$ we have
$a^\ast\to M\in\mathrm{Int}(X^\ast,Y,I)$ is equivalent to 
$M\in\mathrm{Int}(X^\ast,Y,I)$. Indeed, since for $a=1$ we have  $a\to M=M$,
it is sufficient to observe that if $M\in\mathrm{Int}(X^\ast,Y,I)$ then 
$a^\ast\to M\in\mathrm{Int}(X^\ast,Y,I)$ for each $a\in L$. 
This follows from the fact that $\mathrm{Int}(X^\ast,Y,I)$ is an $\mathbf{L}^\ast$-closure
system \cite{BeFuVy:Fcots} (cf. Section \ref{sec:cp}).
\end{proof}

The set $\Fu{T^\times}$ obtained from a given $T^\times$ according to Theorem \ref{thm:completer}
need not be a base of $\tu{X,Y,I}$ even if $T^\times$ is a base of 
$\langle X^\times,Y^\times,I^\times\rangle$. 
Namely, as the next example  shows, $\Fu{T^\times}$ may be redundant.

\begin{example}
\label{ex:Ttimes}
Let $\mathbf{L}$ be the three-element \L ukasiewicz chain with $L = \{0,0.5,1\}$,
let $\ast$ be the globalization, 
and let $X = \{x\}$, $Y = \{y,z\}$, $I(x,y) = I(x,z) = 0$.
One may verify that
(abbreviating $\tu{y,a}$ by $y_a$)
\begin{align*}
  {\cal P}^\times = \{&
  \{y_0,y_{0.5},z_0\},
  \{y_0,y_1,z_0\},
  \{y_0,z_0,z_{0.5}\}, 
  \{y_0,z_0,z_1\},
  \{\}\}
\end{align*}
is the system of pseudo-intents of $\langle X^\times,Y^\times,I^\times\rangle$.
Therefore, 
\begin{align*}
  T^\times = \{&\impl{\{y_0,y_{0.5},z_0\}}{\{y_0,y_{0.5},y_1,z_0,z_{0.5},z_1\}}, \\
  &\impl{\{y_0,y_1,z_0\}}{\{y_0,y_{0.5},y_1,z_0,z_{0.5},z_1\}}, \\
  &\impl{\{y_0,z_0,z_{0.5}\}}{\{y_0,y_{0.5},y_1,z_0,z_{0.5},z_1\}}, \\
  &\impl{\{y_0,z_0,z_1\}}{\{y_0,y_{0.5},y_1,z_0,z_{0.5},z_1\}}, \\
  &\impl{\{\}}{\{y_0,z_0\}}\}
\end{align*}
is a base of $\langle X^\times,Y^\times,I^\times\rangle$.
Clearly,
\begin{align*}
  \Fu{T^\times} = \{&\impl{\{\deg{0.5}{y}\}}{\{y,z\}},
  \impl{\{y\}}{\{y,z\}},
  \impl{\{\deg{0.5}{z}\}}{\{y,z\}}, 
  \impl{\{z\}}{\{y,z\}},
  \impl{\{\}}{\{\}}\}.
\end{align*}
According to Theorem \ref{thm:completer}, $\Fu{T^\times}$ is complete in 
$\langle X,Y,I\rangle$.
Now, $\Fu{T^\times}$ is redundant. 
First, $\Fu{T^\times}$ contains a trivial
implication $\impl{\{\}}{\{\}}$ which holds true in each $M \in \mathbf{L}^Y$. 
Furthermore, 
$\Fu{T^\times} - \{\impl{\{\}}{\{\}}\}$ is still redundant, because
implications $\impl{\{y\}}{\{y,z\}}$  and $\impl{\{z\}}{\{y,z\}}$ semantically follow
from $$S=\{\impl{\{\deg{0.5}{y}\}}{\{y,z\}},\impl{\{\deg{0.5}{z}\}}{\{y,z\}}\},$$ i.e.
$||\impl{\{y\}}{\{y,z\}}||_{S}=1$ and $||\impl{\{z\}}{\{y,z\}}||_{S}=1$.
\end{example}

Example \ref{ex:Ttimes} also shows that the converse claim to that of 
Theorem \ref{thm:completer} does not hold.
That is, it is not true  that if a set $T$ of graded implications is complete in $\tu{X,Y,I}$
then $\Cr{T}=\{ \Cr{A}\Rightarrow\Cr{B} \mid A\Rightarrow B\in T \}$
is complete in $\langle X^\times,Y^\times,I^\times\rangle$.
Namely, if this were true then for the set $S$ from Example \ref{ex:Ttimes},
which is complete in $\tu{X,Y,I}$, the set
\[
    \Cr{S} = \{ \impl{\{y_0,y_{0.5},z_0\}}{\{y_0,y_{0.5},y_1,z_0,z_{0.5},z_1\}}, 
                      \impl{\{y_0,z_0,z_{0.5}\}}{\{y_0,y_{0.5},y_1,z_0,z_{0.5},z_1\}}
                   \}
\]
would be complete in $\langle X^\times,Y^\times,I^\times\rangle$
which it is not, because $\Cr{S}$ is a proper subset of a base of
$\langle X^\times,Y^\times,I^\times\rangle$, namely of $T^\times$.

The two observations, namely that $\Fu{T^\times}$ may be redundant
even when $T^\times$ is not,
and that $\Cr{T}$ need not be complete even when $T$ is, have the following explanation.
The dependencies reflecting the algebraic structure $\mathbf{L}$ of the set of grades
 are implicitly taken into account in the
definition of entailment of graded implications over $Y$, i.e. in the semantics using
$\mathbf{L}$ as the structure of truth degrees, and need not be present in $T$.
Their counterparts, however, are ``not known'' to the 
definition of (bivalent) semantic entailment of ordinary implications 
over $Y\times L$, and need thus be explicitly present in $T^\times$. 

\section{Relationship to Functional Dependencies over Domains with Similarities}
\label{sec:rfd}

In this section, we point out a connection between graded attribute implications
and  certain extensions of Codd's relational model of data.  
Recall that in the ordinary case, which corresponds to $L=\{0,1\}$ in our setting,
the following connection
was presented in \cite{Fag:Fdrdpl}.
Ordinary attribute implications have two basic interpretations,
namely, as propositional logic formulas and as functional dependencies.
An attribute implication, say $\{y_1,y_2,y_3\}\Rightarrow\{z_1,z_2\}$,
may be conceived as a logic formula $y_1\& y_2\&y_3 \supset z_1\& z_2$
in which $y_i$s and $z_j$s are propositional symbols. 
The semantics in this case is the standard propositional logic semantics based
on truth valuations, i.e. assignments of $0$ and $1$ to propositional symbols.
This semantics leads to one notion of entailment of attribute implications,
the standard propositional logic entailment. 
This semantics is relevant to our paper because the truth valuations involved
may be identified with rows of tables with yes-or-no attributes (table entry $I(x,y)$ 
equals $1$ if{}f $y$ is assigned $1$). 
As a consequence,
the propositional logic semantics essentially coincides with the semantics based
on tables with yes-or-no attributes. In particular, these two semantics
have the same entailment relation which we denote by $\models^{\mathrm{AI}}$. 
The other sematnics of attribute implications comes from relational databases
and is given by interpreting attribute implications $A\Rightarrow B$ as
functional dependencies in relations \cite{Arm:Dsdbr,Mai:TRD}.
We thus have two notions of entailment:
first, $A\Rightarrow B$ may follow from a set $T$ of implications as a propositional logic
formula, $T\models^{\mathrm{AI}} A\Rightarrow B$;
and second, $A\Rightarrow B$ may follow from $T$ as a functional 
dependence, $T\models^{\mathrm{FD}}  A\Rightarrow B$.
 Fagin [1977] 
proved that 
\begin{equation}
  \label{eqn:Fag}
  T\models^{\mathrm{AI}} A\Rightarrow B
    \text{ is equivalent to } 
  T\models^{\mathrm{FD}}  A\Rightarrow B.
\end{equation}

Since the semantics based on tables with yes-or-no attributes, and hence the one based 
on propositional logic, is a particular case of the semantics based on tables with graded attributes
developed in this paper, the following question arises:
is there a natural extension of
Codd's relational model of data 
and the notion of functional
dependence  in this extension for which a result analogous to (\ref{eqn:Fag}) holds?
As we show below, the answer is positive. 
Formally, such an extension consists in replacing ordinary relations in Codd's model by $\mathbf{L}$-valued relations.
In particular, the domains in the extended model are be equipped with $\mathbf{L}$-valued relations,
such as similarity relations, replacing the ordinary equality relations, which are implicitly
present in Codd's model and which are utilized e.g. in selection and other queries involving match
of tuples.
Furthermore, relations on relation schemes are replaced by $\mathbf{L}$-valued relations,
which means that a degree in $L$ are assigned to each tuple. Such degree is generally interpreted as
a degree to which the tuple matches a query involving the $\mathbf{L}$-valued relations on domains.
Therefore, the $\mathbf{L}$-valued relations have in fact the same meaning in the extended model
as relations on relational schemes have in the ordinary Codd's model, namely they are understood
as results of queries with the provision that base relations considered as results of empty queries.
The above described extension is interesting in its own right because, as the thorough examination
in  \cite{BeVy:TODS} reveals, when the $\mathbf{L}$-valued relations on domains represent similarities,
the extension plays the same role for relational databases that support similarity queries
as  Codd's model plays for ordinary relational databases.

For brevity, we restrict to a particular case of the above-mentioned extension of Codd's model,
which is sufficient for our purpose. 
Let us assume that for each attribute $y$ of relation scheme (attribute set) $Y$, 
$D_y$ denotes the domain of $y$ and that each domain $D_y$ is equipped
an $\mathbf{L}$-relation $R_y$.
That is, $R_y$ maps the pairs  $\tu{d_1,d_2}\in D_y\times D_y$ to 
grades $R_y(d,d_2)\in L$,  interpreted as grades to which $d_1$ is related
to $d_2$.
A \emph{data table} \emph{over domains with $\mathbf{L}$-relations} on $Y$ we mean a
finite relation $\mathcal{D}$ between the domains $D_y$, i.e. 
$\mathcal{D}\subseteq\prod_{y\in Y} D_y$.

\begin{remark}
(a)
If we require that for every $d_1,d_2\in D_y$,
\begin{eqnarray}
  \nonumber
  \mathrm{(Ref)} &&  R_y(d_1,d_1) = 1,\\
  \nonumber
  \mathrm{(Sym)} &&  R_y(d_1,d_2) = R_y(d_2,d_1),
\end{eqnarray}
$R_y$ may naturally be understood as representing similarity, i.e.
$R_y(d_1,d_2)$ may be interpreted as a  degree to which
$d_1$ and $d_2$ are similar. Furthermore, we may assume that instead of being
an ordinary relation,  $\mathcal{D}$ is an  $\mathbf{L}$-relation,
in which case $\mathcal{D}(t)$ is naturally understood as a degree to which 
the tuple $t$ satisfies a similarity query that involves a similarity query.
For instance, assume that the query reads ``show tuples with value of attribute $\mathrm{age}$ 
similar to 30''.
Then if $t(\mathrm{age})=33$ and if $R_\mathrm{age}(30,33)=0.9$, then
the result of such query (applied to a base relation) is naturally represented
by a table $\mathcal{D}$ in which $\mathcal{D}(t)=0.9$.
This is basically the idea of the model presented in \cite{BeVy:TODS}.

Notice that if $\mathbf{L}$ is the two-element Boolean algebra, i.e. $L=\{0,1\}$
and if every $R_y$ represents equality in that $R_y(d_1,d_2)=1$ if{}f $d_1=d_2$,
the above concept may be identified with the ordinary
concept of relation on $Y$ of Codd's model \cite{Cod:Rmdlsdb,Mai:TRD}.
From this point of view, while the ordinary model supports queries regarding exact
match of domain values, the similarity-based extension supports those regarding approximate
matches.

(b)
If $R_y$ is reflexive and $\mathbf{L}$-transitive \cite{Bel:FRS,Got:TMVL}, i.e. satisfies (Ref) and
\begin{eqnarray}
\nonumber
  \mathrm{(Tra)} &&  R_y(d_1,d_2)\otimes R_y(d_2,d_3) \leq  R_y(d_1,d_3),
\end{eqnarray}
then $R_y$ is naturaly interpreted as a graded preference relation \cite{Ric:Sfp}.
\end{remark}

Ordinary attribute implications, when interpreted in data tables of Codd's model,
represent functional dependencies in this model.
In basically the same way, graded attribute implicatons may be interpreted in 
 data tables over domains with $\mathbf{L}$-relations and represent a similar type of dependencies.
Namely, an ordinary attribute implication $A\Rightarrow B$ asserts that
the same values on attributes in $A$ imply the same values on attributes in $B$.
As we show below,
when the $\mathbf{L}$-relations represent similarities,
a graded attribute implication
asserts that similar values on attributes in $A$ imply similar values on attributes in $B$.

The interpretation of graded attribute implications being introduced follows the
basic principles of predicate fuzzy logic \cite{Got:TMVL,Haj:MFL}.
Our aim is to define a degree to which a graded implication $A\Rightarrow B$
is true in a table $\mathcal{D}$.
First let us define for any two tuples $t_1,t_2\in\prod_{y\in Y} D_y$,
\def\rell{\sim}
\begin{equation}\label{eqn:xC}
  t_1(A)\rell t_2(A) \ =  \  
\textstyle
   \bigwedge_{y\in Y}
   \bigl(A(y)\rightarrow R_y(t_1(y), t_2(y))\bigr).
\end{equation}
Note that $t_1(A)\rell t_2(A)$ is the truth degree of the proposition
``for every attribute $y$ in $A$, the values
  $t_1(y)$ and $t_2(y)$ are $R_y$-related'' (instead of ``$R_y$-related'' one may use ``similar''
  here and below
   to obtain the meaning of the particular case with similarity relations).
The same way we define $ t_1(B)\rell t_2(B)$.
The degree $||A\Rightarrow B||_{{\cal D}}$ to which $A\Rightarrow B$ is true 
in $\mathcal{D}$ is defined by
\begin{equation}\label{eqn:ABD}
\textstyle
  ||A\Rightarrow B||_{{\cal D}} = 
    \bigwedge_{t_1,t_2\in \mathcal{D}}\bigl((t_1(A)\rell t_2(A))^\ast \rightarrow
    (t_1(B)\rell t_2(B))\bigr).
\end{equation}
According to the principles of fuzzy logic, 
$||A\Rightarrow B||_{{\cal D}}$ is the truth degree of the proposition
``for every two tuples $t_1,t_2\in X$: if it is (very) true that
  $t_1$ and $t_2$ have $R_y$-related (e.g. similar) values on attributes from $A$ then $t_1$ and
  $t_2$ have $R_y$-related (similar) values on attributes from $B$''.

\begin{remark}
  (a)
  One may easily observe that if $L=\{0,1\}$ and if every $R_y$ represents
identity, (\ref{eqn:ABD}) becomes the definition of validity of ordinary functional dependencies
in ordinary relations.
Furthermore, if every $R_y$ is reflexive and transitive, and thus represents a preference,
we obtain the definition of validity of ordinal dependencies \cite{GaWi:FCA}.

  (b)
  The hedge $\ast$ in (\ref{eqn:ABD}) has a similar role as 
  in  (\ref{eqn:valDeg}). 
  In particular, if $\ast$ is the globalization, see (\ref{eqn:glob}),
  then if $R_y$s represent similarities, 
   an implication such as   
   $\{{}^{a_1\!\!}/y_1,\dots,{}^{a_p\!\!}/y_p\} \!\Rightarrow\!
  \{{}^{b_1\!\!}/z_1,\dots,{}^{b_q\!\!}/z_q\}$,
is fully true, i.e. true to degree $1$, in $\mathcal{D}$ if{}f 
 similarity to degrees $a_i$ or higher on attributes $y_i$
implies similarity to degrees $b_i$ or higher on attributes $z_i$, 
as mentioned in Section \ref{sec:i}.
   For more information we refer again to \cite{BeVy:TODS}.

  (c) 
  In the literature, several approaches to a relational model over domains with similarities
  and the corresponding functional dependencies have been proposed, \cite{RaMa:Ffdljdfrds} 
  being among the first ones.
  As a rule, these approaches lack a clear connection to an underlying logic calculus
  such as the predicate logic in case of the ordinary Codd's model or
  predicate fuzzy logic as in our case.
  For an overview and comparison of these approaches, we refer the reader 
  to \cite{BeVy:Crmpvfl}.

\end{remark}


In the rest of this section, we denote by $||A\Rightarrow B||^{\mathrm{AI}}_T$
the degree to which the graded attribute implication $A\Rightarrow B$ follows 
from a fuzzy set $T$ of graded implications in the semantics given by
tables with graded attributes, as defined by (\ref{eqn:entAI}).
In much the same way, we  define the degree of entailment
$||A\Rightarrow B||^{\mathrm{FD}}_T$ in which implications are conceived as
functional dependencies in data tables over domains with $\mathbf{L}$-relations:
\begin{eqnarray}
  \label{eqn:entFD}
  ||A \Rightarrow B||^{\mathrm{FD}}_T =
  \textstyle\bigwedge_{M \in \mathrm{Mod}^{\mathrm{FD}}(T)}||A \Rightarrow B||_M
\end{eqnarray}
where 
\[
   \mathrm{Mod}^{\mathrm{FD}}(T)=\{\mathcal{D} \,|\, \mbox{for each }
   A,B\in\mathbf{L}^Y: T(A\Rightarrow B)\leq ||A \Rightarrow B||_\mathcal{D}\}.
\]
denotes the set of models of $T$, i.e. data tables in which 
each $A\Rightarrow B$ holds to a degree larger than or equal to the degree
prescribed by the theory $T$.

To answer the question about the relationship between the two concepts of entailment,
we need the next two lemmata.
\def\aaa{{(')}}
\def\bbb{^{(}{}'{}^{)}}
Let us define for a given $\langle X,Y,I\rangle$ a data table
$\mathcal{D}_{\tu{X,Y,I}}$ as follows:
\bgroup%
\addtolength{\leftmargini}{1.2ex}%
\begin{itemize}
  \item
  for each $y\in Y$,  let $D_y=X\cup X'$ where $X'=\{x'\,|\, x\in X\}$ 
  (i.e., $X\cap X'=\emptyset$ and $|X|=|X'|$);
  
  \item
  for $x_1,x_2\in D_y$, let
  \[
    R_y(x_1, x_2) = 
    \left\{
     \begin{array}{ll}
       1                       & \mbox{ for } x_1=x_2,\\
       I(z_1,y)\wedge I(z_2,y) & \mbox{ for } x_1\not=x_2, x_i=z_i{\bbb}
                                 \mbox{ for } z_i\in X\ (i=1,2),
     \end{array}
    \right.
  \]
  where $x_i=z_i{\bbb}$ means that $x_i$ is $z_i$ or $z_i'$;

  \item
   $\mathcal{D}=\{t_x \mid x\in X\cup X'\}$ where 
   $t_x$ is the tuple in $\prod_{y\in Y} D_y$ for which $t_x(y)=x$ for every $y\in Y$.
\end{itemize}
\egroup%
As the following lemma shows, degrees of validity in $\tu{X,Y,I}$ coincide
with those in ${\cal D}_{\tu{X,Y,I}}$.

\begin{lemma}\label{thm:TtoD}
  For every data table $\tu{X,Y,I}$ with graded attributes
  and any graded attribute implication $A\Rightarrow B$,
  \begin{equation}\label{eqn:TtoD}
   ||A\Rightarrow B||_{\tu{X,Y,I}} = ||A\Rightarrow B||_{{\cal D}_{\tu{X,Y,I}}}.
  \end{equation}
\end{lemma}
\begin{proof}
Let us first observe that 
  \begin{itemize}
    \item[(a)]
    $(a_1^\ast\rightarrow b_1)\wedge(a_2^\ast\rightarrow b_2)\leq 
    (a_1\wedge a_2)^\ast\rightarrow(b_1\wedge b_2)$ for any $a_1,a_2,b_1,b_2\in L$;
    
    \item[(b)]
    $t_{x_1}(C)\rell t_{x_2}(C)=S(C,I_{z_1})\wedge S(C,I_{z_2})$

    for any $C\in\mathbf{L}^Y$ and any $x_1\not=x_2$ such that
    $x_1=z_1{\bbb},x_2=z_2{\bbb}$ for some $z_1,z_2\in X$;
    
    \item[(c)]
    $t_x(C)\rell t_x(C)=1$ for $x\in X\cup X'$.
  \end{itemize}
Indeed, due to adjointness, (a) is equivalent to 
$ (a_1\wedge a_2)^\ast \otimes ((a_1^\ast\rightarrow b_1)\wedge(a_2^\ast\rightarrow b_2)) 
   \leq b_1\wedge b_2$ which holds if{}f
$ (a_1\wedge a_2)^\ast \otimes ((a_1^\ast\rightarrow b_1)\wedge(a_2^\ast\rightarrow b_2)) 
   \leq b_1$ and $\leq b_2$.
Both inequalities are true. Namely, 
$
   (a_1\wedge a_2)^\ast \otimes ((a_1^\ast\rightarrow b_1)\wedge(a_2^\ast\rightarrow b_2)) 
   \leq a_1^\ast \otimes (a_1^\ast\rightarrow b_1) 
   \leq b_1
$
and similarly for $b_2$.

Since $I_{z_i}(y)=I(z_i,y)$, we have 
\begin{eqnarray*}
 &&\textstyle t_{x_1}(C)\rell t_{x_2}(C)=
      \bigwedge_{y\in Y}
    \bigl( C(y)\rightarrow R_y(t_{x_1}(y), t_{x_2}(y))\bigr) = \\
 &=&\textstyle \bigwedge_{y\in Y}  (C(y)\rightarrow R_y(x_1, x_2)) =
\bigwedge_{y\in Y}  (C(y)\rightarrow ( I(z_1,y)\wedge I(z_2,y))) =\\
 &=&\textstyle \bigwedge_{y\in Y} (C(y)\rightarrow I(z_1,y))
\wedge
\bigwedge_{y\in Y} (C(y)\rightarrow I(z_2,y)) =
 S(C,I_{z_1})\wedge S(C,I_{z_2}),
\end{eqnarray*}
establishing (b). (c) is evident.

Let for brevity $\mathcal{D}=\mathcal{D}_{\tu{X,Y,I}}$. We obtain
  \begin{eqnarray*}
   && ||A\Rightarrow B||_{{\cal D}}=
  \textstyle
    \bigwedge_{t_1,t_2\in\mathcal{D}}
    \bigl((t_1(A)\rell t_2(A))^\ast \rightarrow (t_1(B)\rell t_2(B))\bigr) =\\
 &=&
     \textstyle
    \bigwedge_{x_1,x_2\in X\cup X'}
    \bigl((t_{x_1}(A)\rell t_{x_2}(A))^\ast \rightarrow (t_{x_1}(B)\rell t_{x_2}(B))\bigr) 
    =\alpha\wedge\beta\wedge\gamma,
%
  \end{eqnarray*}
  where
  \begin{eqnarray*}
  && 
  \alpha =\textstyle
  \bigwedge_{x_1,x_2\in X\cup X',x_1=x_2}
   \bigl((t_{x_1}(A)\rell t_{x_2}(A))^\ast \to (t_{x_1}(B)\rell t_{x_2}(B))\bigr) =1
  \end{eqnarray*}
 on account of (c),
  %
  \begin{eqnarray*}
  &&\textstyle
  \beta=\bigwedge_{x_1,x_2\in X\cup X',\{x_1,x_2\}=\{z,z'\}}
  \bigl((t_{x_1}(A)\rell t_{x_2}(A))^\ast \to (t_{x_1}(B)\rell t_{x_2}(B))\bigr) =\\
  &=&\textstyle
  \bigwedge_{z\in X} \bigl((S(A,I_z)\wedge S(A,I_z))^\ast \rightarrow
  (S(B,I_z)\wedge S(B,I_z))\bigr)=\\
  &=&\textstyle
  \bigwedge_{z\in X} \bigl(S(A,I_z)^\ast \rightarrow S(B,I_z)\bigr) =
  ||A\Rightarrow B||_{\tu{X,Y,I}}
  \end{eqnarray*}
on account of (b),
and  
  \begin{eqnarray*}
  &&
  \gamma=\bigwedge_{\{x_1,x_2\}=\{z_1{\bbb},z_2{\bbb}\},
  z_1\not=z_2} \!\!\!\!\!\!\!\!\!\!
  \bigl((t_{x_1}(A)\rell t_{x_2}(A))^\ast \to (t_{x_1}(B)\rell t_{x_2}(B))\bigr)=\\
  &=&
  \bigwedge_{\{x_1,x_2\}=\{z_1{\bbb},z_2{\bbb}\},
  z_1\not=z_2} \!\!\!\!\!\!\!\!\!\!
  \bigl((S(A,I_{z_1})\wedge S(A,I_{z_2}))^\ast \rightarrow
  (S(B,I_{z_1})\wedge S(B,I_{z_2}))\bigr)\geq\\
  &=&
  \bigwedge_{\{x_1,x_2\}=\{z_1{\bbb},z_2{\bbb}\},
  z_1\not=z_2} \!\!\!\!\!\!\!\!\!\!
  \bigl([(S(A,I_{z_1})^\ast \rightarrow S(B,I_{z_1}))
  \wedge
  (S(A,I_{z_2})^\ast \rightarrow S(B,I_{z_2}))]\bigr)=\\
  &=&\textstyle \bigwedge_{z\in X} \bigl(S(A,I_z)^\ast \rightarrow S(B,I_z)\bigr) = \beta
  \end{eqnarray*}
on account of (b) and (a). 
Therefore,
\[
  ||A\Rightarrow B||_{{\cal D}_{\tu{X,Y,I}}}= \beta = 
  ||A\Rightarrow B||_{\tu{X,Y,I}},  
\]
completing the proof.
\end{proof}

Conversely, for a given table over domains with $\mathbf{L}$-relations $\mathcal{D}$, 
define a table $\tu{X,Y,I}_\mathcal{D}$ as follows:
\bgroup%
\addtolength{\leftmargini}{1.2ex}%
\begin{itemize}
  \item
  $X=\mathcal{D}\times\mathcal{D}$;
  
  \item
   for $\langle t_1,t_2\rangle\in X$ and $y\in Y$, let
  $I(\langle t_1,t_2\rangle,y)=
       R_y(t_1(y), t_2(y))$.
 \end{itemize}
\egroup%
As in the previous case, $\mathcal{D}$ and $\tu{X,Y,I}_\mathcal{D}$
yield the same truth degrees of attribute implications:

\begin{lemma}\label{thm:DtoT}
 For every ranked data table  $\cal D$
 and any graded attribute implication $A\Rightarrow B$,
  \begin{equation}\label{eqn:DtoT}
   ||A\Rightarrow B||_{\cal D} = ||A\Rightarrow B||_{\tu{X,Y,I}_\mathcal{D}}.
  \end{equation}
\end{lemma}
\begin{proof}
 Notice first that 
\begin{eqnarray*}
 && \textstyle (t_{1}(A)\rell t_{2}(A)) =
   \bigwedge_{y\in Y}(A(y)\rightarrow R_y(t_{1}(y), t_{2}(y))) =\\
  &=& \textstyle  \bigwedge_{y\in Y}  
       (A(y) \to I(\tu{t_{1},t_{2}},y)) = S(A,I_{\tu{t_{1},t_{2}}}), 
\end{eqnarray*}
and the same for $B$.
We therefore get
  \begin{eqnarray*}
    &&
    ||A\Rightarrow B||_{\cal D} = 
    \textstyle\bigwedge_{t_1,t_2\in \mathcal{D}}
     \bigl([t_{1}(A)\approx t_{2}(A)]^\ast
     \rightarrow
     [t_{1}(B)\approx t_{2}(B)]\bigr) =\\
    &=&
     S(A,I_{\tu{t_{1},t_{2}}})^\ast \to S(B,I_{\tu{t_{1},t_{2}}}) =
    ||A\Rightarrow B||_{{\tu{X,Y,I}}_{\cal D}}.
  \end{eqnarray*}
\end{proof}

The following theorem answers the question from the beginning of this section. 

\begin{theorem}\label{thm:semFDAI}
For every fuzzy set $T$ of graded attribute implications and every graded attribute implication
 $A\Rightarrow B$ we have
\begin{eqnarray}\label{thm:semAIFD}
  ||A \Rightarrow B||^{\mathrm{FD}}_T = ||A \Rightarrow B||^{\mathrm{AI}}_T.
\end{eqnarray}
\end{theorem}
\begin{proof}
We need to prove
$||A \Rightarrow B||^{\mathrm{FD}}_T \leq ||A \Rightarrow B||^{\mathrm{AI}}_T$
and
$||A \Rightarrow B||^{\mathrm{FD}}_T \geq ||A \Rightarrow B||^{\mathrm{AI}}_T$.
To check the first inequality, it is enough to show that for each
$M\in\mathrm{Mod}(T)$ there is
${\cal D}\in\mathrm{Mod}^{\mathrm{FD}}(T)$ such that
$||A\Rightarrow B||_M=||A\Rightarrow B||_{\cal D}$. 
This follows directly from Lemma~\ref{thm:TtoD} by taking
 ${\cal D}={\cal D}_{\tu{X,Y,I}}$, where $\tu{X,Y,I}$ is a one-row data table 
corresponding to $M$, i.e. with $X=\{x\}$ and $I(x,y)=M(y)$ for each $y\in Y$.
Namely, we then have $||A\Rightarrow B||_M=||A\Rightarrow B||_{\tu{X,Y,I}}=
||A\Rightarrow B||_{\cal D}$. 
The second inequality is proved in a similar manner using 
 Lemma~\ref{thm:DtoT}.
\end{proof}

\begin{remark}
  The tables $\tu{X,Y,I}_\mathcal{D}$ and $\mathcal{D}_{\tu{X,Y,I}}$
  constructed from $\mathcal{D}$ and $\tu{X,Y,I}$, respectively, are not 
  minimal in size. We use them because their definitions are relatively simple
  and they do their work in the proofs of Lemma~\ref{thm:TtoD} and
  Lemma~\ref{thm:DtoT}.
%
\end{remark}

\section{Conclusions}
\label{sec:c}

We presented an approach to attribute dependencies for data with grades,
such as a grade to which an object is red or a grade to which two objects are similar.
 Such dependencies extend classical dependencies in Boolean data
and classical functional dependencies.
We presented results regarding major issues traditionally investigated for such dependencies,
including entailment, redundancy and bases of dependencies, associated closure structures,
 Armstrong-like axiomatization, and computation issues. In addition, we examined
a relationship between the new kind of dependencies and the classical ones
and showed that the well-known correspondence between attribute dependencies
in Boolean data on one hand and functional dependencies in relational model of data on the other hand 
is retained in the setting with grades but obtains a nontrivial, interesting form. Namely,
in the setting with grades, the role of functional dependencies is played by
their analogue in an extended relational model in which every domain is equipped with a
similarity relation, or a more general binary relation, assigning grades of similarity to 
pairs of domain elements.

In addition, the paper attempts to make a methodological point, the ramifications of which
we consider equally important as the results mentioned in the above paragraph.
The point is the following. Classical dependencies are based in classical logic in that 
the truth value \emph{true} (1) represents presence of an attribute and
match of attribute values, while \emph{false} (0) represents absence and mismatch.
Moreover, the truth values are manipulated by classical logic
connectives and further notions such as validity and entailment of dependencies are derived from
classical logic notions. Broadly speaking, classical dependencies are founded in the agenda of 
classical logic.
The presence of possibly many grades in the new situation and the
ordinal nature of  grades makes the situation challenging
and prone to ad hoc treatments, involving for instance metrics representing
similarities.  Thus, one might attempt to retain the agenda of
classical logic, extend the formalism of classical dependencies by 
a metric (distance function) to represent similarity, and arrive at a 
blend of a logic-based formalism and a metric-based one.
Instead, our approach---like the classical one---is purely logically based,
yet capable of handling grades and their semantics in a reasonable way.
We consider the grades as truth values in the sense of fuzzy logic,
i.e. consider them as truth degrees with $1$ and $0$
representing the boundary cases and the other ones, such as $0.8$, as representing intermediary
cases. 
In a sense, we move to a more general framework, a logical calculus in which statements
 such as
``attribute $y$ applies to object $x$'' and ``objects $x_1$ and $x_2$ are similar (equal)''
are no longer considered bivalent. Rather, these statements are allowed to be assigned,
in addition to $0$ and $1$, an intermediary \emph{grade}, i.e.
a \emph{truth degree} between $0$ and $1$.
Such move can effectively be realized.
Namely,  we argue that data involving grades and reasoning about such data can be 
modeled utilizing a framework of mathematical fuzzy logic, a recently developed
many-valued logic with now well-developed agenda and that this logic may assume the role 
classical logic plays in the established theories of data dependencies and reasoning about data
in general. 
The main advantage of this approach is conceptual clarity. 
On the level of syntax, the key notions in the model with grades have essentially the same form 
as in the classical, bivalent case.
This means that the informal description of the key notions in natural language,
and hence the intuitive meaning of the key notions, remains essentially the same as in the classical model.
Yet, on the level of semantics, grades obtain a proper treatment and permeat the subsequent
notions such as validity or entailment in a natural way. Thus, for instance, validity or entailment
of dependencies are no longer bivalent concepts. Rather, they naturally emerge as graded notions. 
One obtains a degree of validity or degree of entailment of dependencies, corresponding 
to the idea that a compound statement (such as a depedendency claim)
involving partially true constituent statements (such as ``attribute $y$
applies to object $x$'') may itself be only partially true, i.e. true to an intermediary degree.

To sum up, utilization of mathematical fuzzy logic as a formal framework for modeling
data with grades brings conceptual clarity and  makes possible a treatment of  
attribute dependencies essentially the same way as utilization of classical logic does
for data with no intermediary grades.
Clearly, the presented approach is not restricted to the problems dealt with in the present paper.
In this respect, our paper demonstrates that fuzzy logic is a convenient framework
for modeling certain problems that surpass the domain traditionally accounted for by 
classical logic, namely those that may be characterized by a graded nature of data
and reasoning about such data.
Such problems abound particularly in situations where human judgment is involved,
for which the usage of graded, ``fuzzy'' notions, such as \emph{red}, \emph{tall}, \emph{similar},
rather than bivalent ones, is  characteristic.
A further development of theories and methods inspired by such problems presents a challenging
and important research goal.
The associated research agenda includes several complex issues, some of which we intentionally disregarded
in the present paper. One such issue is connected to the fact that the theory we present is not restricted to
a particular set of grades and particular (truth functions of) logical connectives on this set.
Rather, we proceed in a general way and only assume that the set $L$ of grades and the logical connectives
on $L$ satify certain logically reasonable conditions such as the isotony of conjunction, its commutativity,
associativity, and the like. In a sense, the presented theory is qualitative and open to determination
of a quantitative component. Clearly, the choice of this component,
i.e. a particular set $L$ and particular connectives on $L$, is a step one needs to make when applying the theory.
One option in making this step is to proceed on intuitive grounds, which is often the case in applications
of fuzzy logic. In fact, there is an argument for considering such option sufficient for practical purpose, namely, that the
common qualitative properties of all the potential sets of logical connectives are specific enough to the extent that all the
sets of connectives can be considered reasonable for practical purpose. Still, such option may arguably be regarded
as too much ad hoc. In fact, the choice of a  scale of grades and logical connectives for this scale is a matter
that calls for a thorough examination from the point of view of a mathematical and cognitive psychology.
In our view, such examination presents challenging problems with broad ramifications  and is very much needed.

\begin{acks}
Dedicated to Professor Petr H\'ajek.
\end{acks}

\bibliographystyle{acmsmall}
\bibliography{acmsmall-sam}


\received{XXX}{XXX}{XXX}

\end{document}